\PassOptionsToPackage{table}{xcolor}
\documentclass[%
 reprint,
superscriptaddress,
footinbib,
 amsmath,amssymb,
 aps,
pra,
]{revtex4-2}

\usepackage{tikz} 
\usepackage{amsmath,amsfonts,amsthm, mathtools}
\usepackage{float}
\usepackage{dsfont}
\usepackage{csquotes}
\usepackage{amssymb}
\usepackage[table]{xcolor}
\usepackage{verbatim}
\usepackage{mathtools}
\usepackage{rotating}
\usepackage{thmtools, thm-restate}
\usepackage{bm}
\usepackage{hyperref}
\usepackage{mathtools}
\usepackage{mathbbol}
\usepackage{makecell}
\usepackage{pifont}
\usepackage{appendix}
\usepackage{diagbox}      
\usepackage{soul}
\newcommand{\cmark}{\ding{51}}
\newcommand{\xmark}{\ding{55}}
\definecolor{darkred}  {rgb}{0.5,0,0}
\definecolor{darkblue} {rgb}{0,0,0.5}
\definecolor{darkgreen}{rgb}{0,0.5,0}
\hypersetup{
  colorlinks = true,
  urlcolor  = blue,         
  linkcolor = darkblue,     
  citecolor = darkgreen,    
  filecolor = darkred    
}
\theoremstyle{definition}

\newtheorem{definition}{Definition}

\newtheorem{proposition}{Proposition}
\newtheorem{theorem}{Theorem}

\newtheorem{example}{Example}

\newtheorem*{Thmprop}{Theorem 1 and Proposition 1}

\newcommand{\mbf}{\mathbf}
\newcommand{\mbb}{\mathbb}
\newcommand{\mc}{\mathcal}
\newcommand{\msf}{\mathsf}

\newcommand{\tr}{\textrm{Tr}}

\usepackage{multirow}
\newcommand{\ket}[1]{|#1\rangle}
\newcommand{\bra}[1]{\langle #1|}
\newcommand{\op}[2]{|#1\rangle\langle#2|}

\definecolor{cool_green}{rgb}{0.0, 0.5, 0.0}
\newcommand{\cmt}{\color{black}}
\newcommand{\yujie}{\color{black}}
\newcommand{\yz}{\color{black}}

\newcommand{\blk}{\color{black}}

\usepackage{blkarray}
\usepackage{adjustbox}
\usepackage{graphicx}
\usepackage{dcolumn}
\usepackage{bm}
\begin{document}

\preprint{APS/123-QED}

%\title{All entangled states can be certified by a violation of noncontextuality inequalities }
\title{\yujie  Entanglement Certification Using Noncontextuality Inequalities\blk}

\author{Yujie Zhang}
\email{yujie4physics@gmail.com}
\affiliation{Institute for Quantum Computing and Department of Physics \& Astronomy,
University of Waterloo, 200 University Ave W, Waterloo, Ontario, N2L 3G1, Canada}
\affiliation{Perimeter Institute for Theoretical Physics, 31 Caroline Street North, Waterloo, Ontario Canada N2L 2Y5}
\author{Jonah Spodek}
\affiliation{Institute for Quantum Computing and Department of Physics \& Astronomy,
University of Waterloo, 200 University Ave W, Waterloo, Ontario, N2L 3G1, Canada}
\author{David Schmid}
\affiliation{Perimeter Institute for Theoretical Physics, 31 Caroline Street North, Waterloo, Ontario Canada N2L 2Y5}
\author{Carter Reid}
\affiliation{Institute for Quantum Computing and Department of Physics \& Astronomy,
University of Waterloo, 200 University Ave W, Waterloo, Ontario, N2L 3G1, Canada}
\author{Liam J. Morrison}
\affiliation{Institute for Quantum Computing and Department of Physics \& Astronomy,
University of Waterloo, 200 University Ave W, Waterloo, Ontario, N2L 3G1, Canada}
\affiliation{Department of Physics and Atmospheric Science, Dalhousie University, Halifax, Nova Scotia, Canada B3H 3J5}
\author{Thomas Jennewein}
\affiliation{Institute for Quantum Computing and Department of Physics \& Astronomy,
University of Waterloo, 200 University Ave W, Waterloo, Ontario, N2L 3G1, Canada}
\affiliation{Department of Physics, Simon Fraser University, 8888 University Dr W, Burnaby, BC V5A 1S6, Canada}
\author{Kevin J. Resch}
\affiliation{Institute for Quantum Computing and Department of Physics \& Astronomy,
University of Waterloo, 200 University Ave W, Waterloo, Ontario, N2L 3G1, Canada}
\author{Robert W. Spekkens} 
\affiliation{Perimeter Institute for Theoretical Physics, 31 Caroline Street North, Waterloo, Ontario Canada N2L 2Y5}

\date{\today}

\begin{abstract}
\cmt 
By combining the assumptions of Bell locality with those of generalized noncontextuality, we define classes of noncontextuality inequalities for correlations arising in a bipartite Bell circuit. These classes are distinguished by which subsets of the full set of operational identities are taken as input to the principle of noncontextuality; certain natural subsets form a hierarchy that provides a new way of understanding and classifying different forms of quantum correlations, including entanglement, steering, and nonlocality. Each level of this hierarchy gives rise to a corresponding class of noncontextuality inequalities whose violation witnesses one of these forms of bipartite quantum resourcefulness, thereby yielding different sufficient conditions for entanglement. The resulting entanglement certification paradigm requires no prior characterization of the measurements, delivers a verdict independent of tomographic gauge freedom, and can certify any entangled state without auxiliary entangled sources. To illustrate the power of this paradigm, we show that noncontextuality inequalities can certify entanglement for families of two-qubit isotropic states for which certification by Bell or steering inequalities is known to be impossible. We also show that, compared with the Bell test, this approach certifies a much larger fraction of entangled states, while the associated membership problem is computationally more tractable. On the experimental side, we describe techniques to ensure nontrivial operational identities in the presence of noisy and imperfect implementations. We also identify the sufficient condition under which these techniques are valid, namely, a particular notion of \textit{tomographic completeness} \yujie of the implemented Bell circuit\blk, which ensures that the operational identities are independent of tomographic gauge. Finally, we provide an experimental demonstration of the superior performance of this entanglement certification technique using polarization-entangled photons. \blk

\end{abstract}
\maketitle

\tableofcontents

\section{Introduction}
\label{Sec:introduction}
Distinguishing whether a given operational behavior is classically explainable or genuinely quantum is a central problem in the foundations of quantum mechanics. As is well known, certain quantum experiments exhibit correlations that cannot be explained within a classical model that respects the natural causal structure of the experiment. For example, there is no local (i.e., common-cause) explanation for quantum correlations in a Bell circuit~\cite{Bell1964, Wood_2015}. In addition to Bell nonlocality~\cite{Brunner2014}, other types of correlations that arise naturally in this scenario include entanglement~\cite{Horodecki2009} and steerable correlations~\cite{Uola2020}.

These concepts play a central role in quantum information science and have been extensively studied~\cite{Werner1989, Piani2011, Wiseman2007}, yet there is no consensus on how to best understand the relationship between them~\cite{ Schmid2023understanding,Marco2015}. While all pure entangled states exhibit nonlocality~\cite{Gisin1991, Gisin1992}, for mixed states the connection between entanglement and nonlocality is less clear~\cite{Werner1989, Barrett2002}. As a consequence, Bell inequalities~\cite{Clauser1969}, which serve as device-independent tests of nonlocality, cannot detect all entangled states in a standard Bell circuit, and alternative approaches such as quantum state tomography, entanglement witnesses, and measurement device-independent tests have been developed to verify entanglement in practical settings~\cite{Guhne2009, Branciard2013, Buscemi2012}.

We suggest, alternatively, that classical nonexplainability should be understood as the failure of operational statistics to be realizable by a noncontextual ontological model (NCOM)~\cite{Spekkens2005}, a property that we shall refer to as \textit{NCOM-nonrealizability}~\footnote{Note that throughout this work, we use the term {\em noncontextual} to refer to the notion of generalized noncontextuality defined in Ref.~\cite{Spekkens2005} rather than the Kochen-Specker notion~\cite{Kochen1975,Budroni2022}.  We use the term {\em NCOM-nonrealizable} in preference to {\em contextual} to emphasize that in the face of a no-go result, one has the option to abandon the framework of ontological models, rather than the principle of noncontextuality}. A noncontextual ontological model imposes constraints on the representation of each component of a quantum circuit based on the specific operational identities that they satisfy (where an operational identity is simply any equality holding among the different possibilities for a component 
in the circuit~\cite{Schmid2024add}).  This approach generalizes the notion of Kochen–Specker noncontextuality from projective measurements to arbitrary quantum processes, and generalizes the notion of negativity of the Wigner function to arbitrary quasiprobabilistic representations~\cite{Spekkens2008,Schmid2024}.  NCOM-nonrealizability also coincides with the failure of the \blk natural notion of classical explainability in generalized probabilistic theories~\cite{Schmid2021, Shahandeh2021,Schmid2024}.  Recent efforts have made some connections between NCOM-nonrealizability and various notions of resourcefulness of bipartite quantum states~\cite{Tavakoli2020, Wrigh2023, Plavala2024}. However, these studies have focused on comparing bipartite quantum resources to NCOM-nonrealizability in certain prepare-measure scenarios where the set of preparations is induced by the bipartite state through steering. As we will show, examining NCOM-nonrealizability within a bipartite scenario reveals more direct connections to the different classes of bipartite quantum resources.

\par 
\yujie The first part of this work develops this theoretical connection between generalized noncontextuality and bipartite quantum resources. We introduce the notion of NCOM-realizability for operational statistics within a bipartite Bell circuit, where we denote the two parties as Alice and Bob, who can perform local measurements. We demonstrate that different sets of operational identities yield, via the principle of noncontextuality, different sufficient conditions for
 establishing NCOM-nonrealizability\blk. More precisely, we explore four distinct subsets of the full set of operational identities: (i) the union of those that hold among Alice's measurements and those that hold among Bob's measurements, (ii) only those that hold among Bob's measurements, (iii) only those that hold among Alice's measurements, and (iv) the trivial set. We prove that the quantum states that define Bell circuits that are NCOM-nonrealizable relative to these classes correspond, respectively, to the entangled states, the Alice-to-Bob steerable states, the Bob-to-Alice steerable states, and the nonlocal states.

\yujie This unified framework provides tools to witness membership in these classes via distinct sets of noncontextuality (NC) inequalities. \blk The inequalities derived from considering class (i), i.e., the union of the set of all operational identities associated to each party, 
are particularly noteworthy, as they are capable of certifying the entanglement of \textit{any} entangled state, including those that are local and unsteerable. This significantly extends the reach of entanglement certification beyond the limitations of standard Bell and steering inequalities. As a concrete illustration, we derive several such inequalities and experimentally demonstrate the violation of one of these, thereby certifying the entanglement of a class of two-qubit isotropic states that are both local and unsteerable, and which consequently cannot be certified using conventional Bell or steering inequalities. 

As with Bell- and steering-inequality-based approaches, constructing a noncontextuality inequality tailored to a given target state is generally nontrivial:  it typically requires optimizing over the set of measurements, a task that quickly becomes technically and computationally demanding as the number of measurements grows. Despite this, we show that even with a fixed set of local measurements, our framework can certify a significantly broader set of entangled states than standard Bell tests can. For two-qubit entangled states drawn according to the Hilbert–Schmidt measure, we show that a noncontextuality test for a fixed set of measurements can
certify entanglement for more than $75\%$ of the sampled entangled states, whereas the CHSH test certifies only about $1\%$ of these using the same sampling procedure~\cite{Kofman2008}.

\yujie 
The second part of this work addresses the practical question of how to turn the theoretical connection between NCOM-nonrealizability and entanglement into a concrete entanglement-certification protocol. The key challenge is to obtain the operational identities which, via the principle of noncontextuality, imply constraints on the ontological model and ultimately the noncontextuality inequalities. This is important because one can be confident in inferring entanglement from the violation of a noncontextuality inequality only to the extent that one is confident in the operational identity from which the noncontextuality inequality was derived. Unlike a conventional device-dependent entanglement witness, our protocol does not assume a prior (``trusted'') characterization of the measurement devices. Instead, it relies solely on operational identities among the implemented measurement procedures. We describe two general techniques for inferring or establishing the operational identities. In the cut-to-your-cloth approach, one performs self-consistent tomography to characterize the implemented Bell circuit and to infer the linear identities that the realized procedures actually satisfy. In the secondary-procedures approach, one defines new procedures by classical randomization and post-processing of the realized procedures, so that the desired operational identities hold by construction. Unlike a device-independent Bell test, our noncontextuality-inequality test depends on the realized preparations (the steered states) and measurements being tomographically complete. \yz  This is precisely the role played by tomographic completeness in self-consistent tomography~\cite{Nielsen2021,Mazurek2021}: it is the condition under which the observed operational statistics determine the realized procedures up to an invertible tomographic gauge. In our setting, this means that although the reconstructed states and effects are not unique, the linear identities they satisfy are unique. Equivalently, all tomographic reconstructions compatible with the same data agree on the operational identities relevant to the noncontextuality test. \yujie In both the cut-to-your-cloth and secondary procedures techniques discussed below, the condition of tomographic completness ensures that the operational identities used in the test are gauge-independent facts about the implemented procedures, rather than artifacts of a particular representation. Consequently, the quantities entering the noncontextuality-inequality test, namely the operational statistics together with these identities, do not rely on a prior characterization of the measurement devices and are independent of the tomographic gauge.

This perspective also situates our method among existing entanglement-certification techniques, such as quantum state tomography,  entanglement witnesses~\cite{Guhne2009}, self-consistent tomography (i.e., gate set tomography~\cite{Nielsen2021} and GPT tomography~\cite{Mazurek2021}), Bell inequality tests~\cite{Brunner2014}, \cmt steering inequality tests~\cite{Saunders2010}, \blk and other causal compatibility tests within more complex causal networks~\cite{bowles2018device}.
As summarized in Table~\ref{tab:my_label}, our approach combines several desirable features that are not achieved simultaneously by these existing methods.

\yujie 
The remainder of the paper is organized as follows. We first separate the ideal certification problem from its practical implementation, and then spell out the conditions under which the certification method is valid. In Sec.~\ref{Sec:ncom-model}, we define NCOM-realizability for bipartite Bell circuits and introduce the hierarchy of operational identities, together with its relation to bipartite quantum resources. In Sec.~\ref{Sec:entanglement-certification}, we study the ideal certification problem: assuming that the relevant operational identities are exactly specified, we derive the corresponding noncontextuality inequalities and show how they can certify entanglement, including in regimes inaccessible to Bell and steering inequalities. In Sec.~\ref{Sec:inexact-identities}, we turn to practical implementation and explain how the required operational identities can be inferred or enforced when the implemented measurements are noisy and imperfect. In Sec.~\ref{Section:experiment}, we apply these techniques in a photonic experiment and demonstrate entanglement certification beyond steering- and Bell-inequality-based methods. In Sec.~\ref{sec:Loopholes}, we discuss the tomographic incompleteness loophole of the protocol, and the lack of a loophole associated to detector inefficiency. In Sec.~\ref{Section:comparison}, we compare our approach with other entanglement-certification methods.
\blk

\begin{table*}[t]
    \centering
  \resizebox{\textwidth}{!}{\begin{tabular}{c|c|c|c|c|c}
    \hline\hline
      \multirow[c]{2}{*}{\makecell[c]{\\Entanglement Certification}}
      & \multicolumn{5}{c}{{Features}} \\
    \cline{2-6}
      & \makecell{{No prior  characterization}\\{of measurements}}
      & \makecell{{Independent of}\\{tomographic gauge}}
      & \makecell{{Able to certify any}\\{entangled state}}
      & \makecell{{No additional}\\{entangled sources}}
      & \makecell{{No assumption of}\\{tomographic completeness$^*$}}  \\    
    \hline 
     \makecell{Separability test with\\ quantum state tomography } & \xmark & \cmark & \cmark & \cmark & \xmark \\
         \hline     
     Entanglement witness \cite{Guhne2009} & \xmark & \cmark & \cmark & \cmark  & \cmark  \\
          \hline 
   \makecell{Separability test  with\\ GPT/gate‐set tomography \cite{Mazurek2021, Nielsen2021} }
      & \cmark & \xmark & \cmark & \cmark & \xmark \\
          \hline     
    \makecell{
    Bell inequality test \\ in Bell circuit~\cite{Brunner2014, Lin2018} }
      & \cmark & \cmark & \xmark & \cmark  & \cmark \\
          \hline     
    \makecell{
   \cmt Steering inequality test \blk \\ \cmt in Bell circuit~\cite{Cavalcanti2009, Saunders2010}}
      & \xmark & \cmark & \xmark & \cmark  & \xmark \\
    \hline 
    \makecell{
    Inequality test \\ in non-Bell network scenario\cite{bowles2018device}}
      & \cmark & \cmark & \cmark & \xmark  & \cmark\\
    \hline 
    \makecell{Noncontextuality inequality test \\ (Our approach)} 
      & \cmark & \cmark & \cmark & \cmark & \xmark \\
    \hline\hline
\end{tabular}
}       
    \caption{\cmt A summary of the advantages of various methods of entanglement certification.  $^*$In quantum state tomography, tomographic completeness of the set of characterized measurements guarantees a unique reconstructed state. In self-consistent tomography, such as GPT tomography and gate-set tomography, tomographic completeness instead ensures that the operational data determine the states and measurements only up to an invertible linear gauge (i.e., up to a tomographic gauge freedom~\cite{Nielsen2021}). 
Just as tomographic completeness is a sufficient condition for implementing self-consistent tomography, in our entanglement certification technique, it is a sufficient condition for ensuring that the identities respected by the operational statistics suffice to infer the operational identities relevant to the noncontextuality test  (see Sec.~\ref{sec: tomo-loopholes}). In other words, it ensures that the operational identities are independent of the tomographic gauge. \blk}
    \label{tab:my_label}
\end{table*}
\section{Noncontextual Ontological Model}
\label{Sec:ncom-model}
\subsection{NCOM-realizability of the statistics of a prepare-measure circuit}\label{NCOMsubsec}
The notion of NCOM-realizability provides a rigorous framework for assessing the classical explainability of \cmt the statistics obtained in a prepare-measure circuit~\cite{Spekkens2005}.  A prepare-measure circuit is a circuit that consists of a process of the preparation variety (one having a single quantum output) followed by a process of the measurement variety (one having a single quantum input)~\cite{zhang2024parallel}; see Fig.~\ref{fig:ontology}(a) for a simple example. On the preparation side, the process is one that has a classical input (a setting variable), denoted $x$. Such a process is represented by a set of quantum states $\mbf{P}:=\{\rho_{x}\}_{x}$, and is termed a {\em multi-state} in Ref.~\cite{zhang2024parallel}.  On the measurement side, the process is one that has both a classical output (the measurement outcome), denoted $b$, and a classical input (a setting variable), denoted $y$. Such a process is represented by a set of POVMs $\msf{M}:=\{\{M_{b|y}\}_{b}\}_{y}$, and is termed
a {\em multi-measurement} in Ref.~\cite{zhang2024parallel}. The statistics generated by such a prepare-measure circuit have the form of a conditional probability 
\begin{align}
p(b|xy):= {\rm Tr}[M_{b|y} \rho_x].   
\label{eq: QPM}
\end{align}
Such a quantum circuit is \textit{classically explainable} if the statistics it predicts in Eq.~\eqref{eq: QPM} can be realized by a noncontextual ontological model.\footnote{Strictly speaking, the notion of noncontextuality can only be applied to ontological models of theories that are not quotiented relative to operational equivalences, whereas in this work we treat quantum theory as a theory wherein this quotienting procedure has been applied ~\cite{chiribellaprob,Schmid2024}. Refs.~\cite{Schmid2021,Schmid2024} introduce more precise terminology, but we eschew it here for the sake of maintaining familiar terminology. } This is the case if it is possible to write
\begin{equation}
p(b|xy)=\sum_{\lambda} p(b|y,\lambda)p(\lambda|{x}),
\label{eq: noncontexual}
\end{equation}
\blk where $p(\lambda|x)$ denotes a probability distribution over ontic states $\lambda \in \Lambda$ given preparation $x$, and $p(b|y,\lambda)$ is the response function specifying the probability of outcome $b$ of measurement $y$ given that the ontic state was $\lambda$. The principle of generalized noncontextuality implies that 
$p(\lambda|x)$ must be a linear function of the quantum state $\rho_x$, and $p(b|y,\lambda)$ must be a linear function of the effect $M_{b|y}$. A given theory or experiment is then said to be NCOM-nonrealizable if and only if no such representation is possible. NCOM-realizability is the notion of classical explainability we use in this work.  

Although the question of NCOM-realizability has usually been studied only in the prepare-measure circuit just discussed, it is equally applicable to circuits that have more complex structures and more quantum processes~\cite{Schmid2024, zhang2024parallel}. 
In this work, we focus on the particular prepare-measure circuit depicted in Fig.~\ref{fig:ontology}(b), \cmt on the preparation side, it consists of a process having a \textit{bipartite} quantum system $AB$ as outputs and no classical input (i.e., no setting variable). Such a process is associated to a single bipartite state $\msf P=\rho^{AB}$. On the measurement side, it consists of a multi-measurement on A (Alice), denoted $\msf{N}:=\{\{N^A_{a|x}\}_{a}\}_{x}$, and a multi-measurement on B (Bob), denoted  $\msf{M}:=\{\{M^B_{b|y}\}_{b}\}_{y}$. 
We refer to such a triple $\mc B:=(\msf{N},\msf{M},\msf P)$ as
a \textit{Bell circuit} (or a {bipartite Bell circuit}). A multipartite Bell circuit can be defined similarly. \blk

% \yujie 
% Although one could consider a more general bipartite prepare-and-measure circuit involving a multi-state, the present work focuses on the standard Bell circuit, which consists of a single bipartite preparation and local measurements, and in which the task is to certify the entanglement of a distinguished target state. Introducing auxiliary preparations or allowing entangling measurements in order to establish additional operational identities would lead to a different certification protocol, which we leave for future work.
% \blk
\begin{figure}
    \centering
    \includegraphics[width=\linewidth]{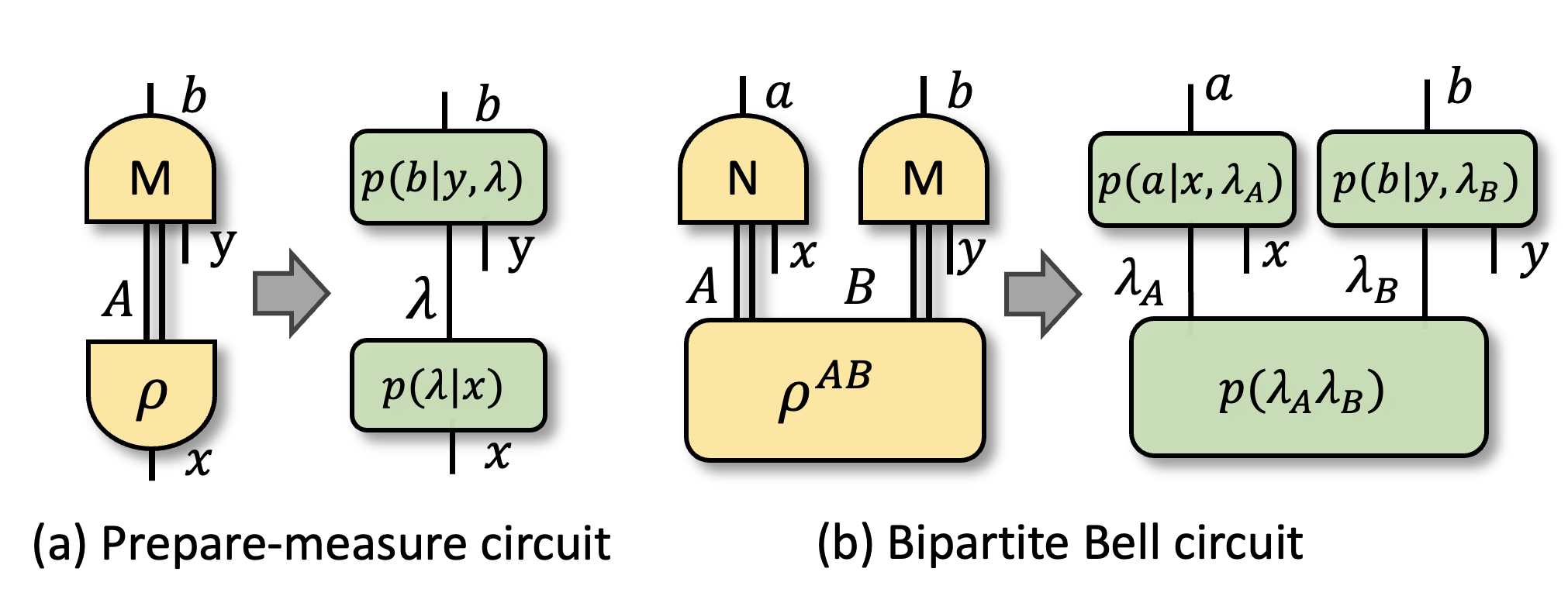}
    \caption{(a) a unipartite prepare-measure (PM) circuit;  (b) a bipartite prepare-measure circuit, i.e., a {\em bipartite Bell circuit} involving two parties labeled Alice (A) and Bob (B). A classical explanation for a given quantum circuit is given by a linear and diagram-preserving map from the quantum representation to an ontological representation; i.e., to a circuit wherein systems are classical variables and transformations are (sub)stochastic maps on these variables. 
    }
    \label{fig:ontology}
\end{figure}
\subsection{NCOM-realizability of the statistics of a bipartite Bell circuit} 
\label{sec:IIB}
\cmt Quantum theory predicts that for the Bell circuit $\mc B:=(\msf{N},\msf{M},\msf P)$\blk, the correlations obtained from performing a local multi-measurement on $A$ with setting $x$ and outcome $a$ and a local multi-measurement on $B$ with setting $y$ and outcome $b$ are given by 
\begin{equation}
p(ab|xy)=\tr[(N^A_{a|x}\otimes M^B_{b|y})\rho^{AB}].
\end{equation}
These are referred to as the \textit{operational statistics}. Similar to the case of a prepare-measure circuit, we take these operational statistics to be \textit{classically explainable} if they are realizable by a noncontextual ontological model~\cite{Schmid2024}. We now review the two key structural features that characterize a noncontextual ontological model when specialized to the Bell circuit.  
\\
\textit{Diagram-preservation (locality) --} The ontological representation must preserve the circuit structure. The quantum system $AB$ is represented by an ontic states space $\Lambda_{AB}$, which is the Cartesian product of the local ontic state spaces, i.e., $\Lambda_{AB}=\Lambda_A \times \Lambda_B$, the bipartite state is represented by a joint probability distribution $p(\lambda_A\lambda_B)$, and the measurements $\{N^A_{a|x}\}_a$ and $\{ M^B_{b|y}\}_b$ are represented by sets of response functions $\{ p(a|x,\lambda_A)\}_a$ and $\{p(b|y,\lambda_B)\}_b$ \blk respectively. Accordingly, diagram-preservation implies that the operational statistics must be reproduced as
\begin{align}
 p(ab|xy)=\sum_{\lambda_A\lambda_B} p(\lambda_A\lambda_B)p(a|x,\lambda_A)  p(b|y,\lambda_B).
\label{eq:cl-bipartite}
\end{align}
Note that this is exactly the standard form of a local hidden variable model for a Bell circuit~\cite{Brunner2014}.
\\
\textit{Linearity (noncontextuality)} 
—  An operational identity satisfied by the multi-measurement $\msf{N}= \{\{ N^A_{a|x}\}_a\}_x$ has the form $\sum_{a,x} \alpha_{a,x}N^A_{a|x}=\mbb{0}$, and is therefore defined by the vector of coefficients ${\vec \alpha} := (\alpha_{a,x})_{a,x}$. Let $\mc{O}_{\rm all}(\msf{N})$ represent the set of all vectors associated to the operational identities of $\msf{N}$. Let a similar definition hold for multi-measurement $\msf{M}$. Explicitly,
\begin{subequations}
\begin{align}
\mc{O}_{\rm all}(\msf{N})
 &:=\{ \vec \alpha\;
 |\sum_{a,x} \alpha_{a,x}N^A_{a|x}=\mbb{0}^A\},\label{op:meas1} \\ 
 \mc{O}_{\rm all}(\msf{M})&:=\{
 \vec \beta\;
|\sum_{b,y}\beta_{b,y} M^B_{b|y}=\mbb{0}^B\}.\label{op:meas2} 
\end{align}
\label{op:meas}
\end{subequations} \blk
For an ontological model to be noncontextual, the set of response functions $\{\{ p(a|x,\lambda_A)\}_a\}_x$ representing the multi-measurement $\msf{N} = \{\{N^A_{a|x}\}_a\}_x$, must be obtained from the latter by a linear map; similarly for the set of response functions $\{\{ p(b|y,\lambda_B)\}_b\}_y$  representing the multi-measurement $\msf{M} = \{\{M^B_{b|y}\}_b\}_y$. This implies that these response functions must satisfy the same identities as are satisfied by the effects, that is,
\begin{subequations}
\begin{align}
\sum_{ax} \alpha_{a,x}p(a|x,\lambda_A)&=0 \quad \forall \vec{\alpha}
\in  \mc{O}_{\rm all}( \msf{N}), \\ 
\sum_{by} \beta_{b,y}p(b|y,\lambda_B)&=0 \quad \forall \vec{\beta}
\in  \mc{O}_{\rm all}( \msf{M}).
\end{align}
\label{oe1}
\end{subequations}

Unlike for a unipartite prepare-measure circuit, there exist non-trivial compositions of processes within a bipartite Bell circuit, namely,
\begin{subequations}
\begin{align}
\msf{N}\circ \msf P
&:=  \{\{ \tr_A[(N^A_{a|x}\otimes \mbb{1}^B)\rho^{AB}]\}_a\}_x=\{\{\tilde{\rho}^B_{a|x}\}_a\}_x, \label{eq:steer A} \\
\msf{M}\circ \msf P &:=  \{\{ \tr_B[(\mbb{1}^A\otimes M^B_{b|y})\rho^{AB}]\}_b\}_y =\{\{\tilde{\rho}^A_{b|y}\}_b\}_y.  \label{eq:steer B}
\end{align}
\label{eq:steer all}
\end{subequations}
Here, $\{\{\tilde{\rho}^B_{a|x}\}_a\}_x$ and $\{\{\tilde{\rho}^A_{b|y}\}_b\}_y$ are steering assemblages~\cite{Skrzypczyk2014, Rosset2020} \blk on parties $B$ and $A$ respectively. (One can also consider them as \textit{multi-sources} on $B$ and $A$ respectively~\cite{zhang2024parallel}.)
%\footnote{One can also consider them as \textit{multi-sources} on $B$ and $A$ respectively~\cite{zhang2024parallel}.}

An operational identity satisfied by a steering assemblage $\{\{\tilde{\rho}^B_{a|x}\}_a\}_x$ has the form $\sum_{a,x} \alpha_{a,x} \tilde{\rho}^B_{a|x}=\mbb{0}$, and is therefore defined by the vector of coefficients ${\vec \alpha} := (\alpha_{a,x})_{a,x}$. Let $\mc{O}_{\rm all}(\msf{N}\circ \msf P)$ represent the set of all vectors associated to 
operational identities of 
 $\msf{N}\circ \msf P$, and similarly for  $\msf{M}\circ \msf P$, 
 that is,\blk
\begin{subequations}
\begin{align}
 \mc{O}_{\rm all}(\msf{N}\circ \msf P)&:=\{\vec{\alpha}\;
 |\sum_{a,x}\alpha_{a,x}\tilde{\rho}^B_{a|x} =\mbb{0}^B\}, \label{op:prep1}\\ 
 \mc{O}_{\rm all}(\msf{M}\circ \msf P)&:=\{ \vec{\beta}\;
 |\sum_{b,y}\beta_{b,y} \tilde{\rho}^A_{b|y}=\mbb{0}^A\}. \label{op:prep2}
\end{align}
\label{op:prep}
\end{subequations}

The steering assemblage on $B$,  $\{\{\tilde{\rho}^B_{a|x}\}_a\}_x$, is represented in the ontological model by a set of subnormalized distributions on $\Lambda_B$, $\{\{ \tilde{p}(a,\lambda_B|x ) \}_a\}_x$. \cmt  Similarly, the steering assemblage on $A$ is represented by a set of subnormalized distributions on $\Lambda_A$, $\{\{\tilde{p}(b,\lambda_A|y)\}_b\}_y$. These are defined as:
\begin{subequations}
\begin{align}
\tilde{p}(a,\lambda_B|x ) :&= \sum_{\lambda_A} 
p(a|x,\lambda_A)p(\lambda_A\lambda_B), \label{eq:classicalsteeringass}   \\ 
\tilde{p}(b,\lambda_A|y ) :&= 
\sum_{\lambda_B} p(b|y,\lambda_B)p(\lambda_A\lambda_B). \label{eq:classicalsteeringassA}    
\end{align}
\end{subequations}
The principle of noncontextuality implies that each of these sets of distributions must be the image of a linear map on the respective steering assemblage, implying that these distributions must satisfy the same identities as are satisfied by the elements of the steering assemblage, i.e., those in Eq.~\eqref{op:prep},
\begin{subequations}
\begin{align}
\sum_{ax} \alpha_{a,x}\tilde{p}(a,\lambda_B|x )&=0~~
\forall \vec{\alpha} \in \mc{O}_{\rm all}(\msf{N}\circ \msf P), 
\\
\sum_{by}\beta_{b,y}\tilde{p}(b,\lambda_A|y ) &=0~~
\forall \vec{\beta} \in \mc{O}_{\rm all}(\msf{M}\circ \msf P).
\end{align} \label{oe2}
\end{subequations}
\blk

\cmt  
Finally, consider the parallel compositions of the elements of the two steering assemblages $\{\{\tilde{\rho}^B_{a|x}\}_a\}_x$ and $\{\{\tilde{\rho}^A_{b|y}\}_b\}_y$, i.e., the product states $\tilde{\rho}^A_{b|y}\otimes \tilde{\rho}^B_{a|x}$. There exists one final set of operational identities that links these product states to the bipartite state $\rho^{AB}$, namely,
\begin{align}
\mc{O}_{\text{all}}(\tilde{\msf P}):&=\{\vec{\gamma}|\rho^{AB}=\sum_{axby}\gamma_{axby}\tilde{\rho}^A_{b|y}\otimes\tilde{\rho}^B_{a|x}\}.
\label{op:bipartite}
\end{align}
Here we use $\msf{\tilde{P}}$ as a shorthand for the tuple $(\msf{P},\msf N\circ \msf{P},\msf M\circ \msf{P})$ that is relevant to Eq.~\eqref{op:bipartite}. The principle of  noncontextuality then implies additional constraints on the distributions $p(\lambda_A\lambda_B)$, namely, 
\begin{align}
p(\lambda_A\lambda_B)&=\sum_{axby}\gamma_{axby}\tilde{p}(a,\lambda_B|x )\tilde{p}(b,\lambda_A|y )~~\forall \vec{\gamma}\in\mc{O}_{\text{all}}(\tilde{\msf P}). \label{oe3}
\end{align}

\blk

For a specific Bell circuit $\mc{B}= (\msf{N},\msf{M},\msf P)$, the operational statistics $p(ab|xy)$ are said to be  NCOM-realizable if they can be realized by a model that is diagram-preserving (and hence local), as in Eq.~\eqref{eq:cl-bipartite}, 
and that respects the implications of noncontextuality for the full set of operational identities 
on each of $\msf{N}$, $\msf{M}$, $\msf{N}\circ \msf P$, $\msf{M}\circ \msf P$ \cmt, and $\tilde{\msf P}$, \blk i.e., that respects Eqs.~\eqref{oe1}, ~\eqref{oe2} \cmt and ~\eqref{oe3}\cmt. We denote the set of all such statistics by $\mc{NC}(\mc{B})$. 

Suppose one wishes to find conditions on the operational statistics that are merely \textit{necessary} but not sufficient for NCOM-realizability. This is achieved when one applies the principle of noncontextuality to \textit{subsets} of the full set of operational identities for one or more of $\msf{N}$, $\msf{M}$, $\msf{N}\circ \msf P$,$\msf{M}\circ \msf P$ \cmt , and $\tilde{\msf P}$\blk. In what follows, for each of them, we consider a single alternative to the complete set of operational identities thereon, namely, what we term the \textit{trivial set}. The trivial set of operational identities for  $\msf{N}$, denoted $\mc{O}_{\rm triv}(\msf{N})$, is all and only those that arise from the completeness condition $\sum_{a} N^A_{a|x} = \sum_{a} N^A_{a|x'}= \mbb{1}$.  These identities are deemed to be trivial because the corresponding ontological identities are automatically satisfied simply by the definition of a set of response functions, i.e.,  $\sum_{a} p(a |x,\lambda_A) = \sum_{a} p(a |x', \lambda_A)=1$. Similarly, the set $\mc{O}_{\rm triv}(\msf{N}\circ \msf P)$ represents those operational identities that arise from the no-signalling condition $\sum_{a} \tilde{\rho}^B_{a|x} = \sum_{a} \tilde{\rho}^B_{a|x'}$, which are trivial because the corresponding ontological identities $\sum_{a} \tilde{p}(a \lambda_B|x) = \sum_{a} \tilde{p}(a\lambda_B|x')$ are automatically satisfied due to the assumption of diagram-preservation. \cmt  For the case of $\tilde{\msf P}$, the trivial set of operational identities is defined to be the null set. 

\par 
Hence, in a specific Bell circuit $\mc{B}=(\msf N, \msf M, \msf P)$, for each of $\msf{N}, \msf{M}, \msf{N}\circ \msf P, \msf{M}\circ \msf P$ \cmt  and $\tilde{\msf P}$\blk, we can associate a binary variable that stipulates whether we use $\mc{O}_{\rm all}$ or $\mc{O}_{\rm triv}$ for it.  That is, letting $i, j, u, v, w\in \{\rm a,t \}$, where a stands for `all' and t stands for `trivial', one can define  \cmt  
\begin{align}
&\mc{O}_{ijuvw}(\mc{B})\\
&:=\mc{O}_{i} (\msf{N})\cup \mc{O}_{j}(\msf{M})\cup
\mc{O}_{u}(\msf{N}\circ \msf P)
\cup
\mc{O}_{v}(\msf{M}\circ \msf P)\cup
\mc{O}_{w}(\tilde{\msf P}).\notag
\end{align}
We thereby define 32 different subsets of operational identities: $\{ \mc{O}_{ijuvw}(\mc{B}) \}_{ijuvw\in \{{\rm a},{\rm t}\}^{\times 5}}$.
Clearly, we have inclusion relations such as:
\begin{align}
\mc{O}_{\rm aaaaa}(\mc{B})&\supseteq\mc{O}_{\rm aaaat}(\mc{B}) \supseteq  \mc{O}_{\rm aaatt}(\mc{B}), \mc{O}_{\rm aatat}(\mc{B})  \label{eq:inclrelsO}\\
&\supseteq\mc{O}_{\rm aattt}(\mc{B}) \supseteq  \mc{O}_{\rm atttt}(\mc{B}), \mc{O}_{\rm tattt}(\mc{B}) \supseteq  \mc{O}_{\rm ttttt}(\mc{B}).\notag
\end{align}

In this Letter, we focus on the four cases of $\mc{O}_{\rm aattt}(\mc{B})$, $\mc{O}_{\rm atttt}(\mc{B})$, $\mc{O}_{\rm tattt}(\mc{B})$ and $ \mc{O}_{\rm ttttt}(\mc{B})$, because these define, respectively, the separable-entangled boundary, the unsteerable-steerable boundary  (in the $B{\rightarrow}A$ and $A{\rightarrow}B$ directions respectively) and the local-nonlocal boundary.  Other subsets are also of interest. $ \mc{O}_{\rm aaatt}(\mc{B})$, for instance, characterizes those bipartite states that can prepare nonclassical ensembles of states on B~\cite{zhang2025tech}, which have been studied in Ref.~\cite{Plavala2024}. 
\blk 

\begin{definition}
Consider a Bell circuit $\mc{B}=(\msf{N},\msf{M},\msf P)$ and a subset of its operational identities, denoted $\mathcal{O}(\mc{B})$. The operational statistics $\mbf{p}:=\{p(ab|xy)\}_{abxy}$ achieved in this circuit are said to be {\em NCOM-realizable relative to $\mc{O}(\mc{B})$} 
if they can be realized by a model that is diagram-preserving (and hence local), as in Eq.~\eqref{eq:cl-bipartite}, and that respects the implications of noncontextuality for the operational identities in $\mathcal{O}(\mc{B})$, i.e., that respects the subset of ontological identities in Eqs.~\eqref{oe1}, \eqref{oe2} \cmt and ~\eqref{oe3}\cmt. 
 \label{def:statistics}
\end{definition}

In Ref.~\cite{zhang2024parallel}, the notion of classical explainability of the operational statistics of an experiment---in other words, the notion of  NCOM-realizability of these statistics---was leveraged to define a notion of classicality for an {\em individual process} within an experiment. In the case of a bipartite state, it is defined to be classical if and only if the operational statistics generated by the set of circuits obtained by composing the state with every possible pair of local measurements are classically explainable. 

\par    
More formally, let $\rho$ denote the bipartite state in question. \cmt We denote the multi-measurements that consist of the {\em full} set of possible quantum measurements on $A$ by $\msf{N}^{\rm full}$ and the one that consists of
the full set on $B$ by $\msf{M}^{\rm full}$ (i.e., these have setting variables that range over all possible measurements)\blk. The resulting Bell circuit, denoted 
$\mc{B}^{\rm full}=(\msf{N}^{\rm full},\msf{M}^{\rm full},\msf P)$, is termed a {\em measurement-full Bell circuit}. The proposal of Ref.~\cite{zhang2024parallel}, then, can be stated as follows:
\begin{definition}
Consider the measurement-full Bell circuit $\mc{B}^{\rm full}=(\msf{N}^{\rm full},\msf{M}^{\rm full},\msf P=\rho^{AB})$. The bipartite state $\rho^{AB}$ is said to be {\em classical} \footnote{\cmt  For readers who would prefer a terminology that is less generic than `classical', these bipartite states can be referred to as \textit{Leibniz-classical}, a terminology introduced in Ref.~\cite{zhang2024parallel}\blk } if the operational statistics arising in this circuit are NCOM-realizable relative to \cmt  $\mc{O}_{\rm aaaaa}(\mc{B}^{\rm full})$ in the sense of Definition~\ref{def:statistics} \blk. 
\label{def:classicalstate}
\end{definition} 

\cmt  We refer the reader to Ref.~\cite{zhang2024parallel} for a complete characterization of \textit{classical} bipartite states.

The present work moves beyond the results of Ref.~\cite{zhang2024parallel} by considering \blk the notion of  NCOM-realizability {\em relative to different subsets} of the set of all operational identities, $\mathcal{O}(\mc{B}^{\rm full})\subseteq \mathcal{O}_{\rm aaaaa}(\mc{B}^{\rm full})$, to define necessary (but in general not sufficient) conditions for the classicality of a bipartite state in the Bell circuit $\mc{B}^{\rm full}$.

\par 
\begin{definition}
Consider a measurement-full  Bell circuit  $\mc{B}^{\rm full}=(\msf{N}^{\rm full},\msf{M}^{\rm full},\msf P=\rho^{AB})$ and a subset of operational identities, denoted $\mathcal{O}(\mc{B}^{\rm full})$.  The bipartite state $\rho^{AB}$ is said to be {\em classical  relative to $\mathcal{O}(\mc{B}^{\rm full})$} 
if the operational statistics arising in this circuit are NCOM-realizable relative to $\mc{O}(\mc{B}^{\rm full})$ in the sense of Definition~\ref{def:statistics}. 
 \label{def:bipartite}
\end{definition}

It was shown in \cite{zhang2024parallel} that the set of bipartite states that are classical (in the sense of Definition~\ref{def:classicalstate}) is a strict subset of the separable states. 
Here, we demonstrate that, within this set of bipartite states, considering classicality {\em relative to various choices of subsets of the full set of operational identities} (in the sense of Definition~\ref{def:bipartite}) provides a novel and unified operational definition of the separable-entangled boundary, the steerable-unsteerable boundary, and the local-nonlocal boundary. \blk

\begin{theorem}
\label{thm: connection}
Let $\mc{B}^{\rm full}=(\msf{N}^{\rm full},\msf{M}^{\rm full},\msf P=\rho^{AB})$ denote a measurement-full Bell circuit. For the bipartite state $\rho^{AB}$,  there is an equivalence between the two properties in  each of the following pairs:\\
(1) classicality relative to $\mc{O}_{\rm aattt}(\mc{B}^{\rm full})$ and separability, \\
(2) classicality relative to $\mc{O}_{\rm atttt}(\mc{B}^{\rm full})$ and $B{\rightarrow}A$ unsteerability, \\
(3) classicality relative to $\mc{O}_{\rm tattt}(\mc{B}^{\rm full})$ and $A{\rightarrow}B$ unsteerability,\\
(4) classicality relative to $\mc{O}_{\rm ttttt}(\mc{B}^{\rm full})$ and locality. \\
\end{theorem}

\begin{figure*}
    \centering
    \includegraphics[width=0.95\linewidth]{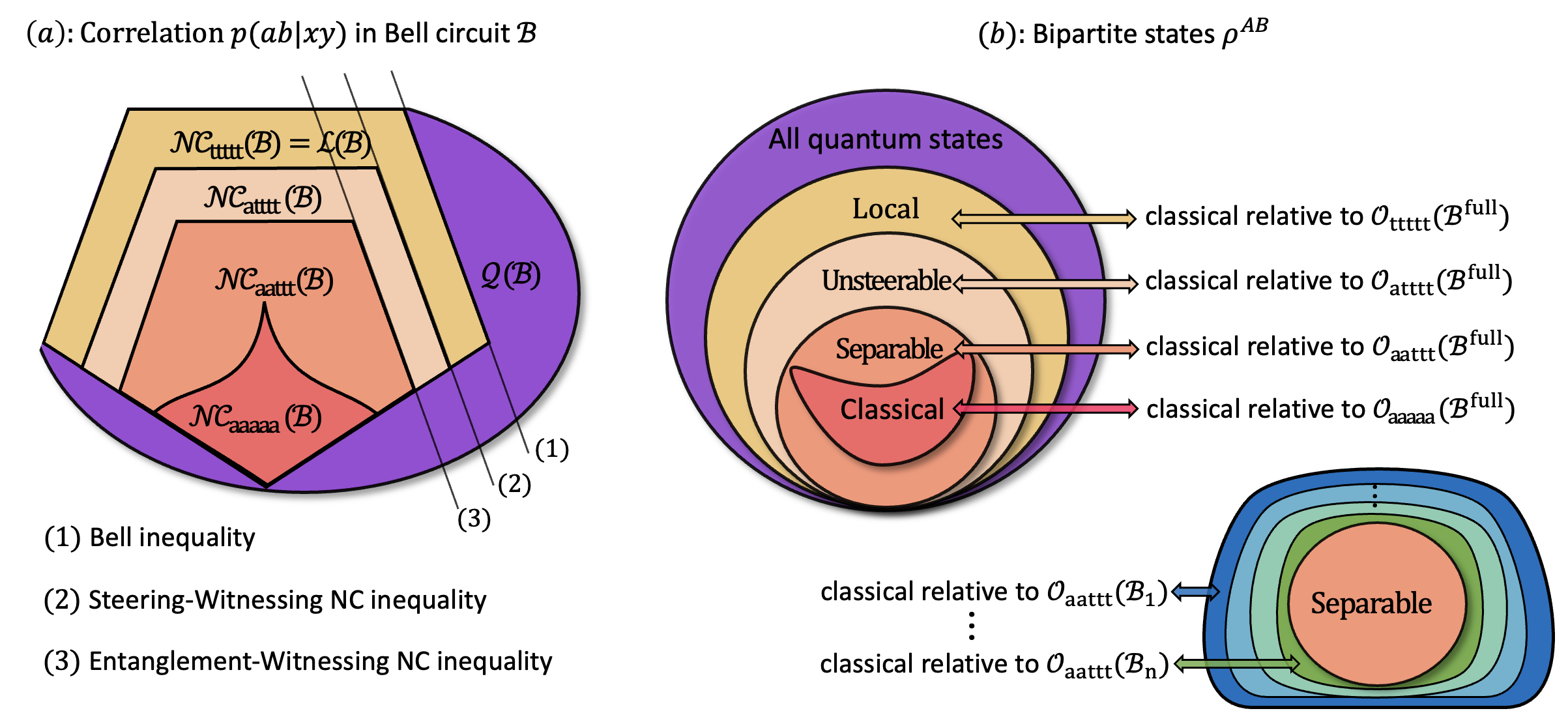}
    \caption{(a): Given a fixed $(\msf N, \msf M,\msf P)$-Bell circuit $\mc{B}$, we sketch different NCOM-realizable polytopes in the space of operational statistics, where $\mc{NC}(\mc{B})=\mc{NC}_{\rm aaaaa}(\mc{B}) \subseteq \mc{NC}_{\rm aattt}(\mc{B}) \subseteq \mc{NC}_{\rm atttt}(\mc{B})\subseteq \mc{NC}_{\rm ttttt}(\mc{B})=\mc{L}(\mc{B})\subseteq \mc{Q}(\mc B)$. Here, $\mc{NC}_{\rm aaaaa}(\mc{B})$ is defined by considering operational identities; $\mc{NC}_{\rm aattt}(\mc{B})$ is defined by considering all nontrivial operational identities on Alice's and on Bob's measurements; $\mc{NC}_{\rm atttt}(\mc{B})$ is defined by considering all nontrivial operational identities only on Alice's measurements; $\mc{NC}_{\rm ttttt}(\mc{B})=\mc{L}(\mc{B})$ is defined by not taking into account any nontrivial operational identities, and $\mc{Q}(\mc B)$ is the set of all quantum correlations that is realizable by a Bell circuit with the same cardinality of the setting and outcome variables. 
(b): a bipartite state being classical,  separable, unsteerable or local,  is equivalent to it being classical relative to $\mc{O}_{\rm aaaaa}(\mc{B}^{\rm full})$, $\mc{O}_{\rm aattt}(\mc{B}^{\rm full})$, $\mc{O}_{\rm atttt}(\mc{B}^{\rm full})$ or $\mc{O}_{\rm ttttt}(\mc{B}^{\rm full})$ respectively. Insert: Using a sequence of Bell circuits $\{\mc{B}_i= (\msf{N}_i,\msf{M}_i,\msf P)\}_i$ that limit to the measurement-full Bell circuit with $\msf{N}_1\subset\cdots\subset \msf{N}_n \subset \cdots\subset\msf{N}^{\rm full}$ and $\msf{M}_1\subset\cdots \subset \msf{M}_n\subset \cdots\subset\msf{M}^{\rm full}$, the set of states that are classical relative to $\mc{O}_{\rm aattt}(\mc{B}_{i})$ limits to 
the set of separable states.  } 
    \label{fig:relations}
\end{figure*}

We leave the proof of this theorem to  Appendix~\ref{sec: methodA}. The last case follows directly from the fact that if none of the nontrivial operational identities are leveraged, then the only constraint on the correlations is that of Eq.~\eqref{eq:cl-bipartite}, which is simply realizability by a local hidden variable model~\cite{Brunner2014}.  

It is worth noting that in a related work~\cite{Plavala2024}, the prepare-measure circuit on system $B$ with preparation that arises from $A\rightarrow B$ steering in a bipartite Bell circuit is considered, and it is shown that a certain witness of NCOM-nonrealizability is a necessary but not sufficient condition for entanglement. The key difference between this and the theorem above lies in the choice of operational identities considered.  Specifically, Ref.~\cite{Plavala2024} effectively considers NCOM-realizability relative to $\mc{O}_{\rm aaatt}$, whereas case $(1)$ in Theorem~\ref{thm: connection} in our work focuses on NCOM-realizability relative to $\mc{O}_{\rm aattt}$. By considering the consequences of noncontextuality for a smaller set of operational identities, we are able to exactly single out the separable-entangled divide among bipartite states. Because Ref.~\cite{Plavala2024} considers more operational identities, 
the constraints from noncontextuality are more stringent, and so single out a different divide, between a strict {\em subset} of the separable states and the complement of this set. This subset of bipartite states can be understood as those that are not capable of steering to a nonclassical set of states on $B$, for the notion of nonclassicality \blk defined in \cite{zhang2025tech, zhang2024parallel}.
More formally, we have: 
\begin{proposition}
\label{cor: connection}
\cite{Plavala2024, zhang2025tech}
Consider the measurement-full Bell circuit $\mc{B}^{\rm full}=(\msf{N}^{\rm full},\msf{M}^{\rm full},\msf P=\rho^{AB})$.
For the bipartite state $\rho^{AB}$,  there is an equivalence between the two properties in  each of the following pairs:\\
(5) classicality relative to $\mc{O}_{\rm aaatt}(\mc{B}^{\rm full})$ and the impossibility of steering to a nonclassical set of states on $B$, \\
(6) classicality relative to $\mc{O}_{\rm aatat}(\mc{B}^{\rm full})$ and the impossibility of steering to a nonclassical set of states on $A$. 
\end{proposition}
The proof of proposition~\ref{cor: connection} above proceeds in the same way as the proof of Theorem 1 in Ref.~\cite{Plavala2024}. \cmt For completeness, and to highlight its connection to our broader framework, we review the argument in Appendix~\ref{sec: methodA} using our notation.  

Experimentally, one cannot implement \textit{all} local measurements that are deemed possible in quantum theory. Consequently, the multi-measurements $\msf{N}$ and $\msf{M}$ that arise in an experimental Bell circuit are necessarily distinct from $\msf{N}^{\rm full}$ and $\msf{M}^{\rm full}$.  In the following section, we consider different choices of $\msf{N}$ and $\msf{M}$ and analyze what they imply for witnessing nonclassicality relative to a given set of operational identities and for certifying bipartite resourcefulness.

\section{Entanglement Certification}  
\label{Sec:entanglement-certification}

In a Bell circuit $\mc B= (\msf{N},\msf{M},\msf P)$, there are different ways of achieving entanglement certification, including those summarized in Table~\ref{tab:my_label}. They differ in several respects, as we discuss in detail in Sec.~\ref{Section:comparison}. Among them, an important distinction concerns the assumptions one has to make about the processes appearing in the Bell circuit. Consider the following alternatives: (1) one has a full characterization of the processes appearing in the circuit; (2) one knows only the cardinalities of the setting and outcome variables; and (3) one has characterized only the operational identities satisfied by the processes in the circuits.  We now consider how to achieve entanglement certification relative to each of these alternatives.

(1) In this case, as long as the sets of measurements in $\msf{N}$ and $\msf{M}$ are characterized and are informationally complete, then one can tomographically reconstruct $\msf{P}=\rho^{AB}$ and certify entanglement in the bipartite state $\rho^{AB}$ using standard separability criteria~\cite{Guhne2009}.

(2) In this case, entanglement certification can be done in a device-independent manner by looking for a violation of a Bell inequality~\cite{Brunner2014}. Specifically,  \blk consider a Bell circuit with outcomes $a, b \in \{0, \dots, o-1\}$ and settings $x \in \{0, \dots, \Delta_{\msf N}-1\}$ and $y \in \{0, \dots, \Delta_{\msf M}-1\}$.  The correlations  $\{p(ab|xy)\}_{abxy}$ can be viewed as points in the probability space $\textbf{p} \in \mbb{R}^{o^2 \Delta_{\msf N} \Delta_{\msf M}}$. 
Device-independent certification in Bell circuits is typically analyzed relative to the notion of a local hidden variable (LHV) model~\cite{Brunner2014}. The convex set of all LHV-realizable probability distributions $p(ab|xy)$, formalized in  Eq.~\eqref{eq:cl-bipartite}, is known as the {\em local polytope}, denoted  $\mc{L}(\mc{B})$, \cmt which depends only on the cardinalities of the setting and outcome variables in $\mc{B}$. Bell inequalities are the facet inequalities of $\mc{L}(\mc{B})$. \blk Because LHV-realizability coincides with NCOM-realizability relative to $\mc{O}_{\rm ttttt}(\mc{B})$ (i.e., both are defined by Eq.~\eqref{eq:cl-bipartite}), Bell inequalities can also be seen as the weakest set of noncontextuality inequalities in the bipartite Bell circuit.  The simplest and most well-known example of Bell inequalities \blk occurs for $\Delta_{\msf N} = \Delta_{\msf M} = o = 2$, leading to the Clauser-Horne-Shimony-Holt (CHSH) inequalities~\cite{Clauser1969, nobel2022}. For certain Bell circuits, such as the one described in example~\ref{example1}, Bell inequalities are the {\em only} noncontextuality inequalities.

\cmt
(3) This is the case of the entanglement certification method proposed here. It differs from (1) in that it does not rely on a full characterization of the measurements in $\msf{N}$ and $\msf{M}$; it also differs from the Bell test in (2), since it relies on operational identities in the Bell circuit $\mc{B}$. These operational identities can be faithfully inferred from the operational statistics of $\mc{B}$ without characterizing the devices, under the sufficient condition of \textit{tomographic completeness}, as we detail in Sec.~\ref{Sec:inexact-identities}. 

Our entanglement certification method relies on characterizing the set of all NCOM-realizable correlations relative to various choices of operational identities. More formally, \cmt
\begin{definition}
The noncontextual polytope $\mc{NC}_{\rm i j u v w}(\mc{B})$ is defined to be the set of correlations $\textbf{p}:=\{p(ab|xy)\}_{abxy} \in \mbb{R}^{o^2 \Delta_{\msf N} \Delta_{\msf M}}$ that are NCOM-realizable relative to the set $\mc{O}_{\rm i j u v w}(\mc{B})$ of operational identities, i.e., 
\begin{align*}
&\mc{NC}_{\rm ijuvw}(\mc{B})\\
&:=\{\mbf{p}~\text{is NCOM-realizable relative to}~\mc{O}_{\rm i j u v w}(\mc{B})\}.\notag 
\end{align*}
Different from the local polytope $\mc L(\msf B)$, this set depends not only on the cardinalities of the setting variables and outcome variables in $\mc B$, but also on additional details, namely, those required to determine the operational identities in $\mc B$. An inequality whose violation witnesses nonmembership in $\mc{NC}_{\rm i j u v w}(\mc{B})$ will be termed a \textit{noncontextuality inequality relative to $\mc{O}_{\rm i j u v w}(\mc{B})$}.
\end{definition}
\blk

Given the subset inclusion relations of Eq.~\eqref{eq:inclrelsO}, we have that, for any Bell circuit $\mc B$, 
\begin{align}
\mc{NC}(\mc{B})&=\mc{NC}_{\rm aaaaa}(\mc{B})  \subseteq \mc{NC}_{\rm aattt}(\mc{B}) \label{eq:hierarchy} \\
&\subseteq \mc{NC}_{\rm atttt}(\mc{B}), \mc{NC}_{\rm tattt}(\mc{B})\subseteq \mc{NC}_{\rm ttttt}(\mc{B})=\mathcal{L}(\mc{B}),\notag 
\end{align}
\blk
as depicted schematically in Fig~\ref{fig:relations}~(a). It should be noted that, apart from $\mc{NC}(\mc{B})$, 
each of these sets of correlations above is convex, as they are each defined by imposing additional linear constraints (in the form of Eq.~\eqref{op:meas}) on the convex set $\mc{L}(\mc{B})$. \cmt Therefore, apart from $\mc{NC}(\mc{B})$,  their boundary will be determined by sets of \textit{linear} \textit{noncontextuality} inequalities.\blk

In the following, we focus mostly on the noncontextual polytope $\mc{NC}_{\rm aattt}(\mc{B})$ in the context of some particular examples of Bell circuits $\mc{B}$. \cmt The properties of $\mc{NC}_{\rm aattt}(\mc{B})$, namely, its dimension and vertex complexity, are detailed in Appendix~\ref{sec: methodB},  as are the \blk numerical procedures underpinning its analysis, namely, facet enumeration and polytope membership linear programming.

\yujie 
Throughout this section, the relevant operational identities are derived from the theoretical specification of the Bell circuit. The question of how these identities can be inferred from or enforced in experimental data is addressed in Sec.~\ref{Sec:inexact-identities} and~\ref{sec: tomo-loopholes}. \blk

\subsection{Certifying two-qubit noisy isotropic states with fixed measurements}
\label{secIIIA}
Recall that the local polytope $\mc{L}(\mc{B})$ depends only on the cardinality of setting variables in $\msf{N}$ and $\msf{M}$, denoted $\Delta_{\msf N}$  and $\Delta_{\msf M}$ respectively, and the number of distinct outcomes for each, denoted $o$.  We can therefore refer to the class of Bell circuits sharing this same local polytope as the $(\Delta_{\msf N},\Delta_{\msf M},o,o)$-Bell circuits, denoted $\msf{Bell}_{\Delta_{\msf N}\Delta_{\msf M}oo}$, and write $\mc{L}(\msf{Bell}_{\Delta_{\msf N}\Delta_{\msf M}oo})$ for their common local polytope. 

By contrast, the noncontextual polytope $\mc{NC}_{\rm i j u v w}(\mc{B})$ depends also on the set of operational identities $\mc{O}_{\rm i j u v w}(\mc{B})$ (though, it still does not depend on the full characterization of $\mc{B}$). In the following, since we are only considering the sets of operational identities $\mc{O}_{\rm aattt}(\mc{B})$, $\mc{O}_{\rm atttt}(\mc{B})$, $\mc{O}_{\rm tattt}(\mc{B})$ and $\mc{O}_{\rm ttttt}(\mc{B})$,  and hence only the nontrivial operational identities in the multi-measurement $\msf{N}$ and $\msf{M}$, it is convenient to encode these identities as the intrinsic geometries of $\msf{N}$ and $\msf{M}$, that is, as shapes, whose vertices satisfy the linear constraints defined by $\mc{O}_{\rm all}(\msf{N})$ and $\mc{O}_{\rm all}(\msf{M})$, respectively. Accordingly, we denote this pair of shapes by $P$ and $Q$ respectively, and refer to the corresponding \textit{class} of Bell circuits as $(P,Q)$-Bell circuits, denoted $\msf{Bell}_{\rm PQ}$. For any $\mc{B}\in \msf{Bell}_{\rm PQ}$, we can write $\mc{NC}_{\rm i j u v w}(\msf{Bell}_{\rm PQ})$ for the corresponding noncontextual polytope relative to $\mc{O}_{\rm i j u v w}(\mc{B})$. \blk

For example, consider the class of the simplest type of Bell circuit with $\Delta_{\msf N} = \Delta_{\msf M} = o = 2$. For any pair of binary-outcome measurements, $\msf{N} = \{\{N^A_{a|x}\}_a\}_x$, the full set of operational identities satisfied by this set of effects is the singleton set consisting of the trivial identity: $N_{0|0}+N_{1|0}=N_{0|1}+N_{1|1}$. The intrinsic geometry of this set of effects is that of a square, since the vertices of the square, denoted $\{\vec{t}_{a,x}\}_{a,x}$, satisfy exactly one linear constraint: $\vec{t}_{0|0}+\vec{t}_{1|0}=\vec{t}_{0|1}+\vec{t}_{1|1}$. \cmt Therefore, this class can be denoted either as the class of $(2,2,2,2)$-Bell circuits, $\msf{Bell}_{2222}$, or under our convention as the class of (square, square)-Bell circuits, denoted $\msf{Bell}_{\rm ss}$. \blk

\begin{example}
\label{example1}
Consider the \cmt class of (square, square)-Bell circuits, i.e.,  $\msf{Bell}_{\rm ss}$. \blk  We have $\mc{O}_{\rm aattt}(\text{Bell}_{\rm ss})=  \mc{O}_{\rm atttt}(\msf{Bell}_{\rm ss}) = \mc{O}_{\rm ttttt}(\text{Bell}_{\rm ss}).$
It follows that we have a collapse of the hierarchy of polytopes in Eq.~\eqref{eq:hierarchy}, that is, 
$$\mc{NC}_{\rm aattt}(\text{Bell}_{\rm ss})= \mc{NC}_{\rm atttt}(\text{Bell}_{\rm ss}) = \mc{NC}_{\rm ttttt}(\text{Bell}_{\rm ss}).$$ 
Because $\mc{NC}_{\rm ttttt}(\text{Bell}_{\rm ss})=\mc{L}(\msf{Bell}_{2222})$, the set of noncontextuality inequalities for this class is equivalent to the set of CHSH inequalities.  
\end{example}

\cmt Similarly, one can observe that for the class of (octahedron, octahedron)-Bell circuits, denoted $\msf{Bell}_{\rm oo}$, the noncontextual polytope associated to $\mc O_{\rm aattt}(\msf{Bell}_{\rm oo})$ also coincides with the local polytope for $(3,3,2,2)$-Bell circuits, i.e., $\mc{NC}_{\rm aattt}(\msf{Bell}_{\rm oo})=\mc L(\msf{Bell}_{3322})$. 
\blk Consequently, to obtain nontrivial noncontextuality inequalities that offer new avenues for entanglement certification (i.e., beyond the avenues already offered by Bell tests, \cmt thereby showing that the final inclusion relation in the hierarchy of Eq.~\eqref{eq:hierarchy} can be made strict), \blk it is necessary to consider a Bell circuit with more structure in the local intrinsic geometries of the measurements. 

\par 

In the next three examples, we take the shape describing the intrinsic geometry of $\msf{N}$ to be the geometric dual of the shape describing the intrinsic geometry of $\msf{M}$.
This strategy is motivated by the fact that such a geometric duality relation achieves the maximal quantum violation of the CHSH inequalities, and the fact that a similar duality is conjectured to also be optimal for quantum violations of noncontextuality inequalities in prepare–measure circuits~\cite{zhang2025tech}.  
In particular, we will use certain pairs of regular polyhedra because they may enable a more noise-robust entanglement certification protocol~\cite{Saunders2010} and because their symmetries simplify the analytical results.

To demonstrate their usefulness, the resulting noncontextuality inequalities are compared with all previously known Bell- and steering-based tests for certifying entanglement in the 1-parameter family of 2-qubit noisy isotropic states, 
\begin{equation}
\rho_{\text{Iso}}^v=v\op{\Phi^+}{\Phi^+}+(1-v)\frac{\mbb{1}}{4},\label{eq:isotropic}
\end{equation}
where $\ket{\Phi^+}=(\ket{00}+\ket{11})/\sqrt{2}$ is a maximally entangled state shared between Alice and Bob. 

\begin{figure}
    \centering
    \includegraphics[width=0.95\linewidth]{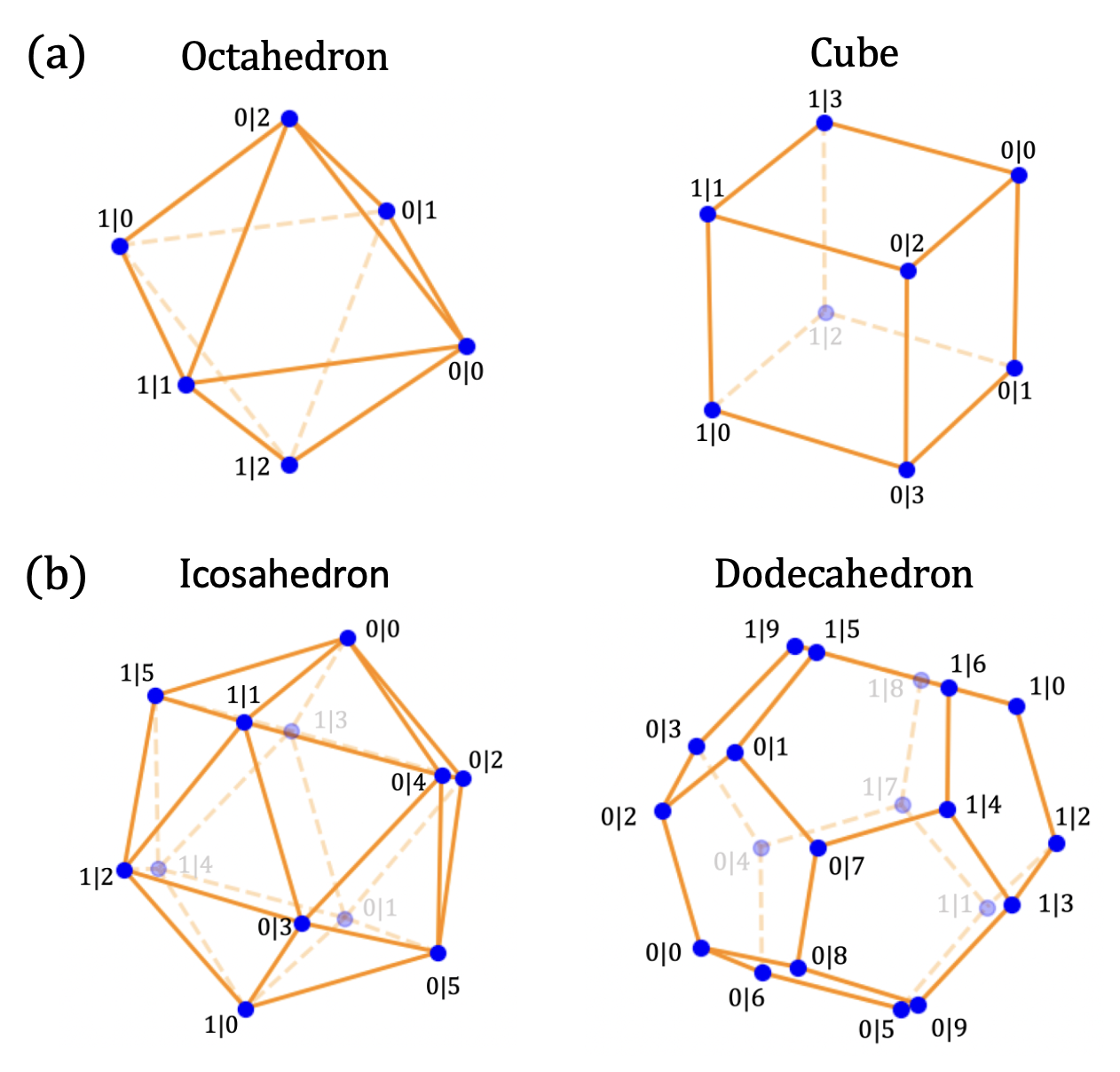}
    \caption{Geometric configuration of $\msf{N}=\{\{N^A_{a|x}\}_{a}\}_{x}$ (Left) and  $\msf{M}=\{\{M^B_{b|y}\}_{b}\}_{y}$ (Right). (a) Example~\ref{example2}:  with cubic and octahedral symmetry; (b) Example~\ref{example3}: with icosahedral and dodecahedral symmetry. } 
    \label{fig:plantoic}
\end{figure}

We now turn to the case where $P$ and $Q$ are an octahedron and a cube, respectively, as depicted in Fig.~\ref{fig:plantoic}(a).
\begin{example}
\label{example2}
Consider the \cmt class of (octahedron, cube)-Bell circuits, denoted $\msf{Bell}_{\rm oc}$. \blk We have $\mc{O}_{\rm aattt}(\msf{Bell}_{\rm oc}) = \mc{O}_{\rm tattt}(\msf{Bell}_{\rm oc}) \subset  \mc{O}_{\rm atttt}(\msf{Bell}_{\rm oc}) = \mc{O}_{\rm ttttt}(\msf{Bell}_{\rm oc}).$  (The equalities are due to the fact that the octahedron, like the square, describes only trivial operational identities, while the strict inclusion is due to the fact that the cube includes some nontrivial operational identities.) It follows that the inclusion relations in the hierarchy of correlations described in Eq.~\eqref{eq:hierarchy} become: 
\begin{align*}
\mc{NC}_{\rm aattt}(\msf{Bell}_{\rm oc}) = \mc{NC}_{\rm tattt}(\msf{Bell}_{\rm oc}) &\subset 
\mc{NC}_{\rm atttt}(\msf{Bell}_{\rm oc}) \\
&= \mc{NC}_{\rm ttttt}(\msf{Bell}_{\rm oc}).
\end{align*}
The characterization of the polytope $\mc{NC}_{\rm aattt}(\msf{Bell}_{\rm oc})$ is discussed in the Appendix~\ref{sec: methodB}. It is found to have 336 nontrivial facet inequalities (excluding the trivial nonnegativity constraints), which can be classified into three distinct classes, where the elements within each class form an orbit under permutations of the values of the outcome variables and of the values of the setting variables. See Table~\ref{tab:ineq68}. 

\begin{table}[h]
\caption{Facet-defining inequalities for the polytope $\mc{NC}_{\rm aattt}(\msf{Bell}_{\rm oc})$
}
\label{tab:ineq68}
\begin{ruledtabular}
\begin{tabular}{lc}
Inequality & Orbit size\\
\hline
$p(01|22) + p(10|03) - p(11|00) + p(11|21) \ge 0$  & $144$\\
$p(01|22) + p(10|03) - p(11|00) + p(11|20) \ge 0$  & $144$\\
\makecell[l]{$p(01|22) +p(10|03) - p(11|00) + p(11|10)$ \\ 
\;\;\;$- p(11|12) - p(11|20) + p(11|21) + p(11|22) \ge 0$} & $48$\\
\end{tabular}
\end{ruledtabular}
\end{table}\par
In particular, the third class of inequalities can be used to certify entangled two-qubit isotropic states for $v > \frac{1}{\sqrt{3}}\simeq 0.577$, thereby outperforming all certification methods for such states that rely on Bell inequalities based on projective measurements, which only certify entanglement for $ v > \frac{1}{K_G(3)}\in (0.6875,0.6961)$~\cite{Designolle2023}, where  $K_G(3)$ is the Grothendieck's constant~\cite{Acin2006}. (The CHSH inequalities
%, by contrast, 
can only be used to certify entangled two-qubit isotropic states if $v>\frac{1}{\sqrt{2}}\simeq 0.717$.)
\end{example}

Note that the range of $v$ for which entanglement is witnessable in example~\ref{example2} \textit{does not} exceed the range that is witnessable in principle by steering inequalities, namely, $v \ge 1/2$~\cite{zhang2024,Renner2024}. Consequently, to obtain nontrivial noncontextuality inequalities that offer avenues for entanglement certification beyond those offered by steering inequalities,  it is necessary to consider a Bell circuit with even more structure in the local intrinsic geometries of the measurements.  The next example does so. It considers multi-measurements having the intrinsic geometries of an icosahedron on one side and a dodecahedron on the other, as depicted in Fig.~\ref{fig:plantoic}(b).
\begin{example}
\label{example3}
Consider the \cmt class of (icosahedron, dodecahedron)-Bell circuits, denoted $\msf{Bell}_{\rm id}$. \blk We have 
$\mc{O}_{\rm aattt}(\msf{Bell}_{\rm id}) \subset \mc{O}_{\rm atttt}(\msf{Bell}_{\rm id}) , \mc{O}_{\rm tattt}(\msf{Bell}_{\rm id}) \subset \mc{O}_{\rm ttttt}(\msf{Bell}_{\rm id}).$  It follows that the inclusion relations in the hierarchy of correlations described in Eq.~\eqref{eq:hierarchy} become: 
\begin{align*}
\mc{NC}_{\rm aattt}(\msf{Bell}_{\rm id}) \subset \mc{NC}_{\rm atttt}(\msf{Bell}_{\rm id}),~&
 \mc{NC}_{\rm tattt}(\msf{Bell}_{\rm id})\\
\subset &\mc{NC}_{\rm ttttt}(\msf{Bell}_{\rm id}).
\end{align*}

Importantly, there exists the following facet-defining inequality of the polytope $\mc{NC}_{\rm aattt}(\msf{Bell}_{\rm id})$: 
\begin{align}
\mc{I}=&p(00|00) + p(00|01) - \left(\phi+1\right)p(00|02) + \frac{1}{\phi}p(00|03)\notag \\ 
+ & \frac{1}{\phi}p(00|10) + \frac{1}{\phi}p(00|11) - 2p(00|12) + p(00|13) \notag \\
- &p(00|20) - \frac{1}{\phi}p(00|21) + \left(\phi+1\right)p(00|22) - p(00|23) \notag \\
- &\frac{1}{\phi^2}p(10|11) + \frac{1}{\phi}p(10|12)\ge 0,\label{eq: Inequality}
\end{align}
\noindent where $\phi=\frac{\sqrt{5}+1}{2}$. This inequality can be violated by two-qubit isotropic states with $v>\sqrt{\frac{1+\phi^2}{3\phi^4}}\approx 0.4195$, and consequently outperforms all entanglement certification methods that rely on steering inequalities. \blk
\end{example}
\par
\cmt  
%As explained in detail in Appendix~\ref{sec: methodB}, the noncontextual polytope $\mc{NC}_{\rm aattt}(\msf{Bell}_{\rm id})$ has affine dimension $15$. %Intuitively, although
In the class of Bell circuits $\msf{Bell}_{\rm id}$, there are $12$ effects on one side and $20$ effects on the other, hence $12\times20=240$ product effects. The operational identities, however, imply that every effect on one side can be expressed as a linear combination of $4$ effects on that side, so that the $240$ product effects can be expressed as linear combinations of just $16$ product effects. It follows that the $240$ distinct conditional probability distributions $\{p(ab|xy)\}_{abxy}$ can be expressed as a linear combination of $16$ of these. For this reason, any noncontextuality inequality can be expressed in terms of at most $16$ conditional probabilities. (The inequality of Eq.~\eqref{eq: Inequality}, for instance, has only $14$ terms, and some setting variables do not appear.)

\blk

\cmt  
\begin{example}
\label{example4}
We now consider a family of dual polyhedra whose vertices become dense on the unit sphere. For each positive integer $\mu$, let $P_\mu$ denote the icosahedral geodesic polyhedron $\{3,5^+\}_{\mu,0}$ and let $Q_\mu$ denote its dual Goldberg polyhedron $\{5^+,3\}_{\mu,0}$, obtained via regular $\mu$-fold subdivisions of the icosahedron and dodecahedron~\cite{Wenninger1979}. (The case $\mu = 1$ recovers the icosahedron–dodecahedron pair used in Example~\ref{example3}.) The polyhedra $P_\mu$ and $Q_\mu$ have $10\mu^2+2$ and $20\mu^2$ vertices, respectively. For each $\mu$, this defines the class of $(P_\mu,Q_\mu)$-Bell circuits, denoted $\msf{Bell}_\mu$.

As summarized in Table~\ref{table:geodesic}, this measurement strategy is highly effective for certifying entanglement of the two-qubit isotropic states of Eq.~\eqref{eq:isotropic}: by using standard qubit projective measurements with $5\mu^2 + 1$ settings for Alice and $10\mu^2$ settings for Bob with the intrinsic geometry $(P_\mu$, $Q_\mu)$ described above,  the visibility $v$ for which entanglement certification is still possible decreases rapidly with $\mu$ and approaches the value of the visibility at the separable–entangled boundary, namely, $v = \frac{1}{3}$. 

\cmt 
\begin{table}[h]
\cmt 
\centering
\caption{\label{table:geodesic} 
Critical visibility $v$ for entanglement certification of the two-qubit noisy isotropic state $\rho_{\text{Iso}}^v$ using projective multi-measurements with the intrinsic geometries of the geodesic polyhedron $P_\mu$ for Alice and the dual Goldberg polyhedron $Q_\mu$ for Bob.}
\begin{ruledtabular}
\begin{tabular}{cccc}
$\mu$ & $\#$ of Vertices in $P_{\mu}$ & $\#$of Vertices in $Q_{\mu}$  & visibility $v$ \\
\hline
1 & 12   & 20  & $\approx 0.4195$ \\
2 & 42  & 80  & $\approx 0.3568$ \\
3 & 92  & 180  & $\approx 0.3434$ \\
4 & 162  & 320 & $\approx 0.3394$ \\
5 & 252 & 500 & $\approx 0.3372$ \\
6 & 362 & 720 & $\approx 0.3360$ \\
7 & 492 & 980 & $\approx 0.3353$ \\
\end{tabular}
\end{ruledtabular}
\end{table}
\end{example}
\subsection{Certifying two-qubit noisy isotropic states with randomized measurements}
\label{sec:random}
In the previous subsection, we restricted our attention to Bell circuits with multi-measurements whose operational identities define intrinsic geometries with nontrivial symmetries. Doing so allowed us to derive closed-form analytical expressions for the noncontextuality inequalities. However, having such symmetry is not necessary for entanglement certification. We now demonstrate that our noncontextuality-based approach remains effective and efficient when Alice and Bob use a strategy in which the measurements implemented at each party are chosen uniformly at random.

In this subsection, we continue to focus on two-qubit isotropic states, but randomly sample a set of projective measurements at each party. For each choice of
$m \in \{10,20,50,100,200\}$, we draw $m$ projective measurements for Alice and 
$m$ for Bob, independently and uniformly at random, and then determine the smallest visibility $v$ of a two-qubit isotropic state that can be certified as
entangled by our noncontextuality-based test.  

The minimal visibility $v$ of the two-qubit isotropic state for which entanglement can be witnessed for a given $m$ depends on the specific sets of measurements sampled. For instance, samplings that distribute more uniformly over the sphere are expected to yield smaller values of the minimal visibility $v$. It is useful, therefore, to report an average value of the minimum visibility $v$ in an ensemble of different samplings. We consider $100$ different samplings for each $m$. The average value of the minimal visibility $v$ for each m is reported in Table~\ref{tab:random-bases-eta}. As $m$ increases, the average critical visibility decreases steadily and approaches the value at the separable-entangled boundary $v=\frac{1}{3}$.

\begin{table}[t]
\cmt  
\centering
\caption{ Average critical visibility $v$ of the two-qubit isotropic state $\rho_{\text{Iso}}^v$ that can be certified as entangled using $m$ randomly sampled projective measurements. The values are obtained by averaging the critical visibility over $100$ independent trials for each $m$.}
\begin{ruledtabular}
\begin{tabular}{cc}
        $\#$ of bases $m$  & visibility $v$  \\
\hline 
        10  & $0.540 \pm 0.027$ \\
        20  & $0.445 \pm 0.015$ \\
        50  & $0.384 \pm 0.006$ \\
        100 & $0.360 \pm 0.003$ \\
        200 & $0.348 \pm 0.001$ \\
\end{tabular}
\end{ruledtabular}
\label{tab:random-bases-eta}
\blk
\end{table}

This randomized measurement strategy offers a useful perspective on the practicality of our entanglement-certification protocol. While employing many measurement settings may initially seem demanding, randomized measurements can, in fact, simplify the implementation. Randomized measurements underpin protocols such as classical shadows~\cite{Huang2020} and many other quantum information tasks~\cite{Elben2023}. A key distinction, however, is that these standard approaches typically rely on precise characterization of the measurement devices, whereas our method requires \textit{no prior characterization}. The fact that our framework accommodates randomized strategies could therefore point to a realistic path toward scalable entanglement certification based on generalized noncontextuality.

Beyond their generality, randomized measurement schemes offer a second key advantage: they remove the need to specify \emph{a priori} a particular set of operational identities to be targeted in the experiment. Instead, entanglement certification can proceed as follows: (i) implement a set of measurements chosen at random, (ii) infer the operational identities that are satisfied by the set using experimental data, and (iii) use these identities to test for entanglement by deriving explicit noncontextuality inequalities tailored to these identities. We return to this strategy in Sec.~\ref{sec: newIVA}, where we discuss practical routes to implementing our noncontextuality-based entanglement certification in the presence of imperfect measurements.

\subsection{Certifying arbitrary two-qubit entangled states}
\label{sec:fraction}
We have shown that our noncontextuality-based entanglement certification scheme outperforms Bell-based and steering-based schemes for two-qubit noisy isotropic states. We now demonstrate that it also outperforms Bell-based schemes for
\textit{generic} entangled bipartite states.
\par 

To do so, we perform numerical simulations on randomly sampled two-qubit entangled states. Following Ref.~\cite{Kofman2008}, we draw mixed states from the Hilbert--Schmidt measure by generating $4\times 4$ complex matrices $A$ with independent complex Gaussian entries (zero mean, equal variance) and setting $\rho = \frac{A^{\dagger} A }{\tr(A^{\dagger} A)}$, \cmt and we post-select those that are entangled. \cmt 
We use fixed projective measurements with Bloch vectors that form dual polyhedra as described in Example~\ref{example1}, \ref{example2}, \ref{example3}, and \ref{example4} with $\mu=2$ (the absolute orientation of the pair of polyhedra is irrelevant since the tested states are already uniformly sampled). For each randomly sampled entangled state, we then perform the noncontextuality-based tests introduced above.  The results are presented in Table~\ref{table:random-entangled-states}. One sees that the noncontextuality test significantly 
outperforms the Bell test. 
\begin{table}[h]
\cmt  
\centering
\caption{\label{table:random-entangled-states}
Percentage of $10^6$ randomly sampled two-qubit entangled states that are certified as entangled by the measurement strategies in Examples~\ref{example1},~\ref{example2},~\ref{example3}, and ~\ref{example4} ($\mu=2$). Here, we compare the certification methods based on the standard Bell test and those based on our noncontextuality test. We note that computing the percentage for the Bell test becomes numerically infeasible for Example.~\ref{example4}, $\mu=2$, since the corresponding local polytope has $2^{61}$ vertices; by contrast, the noncontextual polytope has only $8160$ vertices, so the noncontextuality test remains feasible. 
}
\begin{ruledtabular}
\begin{tabular}{lcc}
\multirow{2}{*}{Bell circuits}  & \multicolumn{2}{c}{(\%) certified by} \\
             & Bell test & Noncontextuality test  \\
\hline
Ex.~\ref{example1}              & $\approx 0.005\%$   & $\approx 0.005\%$ \\
Ex.~\ref{example2}              & $\approx 0.08\%$                & $\approx 0.82\%$   \\
Ex.~\ref{example3}           & $\approx 0.24\%$    & $\approx 24.40\%$ \\
Ex.~\ref{example4}, $\mu=2$     &Infeasible         & $\approx 76.35\%$ \\
\end{tabular}
\end{ruledtabular}
\blk
\end{table}
\par
We note that the measurement bases for the Bell tests in Table~\ref{table:random-entangled-states} above are not optimized to avoid the additional numerical cost they incur.

However, we can also benchmark our approach against the CHSH test described in~\cite{Kofman2008}, which, for each state, maximizes the CHSH value over all local measurement bases using the Horodecki criterion~\cite{Horodecki1995} (equivalently, test in our Ex.~\ref{example1} with measurement bases optimized). The question studied in~\cite{Kofman2008} is: among randomly sampled two-qubit entangled states, what fraction are such that their entanglement can be certified by a CHSH inequality violation? When states are sampled according to the same Hilbert–Schmidt measure assumed above, Ref.~\cite{Kofman2008} found that only about $1.09\%$ of entangled states can be certified by a CHSH test alone, even if one optimizes the choice of measurements for each sampled state (Specifically, it was shown in~\cite{Kofman2008} that about $75.8\%$ of all bipartite states are entangled (via the PPT criterion~\cite{Peres1996}), while only about $0.82\%$ violate some CHSH inequality (via the Horodecki criterion~\cite{Horodecki1995}).).
%\footnote{\cmt It was shown in~\cite{Kofman2008} that about $75.8\%$ of all bipartite states are entangled (via the PPT criterion~\cite{Peres1996}), while only about $0.82\%$ violate some CHSH inequality (via the Horodecki criterion~\cite{Horodecki1995}).} 
This is an improvement over the best fraction achievable by a Bell test without optimizing measurement bases, namely, $0.24\%$, but it still falls far short of what can be achieved by noncontextuality tests in examples~\ref{example3}, and ~\ref{example4} ($\mu=2$). 

\cmt Moreover, Bell-test-based certification with many settings quickly becomes computationally prohibitive:
the local polytope has $o^{\Delta_N}o^{\Delta_M}$ vertices,
i.e., a number of vertices that scales \textit{exponentially} with the cardinality of the setting variables. This is in contrast to the situation with the noncontextual polytope, which has
$O(o^2\Delta_N\Delta_M)$ vertices, i.e., a number of vertices scales only polynomially with the cardinality of the setting variables. (see also Appendix~\ref{sec: methodB1}).
\section{Dealing with inexact operational identities}
\label{Sec:inexact-identities}
\yujie 
In Sec.~\ref{Sec:entanglement-certification}, we derived noncontextuality inequalities from operational identities that are inferred from a theoretical specification of the Bell circuit. We now explain how these identities can be inferred from or enforced in experimental data.
\blk

\cmt Real devices never implement the targeted measurements precisely.  Moreover, one only ever obtains finite-run statistics and consequently one’s estimate of the effects deviates from the true effects because of statistical fluctuations.  It follows that the operational identities that hold among the effects actually realized in the laboratory will never coincide exactly with those that hold for the ideal target effects.
%(the ones we intend to implement). 
It is therefore crucial to develop tools that can contend with this problem in order to make our noncontextuality-based entanglement certification protocol practically applicable. 

\yujie 
There is a second, more subtle consideration. To make the protocol free from prior characterization assumptions, the operational identities used in our noncontextuality test must be established directly from the experimental data (i.e., the observed statistics), rather than being assumed. The sufficient condition ensuring this is a form of tomographic completeness: steered states on one side probe enough operational directions to identify the linear identities that hold among the measurement effects on that side. A formal definition of this condition, together with the associated tomographic-incompleteness loophole, is given in Sec.~\ref{sec: tomo-loopholes}.
\blk

In this section, we describe two complementary methods for dealing with this problem.  (i) infer from the observed statistics the operational identities that {\em happen to hold} for the multi-measurements that were actually implemented.  Because the technique fits the aphorism ``cut your coat to suit your cloth’’, we refer to it as the ``cut-to-your-cloth’’ technique. (ii) Using a generalization of the secondary procedures technique of Ref.~\cite{Mazurek2016}, i.e., from the measurement procedures in the convex hull of those realized, find a set, termed the {\em secondary} measurement procedures, that satisfy the operational identities {\em exactly} and test the statistics that obtain in the circuit composed of the secondary procedures. \blk

\cmt
Both methods begin with the standard self-consistent tomography~\cite{Nielsen2021,Mazurek2021} to mitigate statistical fluctuations in the raw relative frequencies. Therefore, our approach is well-suited to platforms that already use self-consistent tomography~\cite{Merkel2013,Greenbaum2015}.

\subsection{Self-consistent tomography} 
This step can be viewed as a regularization step that transforms the raw relative frequencies $f_{ab|xy}$ into a probability distribution that is realizable in a $D$-dimensional quantum model. Specifically, one varies over choices of bipartite state $\rho^{AB}$ and multi-measurements $\msf N^{\text{pri}}=\{\{N^{A,\text{pri}}_{a|x}\}_a\}_x$ and $\msf M^{\text{pri}}=\{\{M^{B,\text{pri}}_{b|y}\}_b\}_y$ to minimize the distance between $f_{ab|xy}$ and the quantum-realisable probability distribution $p^{\text{pri}}(ab|xy)$. The distance can be quantified in various ways; for instance, using a weighted $\chi^2$ and solving
\begin{align}
&\chi^2:=\min_{(\msf N^{\text{pri}},\msf M^{\text{pri}},\rho^{AB})}\sum_{abxy}\left|\frac{f_{ab|xy}-p^{\text{pri}}(ab|xy)}{\Delta f_{ab|xy}}\right|^2,  \\
& \text{such that~~} p^{\text{pri}}(ab|xy)=\tr[N^{A,\text{pri}}_{a|x}\otimes M^{B,\text{pri}}_{b|y}\rho^{AB}]. \notag 
\end{align}
We refer to the best-fit $p^{\text{pri}}(ab|xy)$ as the ``primary statistics", the $\msf N^{\text{pri}}$ and $\msf M^{\text{pri}}$ as the \textit{primary multi-measurements} (those actually realized in the experiments) and $\msf P=\rho^{AB}$ as the primary preparation. The tuple  $(\msf N^{\text{pri}}, \msf M^{\text{pri}}, \msf P)$ constitutes a non-unique solution of the fitting, and in general $\msf N^{\text{pri}}$ and $\msf M^{\text{pri}}$ will differ from the $\msf N$ and $\msf M$ that are targeted.

%However, with
\yujie 
Given \textit{tomographic completeness}, the regularization step above is just an instance of ``self-consistent tomography"\cite{Nielsen2021,Mazurek2021}, where all possible fits for the primary multi-measurements are related by invertible linear maps. This freedom is termed tomographic gauge, a property that will play an important role further on. 
%in later techniques. 
The particular effect vectors arising from the fit are not gauge-invariant, but the linear identities among these are. This is why tomographic completeness is the condition that allows operational identities to be inferred from experimental data using the cut-to-your-cloth technique, and to be enforced for the secondary procedures in the secondary-procedures technique. \blk

% \subsection{Noncontextuality test with primary statistics}
\subsection{The cut-to-your-cloth technique}

\label{sec: newIVA}
We now outline the main steps of our first scheme to address the fact that the experimentally realized measurement processes differ from the targeted ones.

\noindent\textit{(a) Extract the operational identities from the primary statistics:}\par 
In this first approach, we begin by determining the operational identities satisfied by the primary multi-measurements. Concretely, we define:
\begin{subequations}
\begin{align}
\mc{O}_{\rm all}(\msf{N}^{\text{pri}})
 &:=\{ \vec \alpha\;
 |\sum_{a,x} \alpha_{a,x}N^{A,\text{pri}}_{a|x}=\mbb{0}\}, \\ 
 \mc{O}_{\rm all}(\msf{M}^{\text{pri}})&:=\{
 \vec \beta\;
|\sum_{b,y}\beta_{b,y} M^{B,\text{pri}}_{b|y}=\mbb{0}\},
\end{align}
\label{op:meas}
\end{subequations} 
capturing the sets of linear relations that hold among the primary effects for Alice and Bob, respectively. Importantly, we note that these operational identities can be determined entirely from the primary statistics $p^{\text{pri}}(ab|xy)$. This follows because different tuples $(\msf N^{\text{pri}}, \msf M^{\text{pri}}, \rho^{AB})$ that yield the same $p^{\text{pri}}(ab|xy)$ are related by an invertible linear map on system $A$ and one on system $B$ (i.e., a gauge freedom~\cite{Nielsen2021}) under the stated sufficient condition of \textit{tomographic completeness} (see also Appendix~\ref{sec: methodC1}), and these invertible linear maps left the operational identities unchanged. Therefore, one can obtain the operational identities $\mc{O}_{\rm all}(\msf{N}^{\text{pri}})$ and $ \mc{O}_{\rm all}(\msf{M}^{\text{pri}})$ by determining the linear constraints directly from $p^{\text{pri}}({ab|xy})$.
\begin{subequations}
\begin{align}
\mc{O}_{\rm all}(\msf{N}^{\text{pri}})
 &:=\{ \vec \alpha\;
 |\sum_{a,x} \alpha_{a,x}p^{\text{pri}}({ab|xy})=0,\forall b,y\}, \\ 
 \mc{O}_{\rm all}(\msf{M}^{\text{pri}})&:=\{
 \vec \beta\;
|\sum_{b,y}\beta_{b,y} p^{\text{pri}}({ab|xy})=0,\forall a,x\}.
\end{align}
\label{op:meas-altr}
\end{subequations}
In other words, $\mc{O}_{\rm all}(\msf{N}^{\text{pri}})$ and $ \mc{O}_{\rm all}(\msf{M}^{\text{pri}})$ defined above are gauge-independent quantities. 

\noindent\textit{(b) Derive noncontextuality inequalities:}\par
$\mc{O}_{\rm all}(\msf{N}^{\text{pri}})$ and $ \mc{O}_{\rm all}(\msf{M}^{\text{pri}})$ are in general more complex than the operational identities that hold for multi-measurements that have the intrinsic geometry of a highly symmetric polytope, as in the examples of Sec.~\ref{Sec:entanglement-certification}. Nevertheless, they can still be used to efficiently certify entanglement by simply deriving the explicit noncontextuality inequalities relative to $\mc{O}_{\rm all}(\msf{N}^{\text{pri}})$ and $ \mc{O}_{\rm all}(\msf{M}^{\text{pri}})$. Indeed, the randomized multi-measurement considered in Sec.~\ref{sec:random} lacked symmetry but nonetheless worked efficiently for certifying entanglement. 

\noindent\textit{(c) Test the noncontextuality inequalities on the primary statistics:}\par
Finally, we apply the noncontextuality inequalities in step~(b) to the primary statistics $p^{\text{pri}}(ab|xy)$ obtained in step (a). If the primary statistics violate any of these inequalities, then the underlying state is entangled.

% \subsection{Restoring the exact operational identities and Noncontextuality test with secondary statistics} 
\subsection{The secondary procedures technique} 
\label{sec: newIVB}

The second analysis implements a computation that simulates a \textit{local} classical post-processing of the primary statistics $p^{\text{pri}}(ab|xy)$, the result of which is termed the \emph{secondary} statistics, denoted $p^{\text{sec}}(ab|xy)$. The post-processing is chosen to ensure that the secondary statistics satisfy \textit{exactly} a specific set of operational identities. The noncontextuality inequalities derived from the specific set of operational identities can then be applied directly to $p^{\text{sec}}(ab|xy)$. 

Since the secondary statistics are obtained from the primary ones through a simulation of local classical processing, which can be conceptualized as injecting noise, something that can only make the operational statistics easier to model noncontextually. Hence, any violation of a noncontextuality inequality by the secondary statistics is sufficient to imply the NCOM-nonrealizability of the primary statistics and hence it is also sufficient to certify entanglement in the bipartite state.\footnote{\cmt 
Suppose $\mc B^{\rm pri}$ denotes the realized Bell circuit and $\mc B^{\rm sec}$ denotes the image of this under the local classical post-processing, if the primary statistics are NCOM-realizable relative to $\mc{O}_{\rm aattt}(\mc B^{\rm pri})$, then the secondary statistics must be NCOM-realizable relative to $\mc{O}_{\rm aattt}(\mc B^{\rm sec})$, since local classical processing can be absorbed into a redefinition of the response functions in Eq.~\eqref{eq:cl-bipartite} \blk} 
Note that such certification depends only on the observed data and the choice of target operational identities. This approach is a generalization of the secondary procedures technique introduced in Ref.~\cite{Mazurek2016} for testing NCOM-realizability in the prepare-measure circuit. Here we summarize the main steps and leave the details to Appendix~\ref{sec: methodC}.
\\
\noindent\textit{(a') Derive noncontextuality inequalities:}
\par 
For the sets of operational identities $\mathcal{O}_{\rm all}(\msf N)$ and $\mathcal{O}_{\rm all}(\msf M)$ satisfied by the multi-measurement  $\msf N$ and $\msf M$ in the targeted Bell circuit $\mc B$, determine the noncontextuality inequalities relative to $\mc O_{\rm aattt}(\mc B)$. Examples of these inequalities are given in Sec.~\ref{secIIIA}, and the tools for their derivation are detailed in Appendix~\ref{sec: methodB}.
\\
\noindent\textit{(b') Construct secondary statistics:}\par
Given the primary statistics $p^{\text{pri}}(ab|xy)$ from self-consistent tomography, and a set of targeted operational identities  $\mathcal{O}_{\rm all}(\msf N)$ and $\mathcal{O}_{\rm all}(\msf M)$ from step~(a'), solve for the {\em secondary} multi-measurements, denoted $\msf N^{\text{sec}}=\{\{N^{A,\text{sec}}_{a|x}\}_a\}_x$ and $\msf M^{\text{sec}}=\{\{M^{B,\text{sec}}_{b|y}\}_b\}_y$ as follows:
\begin{itemize}
    \item take each effect to be an element of the convex hull of the primary effects, i.e., 
\begin{subequations}
\begin{align}
&N^{A,\text{sec}}_{a|x} = \sum_{a',x'} u_{a',x'}^{a,x} N^{A,\text{pri}}_{a'|x'},\\
&M^{B,\text{sec}}_{b|y}=\sum_{b',y'}v_{b',y'}^{b,y}M^{B,\text{pri}}_{b'|y'}, \\
\text{such that} \quad &\sum_aN^{A,\text{sec}}_{a|x}=\mbb{1}^A~~\sum_bM^{B,\text{sec}}_{b|y}=\mbb{1}^B~\forall x, y, \notag 
\end{align}
\label{eq:sec-coeff}
\end{subequations}
with $\sum_{a',x'}u_{a',x'}^{a,x}=\sum_{b',y'}v_{b',y'}^{b,y}=1~\forall a,x,b,y$ and $ u_{a',x'}^{a,x},v_{b',y'}^{b,y}\ge 0$;

\item require the effects to satisfy the set of operational identities that are respected by the targeted multi-measurements, i.e., impose the constraints that $\mathcal{O}_{\rm all}(\msf N^{\text{sec}})=\mathcal{O}_{\rm all}(\msf N)$ and $\mathcal{O}_{\rm all}(\msf M^{\text{sec}})=\mathcal{O}_{\rm all}(\msf M)$, or equivalently
\begin{subequations}
\begin{align}
&\sum_{ax}\alpha_{ax}N^{A,\text{sec}}_{a|x}=\mbb{0}^A~~~~\forall \vec{\alpha}\in \mathcal{O}_{\rm all}(\msf N), \\
&\sum_{by}\beta_{by}M^{B,\text{sec}}_{b|y}=\mbb{0}^B~~~~\forall \vec{\beta}\in \mathcal{O}_{\rm all}(\msf M);
\end{align}
\end{subequations}

\item in a variation over the secondary effects, find those that maximize a measure of closeness with the primary effects, namely, those that maximize the objective functions\begin{subequations}
\begin{align}
&C_N = \frac{1}{o\Delta_{\msf N}} \sum_{a,x} u_{a,x}^{a,x},\\
&C_M= \frac{1}{o\Delta_{\msf M}} \sum_{b,y}v_{b,y}^{b,y},
\end{align}
\label{eq:maxi-secondary}
\end{subequations}
where $u_{a,x}^{a,x}$ and $v_{b,y}^{b,y}$ are defined in Eq.~\eqref{eq:sec-coeff}.
% where $o$ is the number of outcomes and $\Delta_{\msf N},\Delta_{\msf M}$ are the number of settings for $\msf{N}$ and $\msf{M}$.
\end{itemize}

At first glance, it may seem to be the case that one needs to specify a choice of tomographic gauge for the primary measurements,
to obtain the coefficients $\{u^{a,x}_{a',x'}\}$ and $\{v^{b,y}_{b',y'}\}$. However, as discussed in Appendix~\ref{sec: methodC1}, this is not the case.  Indeed, assuming \textit{tomographic completeness}\cmt, all possible choices of gauge yield primary measurements with the same operational identities since they are related by an invertible linear map (i.e., operational identities are gauge-independent). It follows that the coefficients $\{u^{a,x}_{a',x'}\}$ can be inferred purely from the {\em primary statistics}, using the following linear program:
\begin{align}
&\max\quad\quad  C_N= \frac{1}{o\Delta_{\msf N}} \sum_{ax}u^{ax}_{ax} \\
&\text{such that} \notag \\
&\sum_{axa'x'}p^{\text{pri}}(a'b'|x'y')u_{a',x'}^{a,x}\alpha_{a,x}=0~~ \forall b',y',  \forall \vec{\alpha}\in \mathcal{O}_{\rm all}(\msf N) \notag \\
&\sum_{aa'x'}p^{\text{pri}}(a'b'|x'y')u_{a',x'}^{a,x}=p^{\text{pri}}(b'|y')~~ \forall x, b',y'.\notag \\
&\sum_{a'x'}u_{a',x'}^{a,x}=1,~~u_{a',x'}^{a,x}\ge0~~~~\forall a,x,a',x', \notag
\label{eq:lp-secondar}
\end{align}
where $p^{\text{pri}}(b'|y'):=\sum_{a'}p^{\text{pri}}(a'b'|x'y')$ is independent of $x'$ since the primary statistics are nonsignaling. 
One can obtain the coefficients $v_{b',y'}^{b,y}$ in an analogous fashion.

\begin{figure*}[t]
    \centering
    \includegraphics[width=0.78\linewidth]{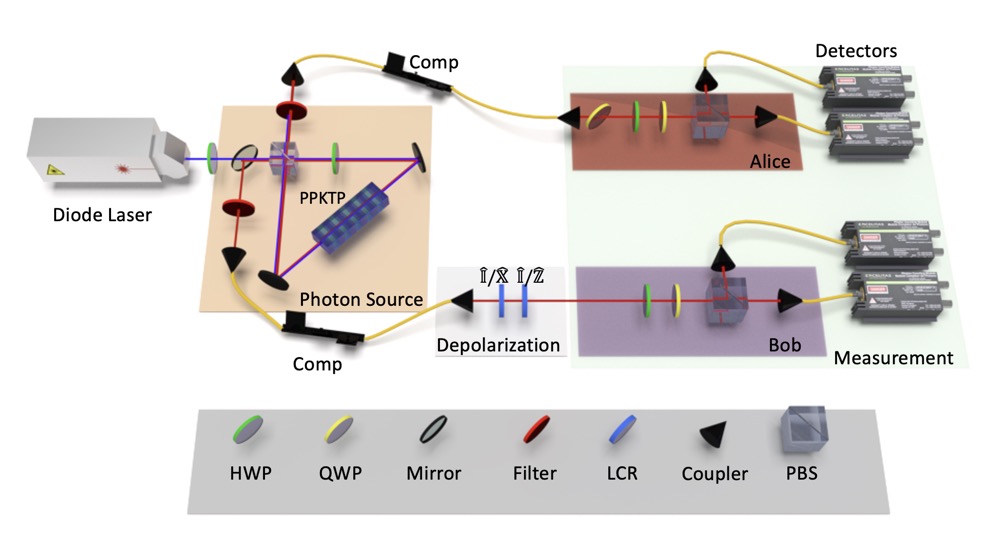}
    \caption{Schematic of {the} experiment. Polarization-entangled photon pairs are generated via parametric down-conversion in a Sagnac interferometer~\cite{Kim2006}. {Polarization controllers (Comp) compensate polarization rotations in the single-mode fibres.  A pair of computer-controlled liquid crystal phase retarders (LCR) aims to implement arbitrary strength depolarizing channels on one of the photons.} Each photon is then directed to {separate polarization analyzers--Alice and Bob--where polarization measurements are performed using half-wave plates (HWPs) and quarter-wave plates (QWPs)}. Coincidence counts are recorded. PPKTP: periodically-poled potassium titanyl phosphate; PBS: polarizing beamsplitter; $\hat{\mbb{1}}$, $\hat{\mbb{X}}$, $\hat{\mbb{Z}}$: {Pauli} gates.} 
    \label{fig:diagram}
\end{figure*} 

The \textit{secondary statistics} $p^{\text{sec}}(ab|xy)$ are then defined as follows
\begin{align}
p^{\text{sec}}(ab|xy)&=\tr[N^{A,\text{sec}}_{a|x}\otimes M^{B,\text{sec}}_{b|y}\rho^{AB}] \notag \\
&=\sum_{a'b'x'y'}v_{b',y'}^{b,y}u_{a',x'}^{a,x}p^{\text{pri}}(a'b'|x'y').
\end{align}

One can check that the secondary statistics are normalized and nonnegative by noting that the secondary measurements $\{\{N^{A,\text{sec}}_{a|x}\}_a\}_x$ and $\{\{M^{B,\text{sec}}_{b|y}\}_b\}_y$ are enforced to be valid POVMs. \\
\noindent\textit{(c) Test the noncontextuality inequalities on  the secondary statistics.}\par
Finally, we apply the noncontextuality inequalities from step (a') to the secondary statistics $p^{\text{sec}}(ab|xy)$ in step (b'). A violation of any such inequality by the secondary statistics witnesses entanglement in the underlying bipartite state. 

\subsection{Comparing the two techniques}
\label{sec:IVD}
\cmt

A disadvantage of the secondary procedures technique, relative to the cut-to-your-cloth technique, is that the former may fail to witness NCOM-nonrealizability in cases where the latter succeeds. This is because, as noted above, post-processing can only make statistics easier to model noncontextually. Consequently, it may happen that the secondary statistics are NCOM-realizable even though the primary statistics are not.

On the other hand, an advantage of the secondary procedures technique over the cut-to-your-cloth technique is that it is much more straightforward to make comparisons across different experimental runs, for instance, runs where the bipartite state is varying or where the primary measurements are varying (due to drift, for instance). In the secondary procedures technique, one enforces a particular set of operational identities to be satisfied by the secondary multi-measurements, independent of the primary multi-measurements or the bipartite state, so the noncontextuality inequality one is testing is also independent of these. Moreover, in this case, one can resample the primary statistics to obtain an estimate of the statistical error in the degree of violation of the fixed noncontextuality inequality. Consequently, the magnitude of the violation of a given inequality can be used to compare different bipartite entangled states, as we demonstrated in Sec.~\ref{sec:exp-sec} for the
family of two-qubit isotropic states.

By contrast, in the cut-to-your-cloth technique, variations in the primary multi-measurements across experimental runs imply variation in the set of operational identities satisfied by these, and hence in the noncontextuality inequality being tested. Variations in the primary multi-measurements might arise from drift (for instance, drift between runs probing different bipartite states) or from fluctuations in the raw frequencies. This variation precludes an even-handed comparison of violation magnitudes across different bipartite states and makes it infeasible to assign statistical error bars to the violation, even for repeated trials on the same state. (We return to this issue in Sec.~\ref{sec:exp-pri}, where we introduce a noise-robustness-based quantifier for the degree of violation which allows a comparison between different experimental runs in the cut-to-your-cloth technique.)
\blk 
\section{Experimental Results}
\label{Section:experiment}
\begin{figure*}[t]
    \centering    \includegraphics[width=0.7\linewidth]{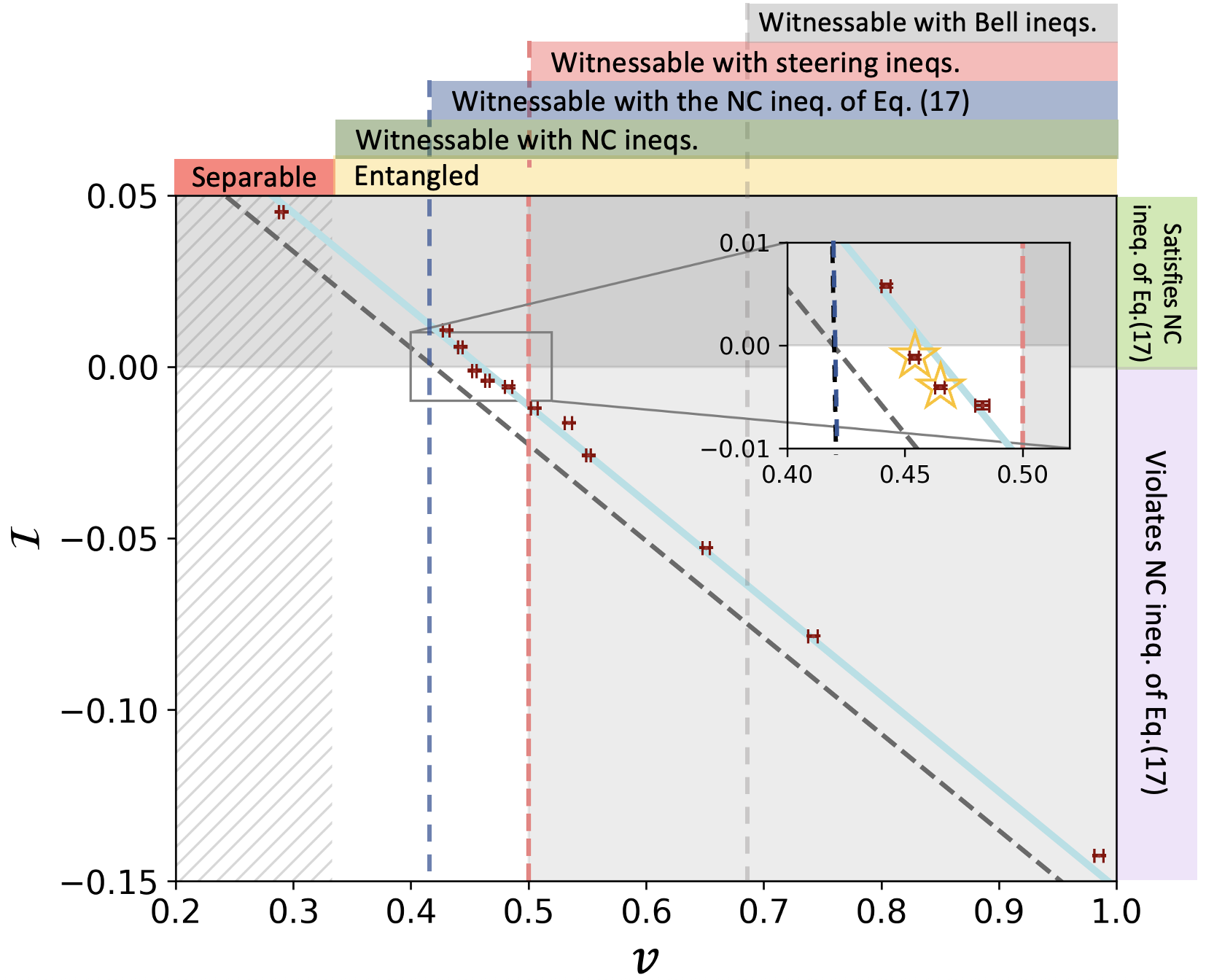}
    \caption{Experimental violation of the noncontextuality (NC) inequality  of Eq.~\eqref{eq: Inequality}, i.e., $\mc{I}\ge 0$, as a function of the isotropic-state visibility $v$.
    Our experimental data are shown as points with vertical and horizontal error bars, while the ideal case is shown as the diagonal dashed gray line, and \yujie the noisy case (i.e., the ideal case modified by the additional \blk $2\%$ white noise in the measurements) is shown as a solid blue line. 
    On top of the diagram, we depict the separable-entangled divide, as well as the parameter regimes certifiable by various techniques: general noncontextuality inequalities, the specific noncontextuality inequality derived in Eq.~\eqref{eq: Inequality}, steering inequalities, and Bell inequalities. 
    The region in white contains \yujie isotropic states \blk that simultaneously i) are entangled, ii) are not certifiable by steering (or Bell) inequalities, and iii) are certifiable by noncontextuality inequalities. Three of our experimental data points lie in this white region, \yujie and the enlarged inset distinguishes the two reconstructed states that are independently certified to be unsteerable in Appendix~\ref{sec: methodE} (hollow yellow stars). (By contrast, the unsteerability of the third state was not confirmed.) \blk
    Error bars for both $\mc{I}$ and $v$ are obtained from 100 trials of Monte Carlo simulation and are depicted. (See Table~\ref{tab:tomoresults} for detailed data.) \blk}
    \label{fig:experiment}
\end{figure*}
Our experiment tests the noncontextuality inequality of Eq.~\eqref{eq: Inequality} for various choices of the value of the visibility $v$ in the 1-parameter family of 2-qubit states. It uses the photonic setup depicted in Fig.~\ref{fig:diagram}. As discussed in depth in Appendix~\ref{sec: methodD},  we generate polarization-entangled photons via type-II spontaneous parametric down-conversion in a Sagnac interferometer. These photons are subsequently directed through a controlled depolarization channel, generating probabilistic mixtures of the four Bell states as outlined in the Appendix~\ref{sec: methodD1}. This allows for the preparation of different mixed states closely approximating (with about  $99\%$ state fidelity) 2-qubit isotropic states with different values of the visibility $v$ (see Eq.~\eqref{eq:isotropic}). By varying $v$, we can prepare separable, entangled, steerable, and nonlocal states~\cite{Werner1989, Wiseman2007, zhang2024,Renner2024}.

In the experiment, we prepared 12 different quantum states $\rho_{\text{exp}}$, each with a measured fidelity to the nearest {isotropic} state $\rho_{\rm iso}^v$ consistently near $99\%$, and remain above $98\%$ (see Table~\ref{tab:tomoresults} in Appendix~\ref{sec: methodD}). The value of the visibility $v$ will be used to represent the state we prepared. The NC inequality of Eq.~\eqref{eq: Inequality} requires Alice's and Bob's multi-measurements to have the intrinsic geometries of an icosahedron and a dodecahedron, respectively.  We target a regular icosahedron and a regular dodecahedron inscribed in the Bloch sphere for Alice and Bob's multi-measurements, respectively, as depicted in Fig.~\ref{fig:plantoic} (see Eq.~\eqref{12-20meas} in the Appendix~\ref{sec: methodB} for a precise definition). This corresponds to implementing 6 projective binary-outcome measurements on Alice's photon and 10 on Bob's. These measurements clearly satisfy the relevant operational identities, and are predicted to violate the NC inequality of Eq.~\eqref{eq: Inequality}. They are implemented using a combination of half-wave plates (HWP), quarter-wave plates (QWP), and polarization beam splitters (PBS).

For each of the 12 states, we implemented 60 distinct pairs of measurement settings (6 on Alice and 10 on Bob). In 2 hours of total run-time for the experiment, each pair of settings accumulated on the order of $2\times10^6$ twofold coincidence counts (consisting of four coincidence pairs), which are denoted as the raw count statistics $N_{ab|xy}$.

Imperfections in experimental implementation, such as polarization optics quality, fiber drift, and laser power drift, are unavoidable. Additionally, the raw count statistics ${N_{ab|xy}}$ represent finite samples from the true probability distribution, and directly normalizing them may yield a distribution that is not quantumly realizable (indeed, possibly not even realizable within a generalized probabilistic theory (GPT)~\cite{hardy2001quantumtheoryreasonableaxioms,barrett2006}). For example, statistical fluctuations can lead to violations of the no-signaling conditions. To address this, the raw counts $N_{ab|xy}$ are converted to relative frequency $f_{ab|xy}$ and these are fitted to a probability distribution ${p^{\text{pri}}(ab|xy)}$ \cmt using self-consistent tomography fit described in Sec.~\ref{Sec:inexact-identities} using a total weighted least squares method~\cite{Krystek2007}.  

From ${p^{\text{pri}}(ab|xy)}$, entanglement certification can be achieved using either of the two different techniques discussed in Sec.~\ref{Sec:inexact-identities}. We consider each in turn. 

\subsection{Certifying entanglement using the secondary procedures technique}
\label{sec:exp-sec}
Given the primary statistics obtained from the self-consistent tomography, we construct secondary statistics ${p^{\text{sec}}(ab|xy)}$ via the technique discussed in Sec.~\ref{sec: newIVB}. For the quantities $C_N= \frac{1}{o\Delta_{\msf N}} \sum_{a,x}u_{a,x}^{a,x}$ and $C_M= \frac{1}{o\Delta_{\msf M}} \sum_{b,y}v_{b,y}^{b,y}$, \cmt defined in~Eq.~\eqref{eq:maxi-secondary}\blk, we achieve, experimentally, an average of $0.98$ and $0.96$ respectively (see Table~\ref{tab:computed}), which confirms that the secondary statistics remain very close to the primary statistics (the simulated local post-processing effectively introduces $\sim 2\%$ white noise to the measurements.) \blk

The values of the quantity $\mc{I}$ appearing on the left-hand-side of the noncontextuality inequalities of Eq.~\eqref{eq: Inequality}, computed from the secondary statistics ${p^{\text{sec}}(ab|xy)}$ for different values of the visibility $v$ of the prepared two-qubit isotropic states are presented in Fig.~\ref{fig:experiment}. Error bars are obtained via Monte Carlo simulations assuming Poisson count statistics. \cmt In these simulations, Poissonian fluctuations in the raw counts are also propagated through the steps of the secondary procedures technique so that different primary statistics lead to different secondary statistics, all of which are evaluated against the same set of exact operational identities and noncontextuality inequality, as noted in Sec.~\ref{sec:IVD}. \blk

Violations of the inequality in Eq.~\eqref{eq: Inequality} are observed for all set-ups that correspond to isotropic states with 
$v >0.43$, including those with $v < 0.5$, for which it is known that their entanglement cannot be certified by steering inequalities ~\cite{zhang2024parallel}. 

Our experimentally prepared states achieve fidelities near $99\%$ with the nearest isotropic state. However, this only guarantees that, for instance, when $v\le 1/2$, they are {\em close to} unsteerable states. It does not guarantee that they are unsteerable (because it remains an open question whether proximity to an unsteerable isotropic state implies unsteerability). 
To address this, we directly assess whether each bipartite state prepared in our experiment is unsteerable (i.e., admits a local‐hidden‐state model) by solving a semidefinite program (see Appendix~\ref{sec: methodE}). 
This analysis showed that the state associated to parameter value $v =0.483$ cannot be confirmed to be unsteerable, but that the states associated to parameter values $v =0.454$ and $v =0.465$ \textit{are} confirmed to be unsteerable (see discussion in Appendix~\ref{sec: methodE}).
This establishes that our noncontextuality-based technique can certify entanglement for two-qubit states that are unsteerable and hence whose entanglement cannot be certified by steering inequalities or Bell inequalities.

\subsection{Certifying entanglement using the cut-to-your-cloth technique}
\label{sec:exp-pri}
\cmt
As noted in Sec.~\ref{sec:IVD}, because the primary statistics $p^{\text{pri}}(ab|xy)$ and the primary measurements typically vary across experimental runs due to small systematic drifts and (in our simulations) Monte Carlo statistical uncertainty, one cannot certify entanglement using a single fixed noncontextuality inequality in the cut-to-your-cloth technique. Consequently, we cannot plot the degree of violation of a fixed noncontextuality inequality as the bipartite state varies, or report an error bar for the degree of violation. 
Using the cut-to-your-cloth technique, therefore, we cannot produce the sort of analysis provided in Fig.~\ref{fig:experiment} or Sec.~\ref{sec:exp-sec} for witnessing entanglement. 

To address this, we propose an alternative way of comparing the degree of NCOM-nonrealizability as we vary the bipartite state. Specifically, we quantify the \emph{noise-robustness} of $\mathbf{p}^{\text{pri}}=\{p^{\text{pri}}(ab|xy)\}_{abxy}$ by asking how much noise needs to be added to make the correlations compatible with a noncontextual ontological model. Concretely, we define
\begin{align}
&\eta_c[\mathbf{p}^{\text{pri}}]:= \min\Bigl\{\eta \in [0,1] \;\Big|\; p^{\text{pri}}_\eta(ab|xy) \in \mc{NC}_{\rm aattt}(\mc B^{\text{pri}})\Bigr\},\notag \\
&p^{\text{pri}}_\eta(ab|xy) := (1-\eta) p^{\text{pri}}(ab|xy) + \eta d^{\text{pri}}(ab|xy) 
\label{eq:noise-robust}
\end{align}
where $d^{\text{pri}}(ab|xy)$  is defined as the statistics obtained by averaging over all extreme vertices of  $\mc{NC}_{\rm aattt}(\mc{B^{\text{pri}}})$. We note that $d^{\text{pri}}(ab|xy)$ obtained from our operational statistics is very close to uniform, but is not, in general, exactly uniform over outcomes, since the uniform distribution need not belong to $\mc{NC}_{\rm aattt}(\mc B^{\text{pri}})$. It will be useful to also consider the noise-robustness of $p^{\text{sec}}(ab|xy)$. We denote this as $\eta_c[\mathbf{p}^{\text{sec}}]$, and define it as the obvious analogue of Eq.~\eqref{eq:noise-robust}.

\begin{table}[htbp]
\cmt 
\centering
\caption{ Critical noise thresholds $\eta_c$ for 12 states with estimated isotropic visibility $v$. The second and third columns give $\eta_c$ obtained from the primary and secondary statistics, respectively. The error bar is calculated from 100 trials of Monte
Carlo simulation with the experimental statistics. The second column reveals that noise-robustness can indeed provide a means of comparing the degree of NCOM-nonrealizability of bipartite states with different values of $v$ in the cut-to-your-cloth technique, in spite of the fact that the noncontextuality inequality  being tested varies with $v$.  A comparison of the second and third columns 
%$\eta_c\bigl[p^{\text{pri}}(ab|xy)\bigr]$ and $\eta_c\bigl[p^{\text{sec}}(ab|xy)\bigr]$ columns shows
shows that \cmt the secondary statistics consistently yield a lower $\eta_c$ (i.e., they are less noise-robust) than the primary statistics, since the construction of secondary statistics from the primary statistics involves local classical post-processing, which introduces additional noise.}
\begin{ruledtabular}
\begin{tabular}{ccc}
Visibility $v$ & $\eta_c[\mathbf{p}^{\text{pri}}]$ & $\eta_c[\mathbf{p}^{\text{sec}}]$ \\
\hline
$0.290\pm 0.002$ & 0 & 0\\
$0.441\pm 0.002$ & 0 & 0 \\
$0.430\pm 0.002$ & 0 & 0 \\
$0.454\pm 0.002$ & $0.0413 \pm 0.0005$ & $0.0045 \pm 0.0009$\\
$0.465\pm 0.002$ & $0.0589 \pm 0.0004$ & $0.0336 \pm 0.0006$\\
$0.483\pm 0.003$ & $0.0993 \pm 0.0002$ & $0.0472 \pm 0.0014$\\
$0.505\pm 0.003$ & $0.1397 \pm 0.0003$ & $0.0929 \pm 0.0008$\\
$0.534\pm 0.002$ & $0.1803 \pm 0.0002$ & $0.1213 \pm 0.0004$\\
$0.551\pm 0.003$ & $0.2166 \pm 0.0003$ & $0.1796 \pm 0.0005$\\
$0.651\pm 0.003$ & $0.3368 \pm 0.0001$ & $0.3092 \pm 0.0004$\\
$0.742\pm 0.003$ & $0.4168 \pm 0.0001$ & $0.3995 \pm 0.0004$\\
$0.982\pm 0.004$ & $0.5613 \pm 0.0001$ & $0.5470 \pm 0.0002$\\
\end{tabular}
\end{ruledtabular}
\label{tab:newresults}
\end{table}
\cmt
\subsection{Comparing the cut-to-your-cloth and secondary-procedures techniques}

Table~\ref{tab:newresults} demonstrates the improved noise-robustness of the cut-to-your-cloth technique relative to the secondary procedures technique, hence the greater sensitivity to NCOM-nonrealizability that is offered by the cut-to-your-cloth technique. 

On the other hand, if one wishes to compare the degree of NCOM-nonrealizability across different experimental runs, e.g., runs implementing different two-qubit isotropic states, then although noise-robustness provides a means of doing so for both techniques (as illustrated in Table~\ref{tab:newresults}), it is arguably less natural than being able to compare the degree of violation of a single noncontextuality inequality (just as we compare the degree of violation of a single Bell inequality when witnessing entanglement via Bell inequality violation).  The latter approach, the outcome of which is illustrated in Fig.~\ref{fig:experiment}, is only available for the secondary procedures technique.   

\blk

\yujie
\section{Loopholes in the noncontextuality test}

The preceding sections describe how noncontextuality inequalities can certify entanglement and how the required operational identities can be inferred from experimental data. We now discuss in detail the conditions under which this inference is valid. 

\label{sec:Loopholes}
\subsection{The tomographic incompleteness loophole in noncontextuality tests
}
\label{sec: tomo-loopholes}
The notion of \emph{tomographic completeness} is the core %sufficient 
condition behind our noncontextuality-based technique for entanglement certification. It plays the same role here that it does in \emph{self-consistent} tomography, including GPT tomography~\cite{Mazurek2021} and gate-set quantum tomography~\cite{Nielsen2021}: one assumes an \emph{effective finite-dimensional vector space} (in the case of GPT tomography, one assumes a finite-dimensional 
GPT state space), and one requires that the \emph{implemented} preparations and measurement effects span this vector space and its dual, respectively. 

\paragraph{From standard quantum state tomography to self-consistent quantum tomography.}
In standard quantum state tomography, one typically assumes that the measurements have already been characterized, in the sense of the POVMs representing them being known a priori. The set of measurement effects is termed tomographically complete if the linear span is the full operator space associated with the (assumed finite-dimensional) Hilbert space. % Concretely, tomographic completeness of the set of steered states ensures that two different linear combinations of effects in $\msf{N}$ yield the same probabilities on the set of steered states, i.e., on $V_A$, if and only if they yield the same probabilities on \textit{all} states, hence if and only they are equal. This is the step that allows us to formulate the required operational identities (and hence the noncontextual constraints) directly at the level of the observed statistics. 

In self-consistent quantum tomography, no device is taken as characterized a priori.
%, and \textit{tomographic completeness} is not a property of the measurements alone.
In a prepare–measure experiment on a quantum system with Hilbert space dimension $D$, one works in the operator space of dimension $K=D^2$ (e.g., qubit quantum theory has a model dimension $K=4$). Tomographic completeness then means that the implemented set of states and the implemented set of measurement effects each span the $K$-dimensional vector space~\cite{Mazurek2021}. (The same idea underlies gate-set tomography: one needs a set of experimentally available preparations and measurements that are tomographically complete for the assumed model space~\cite{Nielsen2021}.) When this holds, the reconstructed description is unique only up to an overall invertible linear change of representation (gauge freedom); when it fails, some directions in the vector space of Hermitian operators cannot be accessed from the data, and the reconstruction becomes underdetermined.

\paragraph{Tomographic completeness in a Bell circuit.}
In a bipartite Bell circuit with a state $\rho^{AB}$ on $\mc{H}^A\otimes \mc{H}^B$, local POVMs
$\msf{N}^A=\{\{N^A_{a|x}\}_a\}_x$ on $A$ and $\msf{M}^B=\{\{M^B_{b|y}\}_b\}_y$ on $B$,
each party’s measurements induce a family of \emph{subnormalized steered states} $\{\{\tilde{\rho}^B_{a|x}\}_a\}_x$ and $\{\{\tilde{\rho}^A_{b|y}\}_b\}_y$ for the other party, as in Eq.~\eqref{eq:steer all}. We can define the corresponding steered-state spans (real vector spaces of Hermitian operators)
\begin{subequations}
\begin{align}
V_A &\coloneqq \mathrm{span}\big(\{\tilde{\rho}^A_{b|y}\}_{b,y}\big)\subseteq \mathrm{Herm}(\mc{H}^A),\\
V_B &\coloneqq \mathrm{span}\big(\{\tilde{\rho}^B_{a|x}\}_{a,x}\big)\subseteq \mathrm{Herm}(\mc{H}^B).
\end{align}
\label{eq:steering-spans}
\end{subequations}
One can then formally define the notion of tomographic completeness in a Bell circuit as follows.

\begin{definition}
The Bell circuit $(\msf{N},\msf{M},\msf P)$ is \emph{tomographically complete} if Alice’s effects separate all states in $V_A$, and Bob’s effects separate all states in $V_B$. Equivalently, for all $X^A\in V_A$ and $Y^B\in V_B$ 
\begin{subequations}
\begin{align}
\tr[N^A_{a|x} X^A]=0,~\forall a,x
&\Longrightarrow
X^A=\mbb 0^A \\
\tr[M^B_{b|y} Y^B]=0,~\forall b,y 
&\Longrightarrow
Y^B=\mbb 0^B.
\end{align}
\end{subequations}
\end{definition}
In other words, Alice’s effects, when restricted to $V_A$, span the dual space $V_A^*$, and Bob’s effects, when restricted to $V_B$, span the dual space $V_B^*$. 

In a Bell experiment, when the dimensions of the quantum systems are known, e.g., $\dim(\mc H^A)=\dim(\mc H^B)=D$, we always have $\dim(V_A)\le D^2$ and $\dim(V_B)\le D^2$. Moreover, a generic bipartite state $\rho^{AB}$ defines steered-state spans $V_A$ and $V_B$ (by Eq.~\eqref{eq:steering-spans}) such that $\dim(V_A)=\dim(V_B)=D^2$. Only on a measure-zero subset of states is it the case that $\dim(V_A)<D^2$ or $\dim(V_B)<D^2$~\cite{zhang2024parallel, Plavala2024}. It follows that, in practice, once a finite-dimensional model is specified, tomographic completeness of a Bell circuit can be assessed experimentally rather than assumed as part of a prior trusted characterization of the devices. For a generic bipartite state, this reduces to checking that the implemented local measurements are tomographically complete on the relevant local operator spaces.

% \paragraph{Why this matters for our noncontextuality test.}
% This assumption is used exactly as in self-consistent tomography: it ensures that operational identities among the implemented measurement effects can be inferred \emph{from the observed statistics}, i.e., identities that hold among any primary measurements that realize the operational statistics are gauge-independent, thus can be lifted to be the operational identities of the measurement processes. Concretely, tomographic completeness of the set of steered states ensures that two different linear combinations of effects in $\msf{N}$ yield the same probabilities on the set of steered states, i.e., on $V_A$, if and only if they yield the same probabilities on \textit{all} states, hence if and only they are equal. This is the step that allows us to formulate the required operational identities (and hence the noncontextual constraints) directly at the level of the observed statistics. 

% If tomographic completeness fails, then identities that appear to hold on the basis of the operational statistics need not be genuine operational identities, and therefore should not be imposed when constructing noncontextuality inequalities or performing a noncontextuality test. 

\paragraph{Tomographic incompleteness loophole.} 
The condition of tomographic completeness is used exactly as in self-consistent tomography: it ensures that operational identities among the implemented measurement effects can be inferred \emph{from the observed statistics}, i.e., identities that hold among any primary measurements that realize the operational statistics are the same, and can thus be obtained from the identities that hold among the operational statistics (see e.g., Eq.~\eqref{op:meas-altr}). 

% i.e., identities that hold among any primary measurements that realize the operational statistics are gauge-independent, thus can be lifted to be the operational identities of the measurement processes. 
% Concretely, tomographic completeness of the set of steered states ensures that two different linear combinations of effects in $\msf{N}$ yield the same probabilities on the set of steered states, i.e., on $V_A$, if and only if they yield the same probabilities on \textit{all} states, hence if and only they are equal. This is the step that allows us to formulate the required operational identities (and hence the noncontextual constraints) directly at the level of the observed statistics. 

If the tomographic-completeness condition fails, then the inferred operational identities may fail to be true identities of the physical system. For example, two effects that appear operationally identical on the realized family of steered states may be distinguishable by some unrealized steered state. This is closely related to what has been called the tomographic incompleteness loophole~\cite{Schmid2024add,schmid2024shadowssubsystemsgeneralizedprobabilistic}: generalized noncontextuality is assessed relative to a choice of system, and requires that the implemented preparations and measurements are tomographically complete for that system. If tomographic completness fails, then a noncontextuality inequality need not witness nonclassicality. This loophole is the way in which our method differs most sharply from entanglement certification based on Bell-inequality violations.

\paragraph{How tomographic completeness can be checked in practice.}
In practice, an expectation about the dimension $K$ cannot be verified by 
%exact rank equalities, i.e., by
showing that the observed statistics when arranged as a data table have rank equal to $K$, since finite statistics and noise generically make the data table full-rank. However, a hypothesis about the dimension $K$ can be assessed using the standard train-and-test methodology used in GPT tomography \cite{Mazurek2021}. This is done by ruling out dimensions less than $K$ by showing that a lower-dimensional model cannot fit the statistics, and disfavoring dimensions greater than $K$ because they overfit the test data. 
In our experiment, we expect $K=4$ (i.e., a two-qubit model) to hold. The experimental data are found to be consistent with this expectation, i.e., the operational statistics can be well fit by such a model, where the fitted steered states and measurements span the relevant operator space, and thus are tomographically complete. 

% \paragraph{When can one use our entanglement certification?}
% As emphasized throughout the paper, the sufficient conditions required for our noncontextuality test are essentially the same as those underlying self-consistent tomography. Consequently, our entanglement-certification protocols can be deployed in any experimental setting where self-consistent tomography is already used for characterization. Crucially, all the quantities entering the noncontextuality test, both the observed operational statistics and the operational identities used to construct the relevant noncontextuality inequalities, are gauge-independent. This is the key advantage of our approach over more naive ways of attempting to certify entanglement using self-consistent tomography alone. More generally, we expect noncontextuality-based tests of this kind to provide a practical route for certifying other forms of resourcefulness of quantum processes in experimental platforms where self-consistent tomography is possible~\cite{zhang2024parallel}.

\subsection{Absence of a detector inefficiency loophole in noncontextuality tests}
\label{sec:detector-inefficiency}

Unlike Bell tests, our methods do not suffer from a detector inefficiency loophole.
In a Bell test, inefficient detectors can open a loophole because whether the detector registers an outcome or fails to register any outcome (termed the ``null outcome'') 
may depend on the value of the hidden variable (as well as the local measurement setting). 
Consequently, the full distribution, including null-outcome events, may admit
a local-hidden-variable model even when the distribution conditioned on {\em not} seeing the null-outcome event violates a Bell inequality~\cite{Gisin1999,Massar2003,Larsson2014}. This is why loophole-free Bell tests include null-outcome events explicitly in the
analysis and rely on achieving sufficiently high efficiency in the detectors to close the
loophole, as in the detection-loophole-free and fully loophole-free Bell
experiments of Refs.~\cite{Christensen2013,Giustina2015,Shalm2015}.

In contrast, Refs.~\cite{Selby2023,accessiblefrags} showed that generalized noncontextuality tests in prepare-and-measure circuits have no such detector inefficiency loophole. The basic reason is that
detector inefficiency is represented operationally by \yz attenuating the original effects and adding a null-outcome effect
to the measurement, provided the efficiencies of all detected outcomes are
nonzero. The set of \yz attenuated \yujie effects is the same as in the ideal
experiment, up to strictly positive rescalings. The null-outcome effect is then fixed by the requirement that the effects sum to the unit effect. Strictly positive rescalings do not change the positive cone
generated by a set of effects. Thus, the sets of effects in the ideal and inefficient experiments are cone-equivalent in the sense of Ref.~\cite{Selby2023,accessiblefrags}. Finally, as shown in that work, cone equivalence preserves NCOM-realizability. 

The same reasoning applies to the bipartite Bell circuit considered in this
work. Specifically, once diagram preservation and tomographic completeness of the relevant Bell circuit are granted, detector inefficiency is a 
modification of the implemented measurement procedures. Let Alice's and Bob's
ideal multi-measurements be described by
$\msf{N}:=\{\{N^A_{a|x}\}_{a}\}_{x}$ and
$\msf{M}:=\{\{M^B_{b|y}\}_{b}\}_{y}$, and let
$\mc B=(\msf N,\msf M,\rho^{AB})$ denote the ideal Bell circuit. Detector
inefficiency implies replacing the ideal multi-measurements with
$\tilde{\msf{N}}:=\{\{\tilde N^A_{\alpha|x}\}_{\alpha}\}_{x}$ and
$\tilde{\msf{M}}:=\{\{\tilde M^B_{\beta|y}\}_{\beta}\}_{y}$, where
$\alpha$ and $\beta$ range over the original outcome set, and the new `\yz null\yujie' outcome 
labeled 
$\varnothing$. The \yz attenuated \yujie effects, i.e., those corresponding to $\alpha\ne \varnothing$,  $\beta \ne \varnothing$, are
\begin{subequations}
\label{eq:inefficient-effects}
\begin{align}
    \tilde N^A_{a|x}
    &=
    \epsilon_{a|x}N^A_{a|x}, \label{eq:at-outcome1}\\
    \tilde M^B_{b|y}
    &=
    \epsilon_{b|y}M^B_{b|y},
\end{align}
\end{subequations}
where $0<\epsilon_{a|x},\epsilon_{b|y}\le 1$ for all $a,x,b,y$, and the
corresponding \yz null-outcome \yujie effects are
\begin{subequations}
\label{eq:inefficient-null}
\begin{align}
    \tilde N^A_{\varnothing|x}
    &=
    \sum_a(1-\epsilon_{a|x})N^A_{a|x}, \label{eq:null-outcome1}\\
    \tilde M^B_{\varnothing|y}
    &=
    \sum_b(1-\epsilon_{b|y})M^B_{b|y}.
\end{align}
\end{subequations}
These measurements, together with the bipartite state $\rho^{AB}$, define an inefficient Bell circuit
$\mc B^\epsilon=(\tilde{\msf N},\tilde{\msf M},\rho^{AB})$ with statistics
\begin{align}
    p_\epsilon(\alpha,\beta|x,y)=\tr\left[
        \rho^{AB}\tilde N^A_{\alpha|x}\otimes
        \tilde M^B_{\beta|y}\right].
\end{align}

The fact that there is no detector inefficiency loophole in noncontextuality tests follows from the following result.
\begin{proposition}
\label{prop:detector-inefficiency}
Let $\mc B^\epsilon$ be the inefficient Bell circuit defined above, with
$0<\epsilon_{a|x},\epsilon_{b|y}\le 1$ for all attenuated effects.
Then $p(a,b|x,y)$ is
NCOM-realizable relative to $\mc O_{\rm aattt}(\mc B)$ if and only if
$p_\epsilon(\alpha,\beta|x,y)$ is NCOM-realizable relative to
$\mc O_{\rm aattt}(\mc B^\epsilon)$.
\end{proposition}

The proof is given in Appendix~\ref{app:proof-detector-inefficiency}. The intuition is that the operational identities of the ideal and inefficient measurements are related by an invertible rescaling on the attenuated effects, and consequently so are the sets of response functions that represent these in the noncontextual ontological model. If the distribution $p(a,b|x,y)$ in the Bell circuit $\mathcal{B}$ is NCOM-realizable, then there exists a set of response functions that provides a noncontextual representation of $\msf{N}$ and one that does so for $\msf{M}$ (i.e., the response functions satisfy the same operational identities as the effects).  But then by rescaling these response functions by the same rescaling factors that take $\msf{N}$ to $\tilde{\msf{N}}$ and $\msf{M}$ to $\tilde{\msf{M}}$, we obtain two new sets of response functions that provide noncontextual representations of $\tilde{\msf{N}}$ and $\tilde{\msf{M}}$, while realizing the statistics $p_{\epsilon}(\alpha,\beta|x,y)$. The converse follows by the same argument.

Therefore, the proposition shows that it suffices to have a nonzero detector efficiency to implement our noncontextuality test. The distribution $p_\epsilon(\alpha,\beta|x,y)$ can be analyzed directly, treating null outcomes as ordinary outcomes, while maintaining the same conclusion about NCOM-realizability~\footnote{Note that certification of EPR steering is also possible in a manner that has no detection inefficiency loophole~\cite{Bennet2012}.}.

% \subsection{Scope of applicability}
% \label{Sec:scope}

% The preceding subsections identify the condition under which our noncontextuality-based certification method is valid: tomographic completeness of the Bell circuit. This is what ensures that the operational identities entering the noncontextuality test can be inferred from the observed statistics. Importantly, there is no Bell-style detector inefficiency loophole for the noncontextuality test considered here. 

% As emphasized throughout the paper, these conditions are essentially the same as those underlying self-consistent tomography. Consequently, our entanglement-certification protocol can be deployed in experimental settings where self-consistent tomography is already a suitable characterization tool. Crucially, all quantities entering the noncontextuality test, both the observed operational statistics and the operational identities used to construct the relevant noncontextuality inequalities, are gauge-independent. This is the key advantage of our approach over more naive attempts to certify entanglement directly from a self-consistent tomographic reconstruction. More generally, we expect noncontextuality-based tests of this kind to provide a practical route for certifying other forms of resourcefulness of quantum processes in experimental platforms where self-consistent tomography is possible~\cite{Nielsen2021,Mazurek2021}.

\blk

\section{Comparison with alternative entanglement certification approaches}
\label{Section:comparison}
We now compare our entanglement certification technique to pre-existing techniques. A summary of the relative merits of the different techniques, including ours, can be found in Table~\ref{tab:my_label}.

\subsection{Separability test with quantum state tomography}
Witnessing entanglement by implementing a separability test on a state obtained via quantum state tomography requires prior characterization of one's measurements.  Such a prior characterization is also required for estimating other entanglement witnesses~\cite{Guhne2009}. 

\subsection{Separability test with GPT/gate-set tomography}
By contrast, our approach resembles that of GPT tomography~\cite{Mazurek2021,Grabowecky2022} and gate-set tomography~\cite{Nielsen2021,Proctor2017}, which do not require a prior characterization of any subset of the procedures in the experiment, but rather infer characterizations of all the procedures self-consistently from how they interact, i.e., in a bootstrap fashion. 

Both GPT tomography and gate-set tomography 
do not return a unique output, neither for the bipartite state, nor for the local effects, but rather a \textit{set} of possible triples of these with different possibilities being related by an invertible linear transformation. As noted earlier, following Ref.~\cite{Nielsen2021}, we refer to this nonuniqueness as a tomographic {\em gauge freedom}. The most obvious way to certify entanglement of the bipartite state in this case is to implement a separability test for all possible choices of gauge. Note that merely fixing a tomographic gauge arbitrarily (rather than optimizing over it) can lead to estimates of quantities of interest that may be off by orders of magnitude from the true value~\cite{Proctor2017}.

Like GPT tomography and gate-set tomography, our approach also begins with self-consistent tomography. However, it differs in how we use this output to certify entanglement.  Rather than performing a separability test on the entangled state that is most `entangled' after an optimization over gauge, \cmt we derive the noncontextuality inequality using the gauge-independent operational identities, and we test the noncontextuality inequality on the primary statistics $p^{\text{pri}}(ab|xy)$ (or the secondary statistics $p^{\text{sec}}(ab|xy)$) that is also gauge-independent. \blk  Our technique consequently provides the possibility of certifying entanglement without requiring any optimization over gauge and indeed without even requiring a description of the full gauge freedom of the state and local effects. 

\subsection{Bell inequality test}
Another technique for certifying entanglement that is gauge-independent and requires no prior characterization of the measurements is to use violations of Bell inequalities (e.g., in the device-independent paradigm~\cite{Brunner2014}). As argued in Ref.~\cite{Lin2018}, the rigorous implementation of this approach requires fitting one's raw data to the nearest correlations in the polytope of nonsignaling correlations and then testing Bell inequalities on these correlations. 
These device‐independent tests have an advantage over the technique described here: they do not require the sets of local effects to be tomographically complete. However, the price Bell tests pay for this advantage is that they can only certify entanglement in a strict subset of all entangled states~\cite{Werner1989}, whereas our technique can certify any entangled state. \cmt Furthermore, we demonstrated in Sec.~\ref {sec:fraction} that, for a given number of measurement settings at each wing,  a noncontextuality test can certify entanglement in a significantly larger fraction of the set of entangled states than can be certified in a Bell test. 
Moreover, the computational cost of solving the membership problem in a Bell test grows rapidly with the number of measurement settings, quickly becoming infeasible, as also noted in Sec.~\ref {sec:fraction}, whereas the analogous membership problem for a noncontextuality test remains feasible. \blk

\cmt 
\subsection{Steering inequality test}
An entanglement certification technique that is closely related to the Bell-based technique is to use steering inequalities.~\cite{Cavalcanti2009, Saunders2010,Uola2020}. Like a Bell test, such a test can certify entanglement only for a strict subset of all entangled states (those that are steerable)~\cite{Werner1989,Wiseman2007}. Furthermore,  this technique is only device-independent on one side, i.e., they require a full characterization of a tomographically complete set of measurements on one side. By contrast, our approach can certify all entangled states and can also be adapted to certify entanglement via noncontextuality-based steering certification, without requiring prior characterization of the measurement devices in either task. Examples of steering certification are provided in Appendix~\ref{sec: methodB1}. \blk

\subsection{Inequality test in non-Bell causal network}
A different device-independent approach to verifying entanglement is to test causal compatibility inequalities within more elaborate causal networks that use the bipartite state as one common cause. As shown in Ref.~\cite{bowles2018device}, such methods can in principle certify 
entanglement in {\em any} entangled state, while retaining the advantages of being a gauge-independent technique 
requiring no prior characterization of the measurements. However, these approaches require network structures that are more experimentally complex than that of the Bell circuit.  Moreover, they require \textit{additional} entangled sources beyond the one being characterized  and rely on an extra assumption of source independence. 

\cmt  
\subsection{Entanglement certification and Kochen-Specker contextuality}
Finally, a closely related idea—certifying entanglement by linking local hidden-variable 
models to realizability by a Kochen–Specker (KS) noncontextuality model —was already noted by Werner in~\cite{Werner1989}. However, this route differs in essential ways from our framework, which is based on the notion of generalized contextuality rather than KS contextuality. First, KS contextuality arises only in Hilbert-space dimension $D\ge 3$ and thus does not, even in principle, provide a certification method applicable to the two-qubit states. Second, KS-type tests are typically formulated for, and experimentally interpreted through, the idealization of projective measurements, which limits their direct practical applicability. Most importantly for the present discussion, Werner’s argument was framed under the assumption of access to \emph{all} projective measurements. It is therefore unclear how strong such a certification can be when one is restricted to only a finite set of projective measurements. In this sense, while this idea is conceptually aligned with the broader theme of using contextuality to witness entanglement, it does not by itself amount to a practical entanglement certification protocol.
\subsection{Entanglement certification using our approach}
Compared to a Bell test, which relies on the weakest assumptions among standard entanglement-certification protocols, our noncontextuality-based approach offers no advantage in regimes where Bell-based certification already succeeds. Indeed, as noted in Sec.~\ref{Sec:entanglement-certification}, in such regimes our noncontextuality test in the $\mc{O}_{\rm ttttt}$ scenario is equivalent to the Bell test. Moreover, all operational identities required for the noncontextuality test are then trivially guaranteed by the normalization of the measurement effects and therefore require no additional justification. 

As noted earlier, our noncontextuality-based approach can certify entanglement in regimes where Bell-based certification fails. The
strength of the extra operational conditions that are sufficient for being able to do so depends on the scope of entangled states one wants to certify. Specifically, as discussed in the main text, all entangled states can be certified using noncontextuality inequalities in the $\mc O_{\rm aattt}$ scenario, for which tomographic completeness of both Alice’s and Bob’s measurements is sufficient to infer the relevant operational identities. By contrast, as discussed in Appendix~\ref{sec: methodB1}, if one seeks to certify entanglement for just the set of steerable states, one may instead use noncontextuality inequalities in the $\mc O_{\rm atttt}$ or $\mc O_{\rm tattt}$ scenarios, in which case tomographic completeness of only one party’s measurements, Alice’s or Bob’s respectively, is sufficient to infer the relevant operational identities.

% As discussed above, there do exist other entanglement certification protocols that can, in principle, certify more entangled states than a Bell-based test; some of which can even certify {\em all} entangled states. However, these either require full characterization of the measurement devices, as in quantum state tomography, entanglement witnesses, and steering-based tests, or they require additional assumptions, such as source independence, as in inequality tests in quantum networks. Our entanglement certification protocol is preferable insofar as one can regard tomographic completeness as an assumption that is weaker than these alternatives and more practically verifiable. 
\yujie 
As discussed above, there do exist other entanglement certification protocols that can, in principle, certify more entangled states than a Bell-based test; some of which can even certify {\em all} entangled states.    For instance, one can do inequality tests in quantum networks that are more complex than a Bell circuit~\cite{Buscemi2012,bowles2018device}.  Our entanglement certification protocol offers an advantage over these insofar as a Bell circuit presents less of an experimental burden than does a more general quantum network.  Other  protocols may involve an experimental burden comparable to realizing a Bell circuit, but require prior characterization of (all or some of) the measurement devices, as in quantum state tomography, entanglement witnesses, and steering-based tests.  Our entanglement certification protocol is preferable to these insofar as one can regard tomographic completeness as a weaker condition than that of having a prior characterization of measurement devices. 
\blk

If the dimension of the underlying physical system is known, one can use self-consistent tomography to directly determine whether the measurement operators span the relevant operator spaces.  In this case, tomographic completeness is inferred, not assumed.  Even if the dimension of the underlying system is not known, one can still perform a hypothesis test for the system dimension, as demonstrated in Ref.~\cite{Mazurek2021}, and then assess tomographic completeness relative to that hypothesis. The fact that such an assessment might yield a negative verdict for tomographic completeness means that our approach provides a concrete opportunity for falsifying some prior assumption about the condition of tomographic completeness. Such a negative verdict then flags the fact that one cannot 
implement a noncontextuality-based entanglement certification technique using the experimental procedures that were implemented. 
%Although tomographic completeness may fail in a real experiment, this feature is in fact an advantage of our entanglement certification framework: it  provides a concrete opportunity to falsify the assumptions underlying the entanglement-certification result.
\section{Discussion}
\label{Sec:discussion}
On the theoretical and foundational side, we showed that the noncontextuality inequalities that witness entanglement, steerability, and nonlocality can be derived from applying the principle of noncontextuality to a hierarchy of subsets of the full set of operational identities. This hierarchy also offers an alternative {\em definition} of entanglement, steering, and nonlocality in terms of NCOM-nonrealizability for the measurement-full Bell circuit.

\yujie 
One clear avenue for future theoretical work is to use our methods to identify a hierarchy of sufficient conditions for nonclassicality in bipartite Bell circuits, i.e., tractable noncontextuality inequalities whose violations can witness nonclassicality (in particular, entanglement or steerability). One can also explore the power of noncontextuality inequalities for certifying nonclassicality of a target state in other experimental scenarios beyond the bipartite Bell circuit. For instance, one could consider the certification of multipartite entanglement using multipartite Bell circuits, or investigate possible improvements to entanglement-certification protocols in circuits with more general causal structures, such as quantum networks. Another avenue worth pursuing is to modify the bipartite Bell circuit itself by allowing entangled measurements, or by including states other than the target state whose entanglement is to be certified, thereby making the target state part of a multi-state. Each such modification can introduce additional operational identities that constrain the noncontextual ontological model, and can thereby lead to more stringent conditions for nonclassicality, stronger noncontextuality inequalities, and improved prospects for entanglement certification.

Another direction for future work is to adapt our methods to identify a hierarchy of sufficient conditions for the nonclassicality of more general types of processes, such as quantum channels and instruments, as initially outlined in Ref.~\cite{zhang2024parallel}. 
This could lead not only to new sufficient conditions for witnessing the nonclassicality of various quantum processes, but also to novel ways of {\em quantifying} nonclassicality. 
\blk

On the experimental and technical side, we have shown that our noncontextuality inequalities can be directly translated into a practical entanglement certification protocol. Uniquely among techniques for doing so, ours offers \blk the possibility of realizing all of the following features at once: requiring no prior characterization of the measurement devices, delivering a verdict based on quantities that are independent of the tomographic gauge, being implementable using only the causal structure of a standard bipartite Bell circuit, and being able in principle \blk to witness the entanglement of any entangled state. 

One interesting direction for future experimental work is to use a noncontextuality-based technique to test sufficient conditions for the nonclassicality of multipartite states, channels, and instruments. Another frontier is to examine whether the assumption of tomographic completeness typically required for experimental tests of noncontextuality inequalities can be relaxed~\cite{pusey2019contextualityaccesstomographicallycomplete,schmid2024shadowssubsystemsgeneralizedprobabilistic}. \blk

\section{Acknowledgements}
We thank Y{\`i}l{\`e} Y{\=\i}ng, 
Emanuele Polino, Elie Wolfe, and Yeong-Cherng Liang for useful discussions. \yujie We thank three anonymous referees for helpful comments, including requests for clarification of the assumptions underlying our entanglement certification protocol, which led us to make many improvements in the presentation. \blk The research was supported in part by the Natural Sciences and Engineering Research Council of Canada
(NSERC), the NSERC-Alliance Hyperspace project, Industry Canada, the Canada Foundation for
Innovation (CFI), Canada Excellence Research Chair (CERC) program, and the Canada First Research Excellence
Fund (CFREF). D.S. and R.W.S. were supported by Perimeter Institute for Theoretical Physics. Research at Perimeter Institute is supported in part by the Government of Canada through the Department of Innovation, Science and Economic Development and by the Province of Ontario through the Ministry of Colleges and Universities. 

\section{Code availability}
The code needed to reproduce the main results is available on GitHub~\cite{zhang2025github}.  

% \section{Author contributions}
% Y.Z., D.S., and R.W.S. conceived the theory. Y.Z. developed the theory. Y.Z., J.S., and K.J.R. developed the experiment.  J.S., C.R., and L.J.M. built the experiment and collected the data.  Y.Z. and J.S. analyzed the data. Y.Z., D.S., J.S., T.J., K.J.R., and R.W.S. composed the manuscript. 

\appendix 

\section{Proof of Theorem~\ref{thm: connection} and Proposition~\ref{cor: connection}}
\label{sec: methodA}
\setcounter{theorem}{0}
Below we provide a detailed proof of Theorem \ref{thm: connection} and Proposition \ref{cor: connection} from the main text.
We begin by recalling the notion of a classical preparation (or classical set of states) introduced in \cite{zhang2024parallel}.
\begin{definition}
\cite{zhang2024parallel} In measurement-full prepare-and-measure circuit $\mc{P}^{\rm full}=(\msf{M}^{\rm full},\msf{P}=\{\{\tilde{\rho}^B_{a|x}\}_a\}_x)$. The multi-source $\{\{\tilde{\rho}^B_{a|x}\}_a\}_x$ is said to be {\em classical} if the operational statistics arising in this circuit are NCOM-realizable relative to $\mc{O}_{\text{all}}(\msf{M}^{\rm full})\cup \mc{O}_{\text{all}}(\msf{P})$.
\label{def:classicalpreparation}
\end{definition} 
\begin{Thmprop}
\label{thm: appendix}
Let $\mc{B}^{\rm full}=(\msf{N}^{\rm full},\msf{M}^{\rm full},\msf P)$ denote a measurement-full Bell circuit. For the bipartite state $\rho$,  there is an equivalence between the two properties in  each of the following pairs:\\
(1) classicality relative to $\mc{O}_{\rm aattt}(\mc{B}^{\rm full})$ and separability, \\
(2) classicality relative to $\mc{O}_{\rm atttt}(\mc{B}^{\rm full})$ and $B{\rightarrow}A$ unsteerability, \\
(3) classicality relative to $\mc{O}_{\rm tattt}(\mc{B}^{\rm full})$ and $A{\rightarrow}B$ unsteerability,\\
(4) classicality relative to $\mc{O}_{\rm ttttt}(\mc{B}^{\rm full})$ and locality, \\
(5) classicality relative to $\mc{O}_{\rm aaatt}(\mc{B}^{\rm full})$ and the impossibility of steering to a nonclassical set of states on $B$, \\
(6) classicality relative to $\mc{O}_{\rm aatat}(\mc{B}^{\rm full})$ and the impossibility of steering to a nonclassical set of states on $A$.\\
\end{Thmprop}
\begin{proof}
The core of the proof is to apply the generalized Gleason theorem \cite{Busch2003} iteratively to the different components of the ontological model.
We start with Case (2) of Theorem \ref{thm: connection}, then treat Case (1). Cases (3) and (4) follow directly from these arguments, and Cases (5) and (6) will be proved at the end.

\textit{Case (2):}  If a bipartite state $\rho^{AB}$ is classical relative to $\mc O_{\rm atttt}(\mc B^{\rm full})$, by Definition~\ref{def:bipartite}, the operational statistics $p(ab|xy)$ it generates with all local POVMs $\{N^A_{a|x}\}_a\in \msf{N}^{\rm full}$ and $\{M^B_{b|y}\}_b\in \msf{M}^{\rm full}$ should be NCOM-realizable relative to $\mc O_{\rm atttt}(\mc B^{\rm full})$ per definition~\ref{def:statistics}. More specifically, one can write
\begin{align}
p(ab|xy)&=\tr[( N^A_{a|x}\otimes M^B_{b|y})\rho^{AB}]\notag \\
&=\sum_{\lambda_A\lambda_B} p(\lambda_A\lambda_B)p(a|x,\lambda_A)  p(b|y,\lambda_B) \label{eq:bi-NCOM} 
\end{align}
the assumption of noncontextuality applied to $\mc{O}_{\text{all}}(\msf N^{\rm full})$ implies that the response function $p(a|x,\lambda)$ is a linear function of $N^A_{a|x}$. The generalized Gleason's theorem then ensures the existence of a unique density operator $\rho^A_{\lambda_A}$ such that $p(a|x,\lambda_A)=\tr[\rho^A_{\lambda_A} N^A_{a|x}]$. Defining the steering assemblage $\{\{\tilde{\rho}^A_{b|y}=\tr[(\mbb{1}^A\otimes M^B_{b|y}) \rho^{AB}]\}_b\}_y$, one has
\begin{align}
p(ab|xy)&=\tr[\tilde{\rho}^A_{b|y}N^A_{a|x}] \notag \\
&=\sum_{\lambda_A\lambda_B}p(\lambda_A\lambda_B)p(b|y,\lambda_B)\tr[N^A_{a|x}\rho^A_{\lambda_A}] \notag \\
\Rightarrow \tilde{\rho}^A_{b|y}&=\sum_{\lambda_A\lambda_B}p(\lambda_A\lambda_B)p(b|y,\lambda_B)\rho^A_{\lambda_A},
\end{align}
where in the last line, we use the fact that the set of all measurements $\msf{N}^{\rm full}$ is tomographically complete.  This decomposition is a local hidden state model for the assemblage, so the state $\rho^{AB}$ is unsteerable from $B$ to $A$.

\par The reverse direction follows directly, if a bipartite state $\rho^{AB}$ is unsteerable, by Definition, the state assemblage $\{\{\tilde{\rho}^A_{b|y}=\tr[(\mbb{1}^A\otimes M^B_{b|y}) \rho^{AB}]\}_b\}_y, \forall \{M^B_{b|y}\}_b\in \msf M^{\rm full}$ should admits a local hidden state model~\cite{Uola2020}:
\begin{equation}
\tilde{\rho}^A_{b|y}=\sum_{\lambda_A}p(\lambda_A)p(b|y,\lambda_A)\rho^A_{\lambda_A},
\end{equation}
with$\sum_{\lambda_A}p(\lambda_A)=\sum_{b}p(b|y,\lambda_A)=1$, $p(\lambda_A), p(b|y,\lambda_A)\ge 0$ and $\rho^A_{\lambda_A}\ge \mbb{0}$. Therefore, the operational statistics $p(ab|xy)$ it generates with all local POVMs $\{N^A_{a|x}\}_a\in \msf{N}^{\rm full}$ and $\{M^B_{b|y}\}_b\in \msf{M}^{\rm full}$ must be NCOM-realizable relative to $\mc O_{\rm atttt}(\mc B^{\rm full})$, since one can define 
\begin{align}
p(ab|xy):&=\tr[\tilde{\rho}^A_{b|y}N^A_{a|x}]=\sum_{\lambda_A}p(\lambda_A)p(b|y,\lambda_A)\tr[N^A_{a|x}\rho^A_{\lambda_A}] \notag
\end{align}
where $\sum_b p(b|y,\lambda_A)=1$ and satisfies the trivial operational identities $\mc O_{\rm triv}(\msf{M}^{\rm full})$, while one can define $p(a|x,\lambda_A):=\tr[N^A_{a|x}\rho^A_{\lambda_A}]$, which is a nonnegative linear function of all $N^A_{a|x}\in \msf{N}^{\rm full}$ that satisfies $\msf O_{\rm all}(\msf{N}^{\rm full})$.

\textit{Case (1):} If a bipartite state $\rho^{AB}$ is classical relative to $\mc O_{\rm aattt}(\mc B^{\rm full})$, by Definition~\ref{def:bipartite}, the operational statistics $p(ab|xy)$ it generates with all local POVMs $\{N^A_{a|x}\}_a\in \msf{N}^{\rm full}$ and $\{M^B_{b|y}\}_b\in \msf{M}^{\rm full}$ should be NCOM-realizable relative to $\mc O_{\rm aattt}(\mc B^{\rm full})$ per definition~\ref{def:statistics}. 

Specifically, its operational statistics $p(ab|xy)$ should be realizable by Eq.~\eqref{eq:bi-NCOM}, where the assumption of noncontextuality applied to $\mc{O}_{\text{all}}(\msf N^{\rm full})$ and $\mc{O}_{\text{all}}(\msf M^{\rm full})$ implies that the response function $p(a|x,\lambda_A)$ is a linear function of the effect $N^A_{a|x}$ and that $p(b|y,\lambda_B)$ is a linear function of $M^B_{b|y}$.  The generalized Gleason's theorem then guarantees the existence of unique density operators $\rho^A_{\lambda_A}$ and $\sigma^B_{\lambda_B}$ such that we can write $p(a|x,\lambda_A)=\tr[\rho^A_{\lambda_A} N^A_{a|x}]$ and  $p(b|y,\lambda_B)=\tr[\sigma^B_{\lambda_B} M^B_{b|y}]$. Plugging these into Eq.~\eqref{eq:bi-NCOM}, for every pair of POVMs $\{N^A_{a|x}\}_a$ and $\{M^B_{b|y}\}_b$, we have
\begin{align}
&\tr[( N^A_{a|x}\otimes  M^B_{b|y})\rho^{AB}]\notag \\&=\sum_{\lambda_A\lambda_B}p(\lambda_A\lambda_B)\tr[N^A_{a|x}\rho^A_{\lambda_A}]\tr[M^B_{b|y}\sigma^B_{\lambda_B}]\notag \\
\Rightarrow \rho^{AB}&=\sum_{\lambda_A\lambda_B}p(\lambda_A\lambda_B)\rho^A_{\lambda_A}\otimes \sigma^B_{\lambda_B},
\end{align}
where in the last line, we used the fact that the set of all product measurements is tomographically complete for the space of quantum states. So the state is separable. 

The reverse direction is straightforward, as one only needs to show that the operational statistics arising from a separable state $\rho^{AB}=\sum_{\lambda_A\lambda_B}p(\lambda_A\lambda_B)\rho^A_{\lambda_A}\otimes \sigma^B_{\lambda_B}$ along with all POVMs $\{N^A_{a|x}\}_a\in \msf{N}^{\text{full}}$ and $\{M^B_{b|y}\}_b\in \msf{M}^{\text{full}}$ are NCOM-realizable relative to $\mc{O}_{\rm aattt}(\mc B^{\rm full})$. This can be met by defining the response functions  $p(a|x,\lambda_A)=\tr[\rho^A_{\lambda_A} N^A_{a|x}]$ and  $p(b|y,\lambda_B)=\tr[\sigma^B_{\lambda_B} M^B_{b|y}]$, which clearly are linear maps from POVMs to their ontological representations and respect $\mc O_{\rm all}(\msf N^{\rm full})$ and $\mc O_{\rm all}(\msf M^{\rm full})$.

Now, we will prove case (5); case (6) follows similarly. 

\par \textit{Case (5):} If a bipartite state $\rho^{AB}$ is classical relative to $\mc O_{\rm aaatt}(\mc B^{\rm full})$, by Definition~\ref{def:bipartite}, the operational statistics $p(ab|xy)$ it generates with all local POVMs $\{N^A_{a|x}\}_a\in \msf{N}^{\rm full}$ and $\{M^B_{b|y}\}_b\in \msf{M}^{\rm full}$ should be NCOM-realizable relative to $\mc O_{\rm aaatt}(\mc B^{\rm full})$ per definition~\ref{def:statistics}. Specifically, its operational statistics $p(ab|xy)$ should be realizable by Eq.~\eqref{eq:bi-NCOM} and the assumption of noncontextuality applied to $\mc{O}_{\text{all}}(\msf M^{\rm full})$ implies the existence of unique density operators $\sigma^B_{\lambda_B}$ such that we can write  $p(b|y,\lambda_B)=\tr[\sigma^B_{\lambda_B} M^B_{b|y}]$. Moreover,  the assumption of noncontextuality applied to $\mc{O}_{\text{all}}(\msf N^{\rm full}\circ \msf P)$ implies the existence of a Hermitian operator $H^B_{\lambda_B}$ such that  $\sum_{\lambda_A}p(a|x,\lambda_A)p(\lambda_A\lambda_B)=\tr[(N^A_{a|x}\otimes H^B_{\lambda_B})\rho^{AB}]=\tr[H^B_{\lambda_B}\tilde{\rho}^B_{a|x}]\ge 0$, where we use the definition of $\tilde{\rho}^B_{a|x}$ in  Eq.~\eqref{eq:steer A}. Combining these facts, we have
\begin{align}
\tr[\tilde{\rho}^B_{a|x}M^B_{b|y}] &=\tr[(N^A_{a|x}\otimes M^B_{b|y})\rho^{AB}] \notag \\
&=\sum_{\lambda_A\lambda_B}p(\lambda_A\lambda_B)p(a|x,\lambda_A)p(b|y,\lambda_B) \notag \\
&=\sum_{\lambda_B}\tr[H^B_{\lambda_B}\tilde{\rho}^B_{a|x}]\tr[M^B_{b|y}\sigma^B_{\lambda_B}] \notag \\
\Rightarrow \tilde{\rho}^B_{a|x}&=\sum_{\lambda_B}\tr[H^B_{\lambda_B}\tilde{\rho}^B_{a|x}]\sigma^B_{\lambda_B}, \label{eq:classicalstate}
\end{align}
where in the last line, we used the fact that the set of all measurements is tomographically complete for the space of quantum states. According to Theorem 2 in \cite{zhang2024parallel}, decomposition in Eq.~\eqref{eq:classicalstate} ensures that the assemblage $\{\{\tilde{\rho}^B_{a|x}\}_a\}_x$ is classical. 

For the reverse direction, if the state assemblage $\{\{\tilde{\rho}^B_{a|x}:=\tr[(N^A_{a|x}\otimes \mbb{1}^B)\rho^{AB}]\}_a\}_x$ defined for $\forall \{N^A_{a|x}\}_a\in \msf{N}^{\rm full}$ is classical, according to Theorem 2 in \cite{zhang2024parallel}, one can find a frame representation $\tilde{\rho}^B_{a|x}=\sum_{\lambda_B}\tr[H^B_{\lambda_B}\tilde{\rho}^B_{a|x}]\sigma^B_{\lambda_B}$ with $\sigma^B_{\lambda_B}\ge \mbb{0}$ and $\tr[H^B_{\lambda_B}\tilde{\rho}^B_{a|x}]\ge 0$ for all $\lambda_B$.

Therefore, combining with all $\{M^B_{b|y}\}_b\in \msf{M}^{\rm full}$ one has:
\begin{align}
\tr[(N^A_{a|x}&\otimes M^B_{b|y})\rho^{AB}]=\tr[\tilde{\rho}^B_{a|x}M^B_{b|y}] \notag \\
=&\sum_{\lambda_B}\tr[H^B_{\lambda_B}\tilde{\rho}^B_{a|x}]\tr[M^B_{b|y}\sigma^B_{\lambda_B}] \notag \\
=&\sum_{\lambda_B}\tr[(N^A_{a|x}\otimes H^B_{\lambda_B})\rho^{AB}]\tr[M^B_{b|y}\sigma^B_{\lambda_B}] \notag \\
:=&\sum_{\lambda_B}\tr[N^A_{a|x}\tau^A_{\lambda_B}]\tr[M^B_{b|y}\sigma^B_{\lambda_B}] \end{align}
In the last equality, we define $\tau^A_{\lambda_B}:=\tr_B[(\mbb{1}^A\otimes H^B_{\lambda_B})\rho^{AB}]$, where $\tau^A_{\lambda_B}\ge\mbb{0}$ is ensured because $\tr[N^A_{a|x}\tau^A_{\lambda_B}]=\tr[H^B_{\lambda_B}\tilde{\rho}^B_{a|x}]\ge 0$ for all $\{N^A_{a|x}\}_a\in \msf{N}^{\rm full}$.  Hence, one can simply define $p(\lambda_A\lambda_B)=\tr[\tau^A_{\lambda_A}]\delta_{\lambda_A\lambda_B}$, $p(b|y,\lambda_B)=\tr[M^B_{b|y}\sigma^B_{\lambda_B}]$ and $p(a|x,\lambda_A)=\tr[N^A_{a|x}\tau^A_{\lambda_A}]/\tr[\tau^A_{\lambda_A}]$ (if $\tr[\tau^A_{\lambda_A}]=0$, $p(a|x,\lambda_A)$ can be defined arbitrarily), respectively,  as the ontological representation for each process in the ontological model. 

\par 
Checking that this model satisfies every noncontextuality constraint in $\mc{O}_{\rm aaatt}(\mc{B}^{\rm full})$ is then straightforward. The constraints in $\mc{O}_{\rm all}(\msf{N}^{\rm full})$ and $\mc{O}_{\rm all}(\msf{M}^{\rm full})$ hold trivially, since $\tau^A_{\lambda_B}, \sigma^B_{\lambda_B}\ge \mbb{0}$. For $\mc{O}_{\rm all}(\msf{N}^{\rm full}\circ \msf P)$, note that $\sum_{\lambda_A }p(\lambda_A\lambda_B)p(a|x,\lambda_A) =\sum_{\lambda_A }\delta_{\lambda_A,\lambda_B}\tr[N^A_{a|x}\tau^A_{\lambda_A}]=\tr[(N^A_{a|x}\otimes H^B_{\lambda_B})\rho^{AB}]=\tr[\tilde{\rho}^B_{a|x}H^B_{\lambda_B}]\ge 0$, hence the composition map is linear in the assemblage $\{\{N^A_{a|x}\circ \msf P^{AB}=\tilde{\rho}^B_{a|x}\}_a\}_x$ for all $\{N^A_{a|x}\}_a\in \msf{N}^{\rm full}$. Therefore, we have a valid ontological model that respects  $\mc{O}_{\rm aaatt}(\mc{B}^{\rm full})$
\end{proof}

\blk

\section{Deriving the noncontextuality inequalities}\label{sec: methodB}

Consider a Bell circuit $\mc{B}=(\msf N, \msf M, \msf P)$ where Alice and Bob each perform $\Delta_{\msf N}$ and $\Delta_{\msf M}$ different $o$-outcome measurements, respectively. One can visualize the operational statistics $p(ab|xy)$ as a point $\mbf{p}:=\{p(ab|xy)\}_{a,b,x,y}\in \mbb{R}^{o^2\Delta_{\msf N}\Delta_{\msf M}}$. 
The set of all the operational statistics that are NCOM-realizable relative to $\mc O_{\rm aattt}(\mc B)$ forms a polytope in this space, which we denote as $\mc{NC}_{\rm  aattt}(\mc B)$. 

As stated in the main text (and illustrated by Example~\ref{example1}), a set of operational identities $\mc{O}_{\rm all}({\msf{N}})$ can be represented by a geometric object---specifically, a polytope.\footnote{the polytope describing the intrinsic geometry of $\msf{N}$
should not be confused with the noncontextual polytope of NCOM-realizable
distributions.} Concretely, each operational identity is a set of coefficients $\mc{O}_{\rm all}({\msf{N}}) = \{\{ \alpha_{a,x}\}_{a,x}\}$, so one can associate to it a polytope $P$ whose vertices are indexed by $\{\vec{t}_{a,x}\}$ satisfying all and only the equalities $\sum_{a,x}\alpha_{a,x}\vec{t}_{a,x}=\vec{0}$ and no other equalities. Conversely, given a polytope $P$ with vertices indexed by ${a,x}$, one can pick an arbitrary coordinate system to represent these vertices as vector $\vec{t}_{a,x}$ and define the set of operational identities via the equalities they satisfy: $\mc{O}_{\rm all}({\msf{N}})=\{\{\alpha_{a,x}\}|\sum_{a,x}\alpha_{a,x}\vec{t}_{a,x}=\vec{0}\}$. (An explicit method for computing these coefficients from a given set of vectors is given in Appendix~A of Ref.~\cite{schmid2024a}.) This duality allows us to visualize operational constraints in purely geometric terms (as we did in the main text).

The extreme points $\vec{v}\in\mc{NC}_{\rm  aattt}(\mc B)$ are those that can be obtained from the extreme operational statistics $p(ab|xy)=\sum_{\lambda_A\lambda_B}p(\lambda_A\lambda_B)p(a|x,\lambda_A)p(b|y,\lambda_B)$ when $p(\lambda_A\lambda_B)$ and when the response functions $\{p(a|x,\lambda_A)\}_{a,x}$ and $\{p(b|y,\lambda_B)\}_{b,y}$ are also extremal within the {\em noncontextual-measurement-assignment polytope}~\cite{David2018,zhang2025tech}. The noncontextual-measurement-assignment polytope $\mbb{P}_{\msf{N}}$ is formally defined as a collection of points $\vec{u}:=\{p(a|x)\}_{a,x}$ such that
\begin{subequations}
\begin{align}
\text{(i)}~~&p(a|x) \geq 0~~~~\forall a,x  \label{eq:cons_pos}\\
\text{(ii)}~~&\sum_a p(a|x)=1~~~~\forall x  \label{eq:cons_norm}\\
\text{(iii)}~~&\sum_{a,x} \alpha_{a,x} p(a|x)=0 \quad\quad \forall \{\vec{\alpha}\}\in \mc{O}_{\rm all}({\msf{N}}).
\label{eq:cons_nonc}
\end{align}
\label{eq:assignment polytope}
\end{subequations}
\cmt Similarly, one can define the \textit{measurement-assignment polytope} $\mbb{D}_{\msf N}$ as the collection of points $\vec{u}:=\{p(a|x)\}_{a,x}$ satisfying only Eq.~\eqref{eq:cons_pos} and Eq.~\eqref{eq:cons_norm} above. 

Therefore, all extreme points $\vec{v}\in\mc{NC}_{\rm aattt}(\mc B)$ can be obtained from the extreme points  $\vec{u}\in \mbb{P}_{\msf{N}}$ and $\vec{w}\in \mbb{P}_{\msf{M}}$ via $\vec{v}= \vec{u}\otimes\vec{w}$,  and analogously, all extreme points $\vec{v}\in\mc{NC}_{\rm tattt}(\mc B)$ can be obtained from the extreme points  $\vec{u}\in \mbb{D}_{\msf{N}}$ and $\vec{w}\in \mbb{P}_{\msf{M}}$ via $\vec{v}= \vec{u}\otimes\vec{w}$.

\cmt It is important to note that the noncontextual-measurement-assignment polytope $\mbb{P}_{\msf{N}}$ is a polytope in the space of conditional probability vectors $\{p(a|x)\}_{a,x}\in \mbb R^{o\Delta N}$ constrained by the normalization conditions in Eq.~\eqref{eq:cons_norm} and ontological identities in Eq.~\eqref{eq:cons_nonc}. Consequently, the affine dimension of $\mbb{P}_{\msf{N}}$ is $K-1$, bounded solely by the operator space dimension $K$ that the set of effects in $\msf{N}$ spans (Assuming the effects in $\msf{N}$ are informationally complete for a $D$-dimensional Hilbert space, then $K=D^2$). More importantly, since $\mbb{P}_{\msf N}$ has only $o\Delta_N$ inequality constraints in Eq.~\eqref{eq:cons_pos}, the number of vertices in $\mbb{P}_{\msf{N}}$ is upper bounded by a polynomial in $o\Delta_N$ as per the dual form of the Upper Bound Theorem~\cite{McMullen1970,alon1985}. In contrast, the local response functions in the Bell test define the measurement-assignment polytope $\mbb{D}_{\msf{N}}$ based solely on the positivity and normalization constraints in Eq.~\eqref{eq:assignment polytope}. This polytope lives in a $(o-1)\Delta_N$-dimensional affine subspace of $\mbb R^{o\Delta N}$ and has $o^{\Delta_\msf N}$ deterministic vertices.

Similarly, one can verify that the noncontextual polytope $\mc{NC}_{\text{aattt}}(\mc B)$ resides in a $(D^4-1)$-dimensional subspace of the vector space $\mbb R^{o^2\Delta_N\Delta_M}$ (assuming the products of effects in $\msf{N}$ and $\msf{M}$ are informationally complete for the $D^2$-dimensional Hilbert space), which is independent of the number of outcomes $o$ and settings $\Delta_N, \Delta_M$. Moreover, the number of vertices in $\mc{NC}_{\text{aattt}}(\mc B)$ scales only polynomially with $o^2\Delta_N\Delta_M$ since they are determined by the number of vertices in $\msf P_{\msf N}$ and $\msf P_{\msf M}$. By comparison, the Bell polytope lives in an affine subspace of dimension $(o-1)(\Delta_N+\Delta_M)+(o-1)^2\Delta_N\Delta_M$~\cite{Brunner2014}, and its vertex count scales exponentially as $o^{\Delta_N+\Delta_M}$. As a result, in contrast with characterizing the local polytope, characterizing $\mc{NC}_{\text{aattt}}(\mc B)$ remains computationally feasible, even when the number of settings for each party is very large.\blk

For simple cases, one can completely characterize the polytope $\mc{NC}_{\rm aattt}(\mc B)$ by enumerating all the extreme points $\{\vec{v}=\vec{u}\otimes \vec{w}\}$ with $\vec{u}$ being an extreme point of $\mbb{P}_{\msf{N}}$ and $\vec{w}$ being an extreme point of $\mbb{P}_{\msf{M}}$. Taking the class of (octahedron, cube)-Bell circuits for instance: for multi-measurements $\msf{N}=\{\{N^A_{a|x}\}_{a}\}_{x}$ and $\msf{M}=\{\{M^B_{b|y}\}_{b}\}_{y}$ where $a,b\in[0,1]$, $x\in[0,1,2]$ and $y\in [0,1,2,3]$ having cubic and octahedral symmetry respectively, we find that the  $\mc{NC}_{\rm aattt}(\mc B)$ polytope lives in a $15$-dimensional subspace of the full $48$-dimensional probability space ($d = o^2\Delta_{\msf N}\Delta_{\msf M}=48$), where the spanned dimension is reduced by equality constraints coming from normalization, and noncontextuality. \par

In the $15$-dimensional subspace,  we find that the noncontextual polytope has $336$ facet inequalities in total, which can be classified into three equivalence classes under permutations of values of the inputs and of the outputs that respect the cubic and octahedral symmetries, respectively. We therefore explicitly list a single representative inequality for each class, as in Table~\ref{tab:ineq68}. Importantly, one can show that, for a two-qubit isotropic state $\rho^v_{\text{iso}}$ with $v>\frac{1}{\sqrt{3}}$,  the last inequality can be violated using the following local qubit measurements that respect operational identities $\mc O_{\rm aattt}(\msf{Bell}_{oc})$ in the class of (octahedron, cube)-Bell circuits: \cmt $N^A_{a|x}=\frac{1}{2}(\mbb{1}+(-1)^a\hat{n}_{x}\cdot\vec{\sigma})$ and $M^B_{b|y}=\frac{1}{2}(\mbb{1}+(-1)^b\hat{n}_{ y}\cdot\vec{\sigma})$, where 
\begin{align}
&\hat{n}_{0}=[1,0,0],~\hat{n}_{1}=[0,1,0],~\hat{n}_{2}=[0,0,1] \notag \\
&\hat{m}_{0}=\frac{1}{\sqrt{3}}[1,1,1],~\hat{m}_{1}=\frac{1}{\sqrt{3}}[1,1,-1],\notag \\
&\hat{m}_{2}=\frac{1}{\sqrt{3}}[1,-1,1],~\hat{m}_{3}=\frac{1}{\sqrt{3}}[1,-1,-1].
\end{align}
In the Bloch ball, these have the geometries depicted in Fig~\ref{fig:plantoic}(a)
\par 

In a more complex scenario, the computational complexity of the facet enumeration problem increases with the number of measurement settings. \cmt Moreover, unlike the local polytope, where vertices are rational vectors (deterministic strategies), the noncontextual polytope generally contains vertices with irrational components, arising from the presence of irrational components in the operational identities. While this complicates the full characterization the polytope $\mc{NC}_{\rm aattt}(\mc B)$ as most standard facet-enumeration tools assume rational vertices (so they can perform Fourier–Motzkin elimination efficiently), the lower dimensionality and the moderate number of vertices of the noncontextual polytope $\mc{NC}_{\rm aattt}(\mc B)$ (relative to the local polytope) renders \emph{polytope membership certification} efficient. For scenarios with hundreds of settings, the membership problem can be solved exactly via standard linear programming, avoiding the need for iterative algorithms such as Frank-Wolfe~\cite{Designolle2023} or column generation. 

We now turn to the problem of deciding membership of a given correlation $p(ab|xy)$ in the noncontextual polytope $\mc{NC}_{\rm aattt}(\mc B)$, and how this problem compares with that of deciding membership in the local polytope. 

Consider a  Bell circuit $\mathcal{B}$ with setting variable cardinalities $\Delta_{\msf N}$ and $\Delta_{\msf M}$  and outcome variable cardinality $o$ again. We can represent the $k$ extreme points $\vec{v}\in \mc{NC}_{\rm aattt}(\mc B) \subset \mbb{R}^{o^2\Delta_{\msf N}\Delta_{\msf M}}$ as columns of a $o^2\Delta_{\msf N}\Delta_{\msf M}\times k$ matrix $\mbf{T}$. A correlation $p(ab|xy)$, expressed as a probability vector $\mbf p:= \{p(ab|xy)\}_{abxy}\in \mbb{R}^{o^2\Delta_{\msf N}\Delta_{\msf M}}$, lies within the noncontextual polytope $\mc{NC}_{\rm aattt}(\mc B)$ if and only if
\begin{align}
\text{there exists some }\quad\quad  &\mbf{x}  \notag \\
 \text{ such that}\quad\quad     &\mbf{T}\mbf{x}= \mbf{p}, \\
\text{where} \quad\quad &\mbf{x}\ge \mbf{0},  \notag 
\end{align}
with $\mbf{x}\in \mbb{R}^k$. We can also write up the dual LP as:
\begin{align}
    &\min_{\mbf{s}}\quad\quad\quad \mbf{s}^T\mbf{p} \notag \\
    &\text{such that} \quad \mbf{0}\leq \mbf{s}^T\mbf{T} \leq \mbf{1}.
\end{align}
with $\mbf{s}\in \mbb{R}^{o^2\Delta_N\Delta_M}$.

If a correlation $\mbf p= \{p(ab|xy)\}_{abxy}$ does not lie in the polytope, then this dual linear program will always find a real vector $\mbf{s}\in \mbb{R}^{o^2\Delta_N\Delta_M}$ satisfying the condition $\mbf{s}^T\mbf{p} < 0$. This real vector $\mbf{s}$ can then be used to define a facet-defining noncontextuality inequality that $\mathbf{p}$ violates.

Crucially, as noted earlier, the noncontextual polytope resides within a space of dimension $D^4-1 \ll o^2\Delta_N\Delta_M$, where $D$ is the Hilbert space dimension, which is much smaller than the dimension of the space in which the local polytope resides. Consequently, the linear program defined above can always be recast as an equivalent problem in its minimal affine representation (i.e., using reduced coordinates of dimension $D^4-1$). As demonstrated in our example~\ref{example4}, the membership problem above can be efficiently solved even when the full probability space dimension $d = o^2\Delta_{\msf N}\Delta_{\msf M} \gg 10^5$.\blk

The inequality we gave in Eq.~\eqref{eq: Inequality} is a facet inequality for the noncontextual polytope $\mathcal{NC}_{\rm aattt}(\msf{Bell}_{\rm id})$, where $\msf{Bell}_{\rm id}$ denotes \blk the class of (icosahedron, dodecahedron)-Bell circuits, considered in Example~\ref{example3}. For an isotropic state $\rho^v_{\text{iso}}$ with $v>\sqrt{\frac{1+\phi^2}{{3\phi^4}}}$,  the inequality can be violated using the following measurements~\cite{Saunders2010,zhang2023} that have icosahedral and dodecahedral symmetries, respectively: \cmt $N^A_{a|x}=\frac{1}{2}(\mbb{1}+(-1)^a\hat{n}_{x}\cdot\vec{\sigma})$ and $M^B_{b|y}=\frac{1}{2}(\mbb{1}+(-1)^b\hat{m}_{ y}\cdot\vec{\sigma})$ with $\hat{n}_{x}=\frac{1}{\sqrt{1+\phi^2}}\vec{n}_x$ and $\hat{m}_{y}=\frac{1}{\sqrt{3}}\vec{m}_y$
\begin{align}
&\vec{n}_{0}=[0, 1, \phi],~\vec{n}_{1}=[0, 1, -\phi],~\vec{n}_{2}=[1, \phi, 0] \notag \\
&\vec{n}_{3}=[1, -\phi, 0],~\vec{n}_{4}=[\phi, 0, 1],~\vec{n}_{5}=[\phi, 0, -1] \notag \\
&\vec{m}_{0}=[-1,-1,-1],~\vec{m}_{1}=[-\phi,\frac{1}{\phi},0],~\vec{m}_{2}=[-\phi,-\frac{1}{\phi},0]\notag \\
&\vec{m}_{3}=[-1,-1,1],~\vec{m}_{4}=[0,-\phi,\frac{1}{\phi}],~\vec{m}_{5}=[1,-1,-1]\notag \\
&\vec{m}_{6}=[0,-\phi,-\frac{1}{\phi}],~\vec{m}_{7}=[-1,1,-1],~\vec{m}_{8}=[-\frac{1}{\phi},0,-\phi]\notag \\
&\vec{m}_{9}=[\frac{1}{\phi}, 0,-\phi]. 
\label{12-20meas}
\end{align}
In the Bloch ball, these have the geometries depicted in Fig~\ref{fig:plantoic}(b).
\yujie
\blk
\subsection{Deriving the facet inequalities from the noncontextual polytope $\mc{NC}_{\text{tattt}}$}\label{sec: methodB1}

As noted in Sec.~\ref{Sec:entanglement-certification}, characterizing the facet inequalities of the noncontextual polytope $\mc{NC}_{\text{aattt}}(\mc B)$, where one uses the nontrivial sets of operational identities for both Alice's and Bob's multi-measurements,  offers distinct advantages for entanglement certification over existing protocols. For instance, unlike quantum tomography, it does not require a prior characterization of the measurement. Moreover, it improves upon the Bell-inequality-based approach because, in principle, it can certify any entangled state.

A similar argument applies to characterizing the facet inequalities of the noncontextual polytope  $\mc{NC}_{\text{tattt}}(\mc B)$ for $A\rightarrow B$ steering, where one only uses the nontrivial sets of operational identities for Bob's multi-measurement. Our noncontextuality-based approach does not require characterization of the measurement devices on either side. In contrast, standard steering inequalities require full tomography of the measurement devices on one side, while remaining fully device-independent on the other side (often termed one-sided device-independent). Moreover, while violating a Bell inequality is sufficient for device-independent steering certification~\cite{Uola2020}, the noncontextuality inequalities derived from our noncontextual polytope $\mc{NC}_{\rm tattt}(\mc B)$ can outperform all known Bell inequalities, as demonstrated in Table~\ref{table:steering}.

In the main text, we have already presented two classes of Bell circuits (Examples~\ref{example1} and~\ref{example2}) where $\mc{NC}_{\rm tattt}(\mc B)=\mc{NC}_{\rm aattt}(\mc B)$. To derive more generic noncontextual inequalities for certifying quantum steering,  one needs to construct the corresponding noncontextual polytope $\mc{NC}_{\rm tattt}(\mc B)$
by enumerating all its extreme points $\vec{v}$ from the extreme points  $\vec{u}\in \mbb{D}_{\msf{N}}$ and $\vec{w}\in \mbb{P}_{\msf{M}}$ via $\vec{v}= \vec{u}\otimes\vec{w}$. While the subsequent analysis of $\mc{NC}_{\rm tattt}(\mc B)$
remains the same as that of $\mc{NC}_{\rm aattt}(\mc B)$ introduced in Appendix~\ref{sec: methodB}, the linear program for its characterization is more complex in practice, as both the dimension and vertex complexity of the measurement-assignment polytope $\mbb{D}_{\msf{N}}$ are far greater than those of the noncontextual measurement-assignment polytope $\mbb{P}_{\msf N}$ (see Appendix~\ref{sec: methodB}). 

\cmt 

Consequently, to reduce the numerical complexity, we restrict our analysis to cases in which only the number of settings on Bob's multi-measurement increases, while Alice employs a fixed multi-measurement with $6$ measurement settings. It is important to note that because we fix the cardinality of Alice's setting variable, the resulting critical visibility $v$ will not converge to the fundamental limit $v =1/2$ as the cardinality of  Bob's setting variable approaches infinity; however, our approach is shown to already outperform all known Bell-inequality-based steering certification~\cite{Designolle2023}.

We note again that our test requires only the standard assumption for testing noncontextuality in a prepare-measure circuit on Bob's side, namely, tomographic completeness of Bob's measurements and steered states on $B$. This suggests that our protocol could offer practical advantages relative to standard steering-certification methods. A more thorough analysis and comparison along these lines is left for future work.
\blk
\begin{table}[h]
\cmt 
\centering
\caption{\label{table:steering}
Critical visibility $v$ of the two-qubit isotropic state $\rho_{\text{Iso}}^v$  that can be certified as steerable using noncontextuality inequalities derived from the noncontextual polytope $\mc{NC}_{\text{tattt}}(\mc B)$. Note that the result for Row 3 differs from the corresponding entry in Table~\ref{table:geodesic}, which uses inequalities derived from the noncontextual polytope $\mc{NC}_{\text{aattt}}(\mc B)$ for entanglement certification.}
\begin{ruledtabular}
\begin{tabular}{ccc}
 $\#$ of Alice's settings & $\#$ of Bob's settings  & visibility $v$ \\
\hline
2 (square) & 2 (square) & $\approx 0.7071$\\
3 (octahedron) & 4 (cube) & $\approx 0.5773$ \\
6 (icosahedron)  & 10 (dodecahedron) & $\approx 0.6787$ \\
6 (icosahedron)  & 40 (Goldberg $\nu=2$) & $\approx 0.5655$ \\
6 (icosahedron)  & 90 (Goldberg $\nu=3$) & $\approx 0.5493$ \\
6 (icosahedron)  & 160 (Goldberg $\nu=4$) & $\approx 0.5445$ \\
6 (icosahedron)  & 250 (Goldberg $\nu=5$) & $\approx 0.5425$ \\
6 (icosahedron)  & 360 (Goldberg $\nu=6$) & $\approx 0.5414$ \\
6 (icosahedron)  & 490 (Goldberg $\nu=7$) & $\approx 0.5408$ \\
\end{tabular}
\end{ruledtabular}
\end{table}

\section{Dealing with inexact operational identities}
\label{sec: methodC}
\subsection{Self-consistent tomography of raw data: frequencies to probabilities}
\blk
\label{sec: methodC2}
Experiments do not yield probabilities directly. Rather, they yield relative frequencies which, under ideal circumstances (of the samples being independent and identically distributed), would converge 
to the true underlying probabilities in the infinite-run limit. 
Interpreting relative frequencies as probabilities can lead to a misinterpretation of experimental results. \blk For instance, in a typical Bell experiment, interpreting frequencies as probabilities without accounting for statistical fluctuations can result in an apparent violation of the no-signaling principle~\cite{Lin2018}.

This issue is equally relevant in our experiment. In principle, statistical fluctuations in the relative frequencies \blk could lead to a situation wherein interpreting these as probabilities yields an apparent violation of the predictions of quantum theory (for a given choice of dimensionality of systems). To address this issue, we adapted a well-motivated regularization method, similar to those used in previous works~\cite{Mazurek2016,Lin2018,Mazurek2021}, to convert the experimentally observed relative frequencies into {\em estimates} of the true distribution. 
Specifically, we denote the frequency obtained directly from raw data in a bipartite Bell circuit as $f_{ab|xy}$ and its variance as $\Delta f_{ab|xy}$. We then consider the weighted $\chi^2$ optimization, where \blk
\begin{equation}
    \chi^2
  \blk  :=\min_{p(ab|xy)}\sum_{abxy}\left|\frac{f_{ab|xy}-p(ab|xy)}{\Delta f_{ab|xy}}\right|^2.
\end{equation}
Note that one could also replace the $\chi^2$ function with the $L^2$ function or the KL divergence, as discussed in \cite{Lin2018}.

\par 
In our data analysis, we considered two types of theoretical models to fit the relative frequencies. 
The first assumes quantum mechanics with a fixed Hilbert space dimension  $D$ (hence, operator space dimension $D^2$),\footnote{This step assumes the dimension of the underlying quantum system. However, one can also perform the fit in a way that \emph{infers} an effective quantum dimension from the experimental data using a train-and-test methodology, rather than fixing $D$ a priori~\cite{Grabowecky2022, Mazurek2021}.} \blk where $p(ab|xy) = \text{tr}[(N^A_{a|x} \otimes M^B_{b|y}) \rho^{AB}]$. The second allows for a realization in an arbitrary generalized probabilistic theory (GPT) with fixed GPT vector space dimension $D^2$, where $p(ab|xy) = N^A_{a|x} \otimes M^B_{b|y} \cdot G^{AB}$. In this case, $N^A_{a|x}$, $M^B_{b|y}$, and $G^{AB}$ are GPT vectors corresponding to unipartite effects on $A$ and $B$ and the bipartite state on $AB$ respectively. \blk

\cmt \blk

As the two regularization methods yielded no statistically significant differences, we reported the quantum realization in the main text. We note that the associated optimization problem cannot be solved in a single step; instead, we employ a see-saw algorithm to iteratively optimize the bipartite state and local effects, i.e.,
\begin{enumerate}
    \item Given an initial guess for $\msf{N}$, $\msf{M}$ and $\rho^{AB}$,
    \item Fix $\msf{N}$ and $\msf{M}$ and update states $\rho^{AB}$ to be $   \text{argmin}_{\rho^{AB}}\sum_{abxy}\left|\frac{f_{ab|xy}-\tr[(N^A_{a|x}\otimes M^B_{b|y})\rho^{AB}]}{\Delta f_{ab|xy}}\right|^2$,
    \item Fix $\msf{N}$ and state $\rho^{AB}$, and update $\msf{M}$ to be $    \text{argmin}_{\msf{M}}\sum_{abxy}\left|\frac{f_{ab|xy}-\tr[(N^A_{a|x}\otimes M^B_{b|y})\rho^{AB}]}{\Delta f_{ab|xy}}\right|^2$,
    \item  Fix $\msf{M}$ and state $\rho^{AB}$, and update $\msf{N}$ to be $\text{argmin}_{\msf{N}}\sum_{abxy}\left|\frac{f_{ab|xy}-\tr[(N^A_{a|x}\otimes M^B_{b|y})\rho^{AB}]}{\Delta f_{ab|xy}}\right|^2$,
\item Repeat step 2-4 until it converges relative to a certain threshold. 
\end{enumerate}

Each step of the optimization above can be expressed as a quadratic program, which could be efficiently solved computationally.  However, since the bipartite statistics $p(ab|xy)$ are invariant under invertible local transformations (the tomographic gauge freedom noted in Ref.~\cite{Nielsen2021} and described in the main text) the optimization does not yield a unique solution for the bipartite state $\rho^{AB}$  and the sets of local effects $\msf{N}$ and $\msf{M}$. 

\par 
These regularization procedures are closely related to GPT tomography \cite{Mazurek2021} and gate-set tomography \cite{Nielsen2021}. Consequently, any tools developed for such bootstrap tomography schemes can be directly applied here as well. However, we do not use this self-consistent tomography naively to certify entanglement, namely by performing a separability test on the bipartite state obtained from the tomography scheme.  We avoid this approach because this output $\rho^{AB}$ from a self-consistent tomography is gauge-dependent and would require an optimization over the choice of gauge.  Instead, our entanglement certification technique evaluates noncontextuality inequalities on the regularized statistics, and since these statistics are gauge-independent quantities, no gauge optimization is required. 

Furthermore, since characterizing the bipartite state is not our objective, a distinct yet equivalent approach can be taken to regularize the raw statistics. Specifically, rather than employing a see-saw algorithm to iteratively optimize the bipartite state on $AB$ and the sets of local effects on $A$ and on $B$, one can instead just iteratively optimize the steering assemblage on $A$,  $\{\{\tilde{\rho}^B_{b|y}\}\}_{b,y}$, where $\tilde{\rho}^B_{b|y} = \tr_B[(\mathbb{1}^A \otimes M^B_{b|y})\rho^{AB}]$ and the set of local effects on $A$, $\msf{N}$  (or, equivalently, the steering assemblage and set of local effects on $B$). The GPT version of this regularization process is precisely the one studied in the prepare-measure circuit in \cite{Mazurek2021, Mazurek2021}.

\cmt  It is worth noting again that because of the tomographic gauge freedom described in the main text, there is a nonuniqueness in the output of the fitting procedure that yields the bipartite state and the primary measurements, and hence a nonuniqueness in the secondary measurements as well. However, under the assumptions of tomographic completeness, the operational identities and the statistics are manifestly gauge-independent.  Because the noncontextuality inequalities are derived from those gauge-independent operational identities, and it is the gauge-independent statistics on which the noncontextuality inequalities are tested, our noncontextuality tests are gauge-independent.
\blk

\subsection{Secondary measurements and secondary statistics}
\label{sec: methodC1}
\cmt
As in the main text, we will denote the output of the regularization step as the \textit{primary statistics}, denoted as:
\begin{align}
p^{\text{pri}}(ab|xy)=\tr[N^{A,\text{pri}}_{a|x}\otimes M^{B,\text{pri}}_{b|y}\rho^{AB}],
\label{eq:pri-stat}
\end{align}
where $\{N^{A,\text{pri}}_{a|x}\}$ and $\{M^{B,\text{pri}}_{b|y}\}$ are referred to as the primary measurements. \blk

Here, we will detail the method of secondary procedures.  This method takes as input a set of processes that are obtained by fitting to the relative frequencies (i.e., self-consistent tomography) and that do not satisfy the required operational identities used in deriving the inequality test, and outputs a set of processes in their convex hull that satisfy these identities exactly. 

Consider again a multi-measurement on $A$, $\msf{N}=\{\{N^A_{a|x}\}_a\}_x$, satisfying operational identities $\mc{O}_{\rm all}({\msf{N}})=\{\{\alpha_{ax}\}_{ax}|\sum_{ax}\alpha_{ax}N^A_{a|x}=\mbb{0}\}$, describing the ideal identities that one is targeting in the experiment. Any set of operational identities of this form can be reduced to a finite {\em generating set}~\cite{schmid2024a}, meaning that all other operational identities can be deduced from this generating set, and NCOM-realizability relative to this generating set is equivalent to NCOM-realizability relative to the full set. Therefore, we consider only generating sets of operational identities of minimal cardinality. This cardinality, which we denote by $|T|$, coincides with the number of linearly independent operational identities in $\mc{O}_{\rm all}(\msf{N})$.   A given operational identity can be associated to a vector $\vec{\alpha} = (\alpha_{ax})_{ax}$. Consequently, a minimal generating set can be associated to a set of vectors 
$\{\vec{\alpha}^{(t)}\}_{t\in T}$.  Similar comments apply to the measurements on $B$.

The procedure for finding the best-fit probabilities to the relative frequencies obtained in the experiment~\cite{Lin2018, Mazurek2016}, i.e., the regularization procedure,  is detailed in the appendix~\ref{sec: methodC2}. %\footnote{We note that the procedure relies on an assumption about the dimension of the quantum (or GPT) system, although one can also do the fit in a way that {\em infers} the optimal GPT dimension from the experimental data using a train-and-test methodology rather than assuming it a priori~\cite{Grabowecky2022,Mazurek2021}.}
\blk 

As stated above, we refer to the measurements that come out of the regularization procedure as \textit{primary measurements} and denote them by $\msf{N}^{\text{pri}}= \{\{N^{A,\text{pri}}_{a|x}\}_a\}_x$, and $\msf{M}^{\text{pri}}= \{\{M^{B,\text{pri}}_{b|y}\}_b\}_y$ respectively. \blk

In general, these measurements do not satisfy the targeted operational identities. Consider the case of the measurements on $A$. \blk In general $\sum_{ax}\alpha
^{(t)}_{ax}N^{A,\text{pri}}_{a|x}\ne \mbb{0}$ for one or more $\vec{\alpha}^{\; (t)}$ in a minimal generating set for the targeted set of operational identities, $\mc{O}_{\rm all}(\msf{N})$.  However, it is possible to find a set of \textit{secondary measurements}, denoted $\msf{N}^{\text{sec}} = \{\{N^{A,\text{sec}}_{a|x}\}_a\}_x$, such that each of the secondary effects lies in the convex hull of the set of primary effects $\msf{N}^{\text{pri}}$, and where the secondary effects exactly satisfy the operational identities $\mc{O}_{\rm all}(\msf{N})$, i.e., \blk
\begin{align}
&N^{A,\text{sec}}_{a|x}=\sum_{a',x'}u_{a',x'}^{a,x}N^{A,\text{pri}}_{a'|x'}~~~~\forall a,x \label{eq:const N} \\ 
&\sum_{ax}\alpha^{(t)}_{ax}N^{A,\text{sec}}_{a|x}=\mbb{0}~~~~\forall t\in T \notag \\
&\sum_aN^{A,\text{sec}}_{a|x}=\mbb{1}^A~~~~\forall x \notag \\
&\sum_{a'x'}u_{a',x'}^{a,x}=1,~~u_{a',x'}^{a,x}\ge0~~~~\forall a,x,a',x'. \notag 
\end{align}
In particular, we solve an optimization problem that selects the closest approximation to the primary measurements while satisfying these constraints, by maximizing the objective function $C_N = \frac{1}{o\Delta_{\msf N}} \sum_{a,x} u_{a,x}^{a,x}$, that is, through the following linear programming
\begin{align}
   \max\quad\quad  &C_N= \frac{1}{o\Delta_{\msf N}} \sum_{ax}u^{ax}_{ax} \label{eq:sdp-secondar}\\
\text{such that}~~
&\sum_{axa',x'}u_{a',x'}^{a,x}N^{A,\text{pri}}_{a'|x'}\alpha^{(t)}_{ax}=\mbb{0}~~~~\forall t\in T \notag \\
&\sum_{aa',x'}u_{a',x'}^{a,x}N^{A,\text{pri}}_{a'|x'}=\mbb{1}^A~~~~\forall x \notag \\
&\sum_{a'x'}u_{a',x'}^{a,x}=1,~~u_{a',x'}^{a,x}\ge0~~~~\forall a,x,a',x'. \notag 
\end{align}
where the value of $C_N$ will be close to $1$ when the primary measurements are close to the ideal measurements being targeted. We define  the secondary measurements on $B$, denoted $\msf{M}^{\text{sec}}$, similarly.  That is, we define $M^{B,\text{sec}}_{b|y}=\sum_{b',y'}v_{b',y'}^{b,y}M^{B,\text{pri}}_{b'|y'}$  and  maximize the objective function  $C_M=\sum_{b,y} \frac{1}{o\Delta_{\msf M}} v_{b,y}^{b,y}$ while satisfying the analogues of the constraints in Eq.~\eqref{eq:const N}. 

\cmt 
With the tomographic completeness condition already made in the self-consistent tomography step, on Alice’s side, the assemblage on $A$ induced by Bob’s primary measurement,
i.e., the set of state $\{\tilde{\rho}^A_{b|y}:=\tr_B[(\mbb{1}^A\otimes M^{B,\text{pri}}_{b|y})\rho^{AB}]\}_{b,y}$ will be tomographically complete (relative to the operator subspace spanned by  $\msf{N}^{\text{pri}}$) so that the operational identities of $\msf{N}^{\text{pri}}$ can be inferred from the operational statistics.\footnote{ We note that, for cases like our Example~\ref{example1}, tomographic completeness is not even required, since the operational identities on measurements that one needs to enforce in the test are automatically satisfied by the condition that the effects of each measurement sum to the unit effect. However, for generality, we will always assume tomographic completeness in our entanglement certification protocols. }.  And similarly, with the roles of $\msf{N}^{\text{pri}}$ and $\msf{M}^{\text{pri}}$ reversed.
This condition is significantly milder than those requiring a prior characterization of the measurement devices, but it already allows us to infer the coefficients $u_{a',x'}^{a,x}$ purely from the {\em primary statistics}, using the following linear program: \blk
\begin{align}
    \max\quad\quad  &C_N= \frac{1}{o\Delta_{\msf N}} \sum_{ax}u^{ax}_{ax} \\
\text{such that}~~&\sum_{axa'x'}p^{\text{pri}}(a'b'|x'y')u_{a',x'}^{a,x}\alpha^{(t)}_{a,x}=0~~ \forall b',y',  \forall t\in T\notag \\
&\sum_{aa'x'}p^{\text{pri}}(a'b'|x'y')u_{a',x'}^{a,x}=p^{\text{pri}}(b'|y')~~ \forall x, b',y'.\notag \\
&\sum_{a'x'}u_{a',x'}^{a,x}=1,~~u_{a',x'}^{a,x}\ge0~~~~\forall a,x,a',x'. \notag
\label{eq:lp-secondar}
\end{align}
\cmt We note that the second condition resembles the normalization condition in Eq.~\eqref{eq:sdp-secondar}, and here $p^{\text{pri}}(b'|y'):=\sum_{a'}p^{\text{pri}}(a'b'|x'y')$ which is independent of $x'$ since the primary statistics are defined by Eq.~\eqref{eq:pri-stat}, and this form of statistics is necessarily nonsignalling. 
\blk

One can obtain the coefficients $v_{b',y'}^{b,y}$ in an analogous fashion. In the end, the \textit{statistics} defined by the secondary measurements, denoted  $p^{\text{sec}}(ab|xy)$ and termed the {\em secondary statistics}, are given by
\begin{align}
p^{\text{sec}}(ab|xy)&=\tr[N^{A,\text{sec}}_{a|x}\otimes M^{B,\text{sec}}_{b|y}\rho^{AB}] \notag \\
&=\sum_{a'b'x'y'}v_{b',y'}^{b,y}u_{a',x'}^{a,x}p^{\text{pri}}(a'b'|x'y'),
\end{align}
where $p^{\text{pri}}(a'b'|x'y')$ are the statistics defined by the set of primary measurements.  

\cmt 
\yujie 
\subsection{Proof of Proposition~\ref{prop:detector-inefficiency}}
\label{app:proof-detector-inefficiency}
\begin{proof}
We prove the two directions explicitly.

First suppose that the operational statistics $p(a,b|x,y)$ from the ideal Bell circuit $\mc B=(\msf N,\msf M,\rho^{AB})$  is NCOM-realizable relative to
$\mc O_{\rm aattt}(\mc B)$. Then there exists an ontic state space
$\Lambda_A\times\Lambda_B$, a distribution $p(\lambda_A\lambda_B)$, and local
response functions $p(a|x,\lambda_A)$ and $p(b|y,\lambda_B)$ that reproduce
$p(a,b|x,y)$ and respect the operational identities of $\mc O_{\rm all}(\msf N)$ and $\mc O_{\rm all}(\msf M)$ as in Eqs.~\eqref{eq:cl-bipartite} and~\eqref{oe1}.

One can then
define response functions for the inefficient multi measurements $\tilde{\msf{N}}$ on Alice by
\begin{subequations}
\label{eq:ideal-to-inef}
\begin{align}
p_\epsilon(a|x,\lambda_A)&=\epsilon_{a|x}p(a|x,\lambda_A),\\
p_\epsilon(\varnothing|x,\lambda_A)&=\sum_a(1-\epsilon_{a|x})p(a|x,\lambda_A),
\end{align}
\end{subequations}
and analogously for Bob. Since all efficiencies $\epsilon_{a|x}$ and $\epsilon_{b|y}$ are strictly positive and less than or equal to $1$,  \yujie these response functions are nonnegative and normalized. One can easily verify that \yujie they also respect the operational identities in $\mc{O}_{\rm all}(\tilde{\msf N})$ and
$\mc{O}_{\rm all}(\tilde{\msf M})$. To see this, consider an arbitrary operational identity among
Alice's \yz measurement effects\yujie,
\begin{align}
    \sum_{\alpha,x}\tilde\gamma_{\alpha,x}\tilde N^A_{\alpha|x}=0,
\end{align}
where $\alpha$ ranges over the original outcome set together with the null outcome $\varnothing$.
Using Eqs.~\eqref{eq:inefficient-effects} and
\eqref{eq:inefficient-null}, this identity can be equivalently rewritten as
\begin{align}
\sum_{a,x}\left[\epsilon_{a|x}\tilde\gamma_{a,x}+(1-\epsilon_{a|x})\tilde\gamma_{\varnothing,x}\right]N^A_{a|x}=0.
\end{align}
Thus, every identity among the attenuated effects and null effect pulls back to an identity
among the ideal effects. By assumption, the ideal response functions
$p(a|x,\lambda_A)$ respect all operational identities in $\mc O_{\rm all}(\msf N)$, so we have
\begin{align}
\sum_{a,x}\left[\epsilon_{a|x}\tilde\gamma_{a,x}+(1-\epsilon_{a|x})\tilde\gamma_{\varnothing,x}\right]p(a|x,\lambda_A)=0.
\end{align}
Using Eq.~\eqref{eq:ideal-to-inef}, this is precisely
\begin{align}
\sum_{\alpha,x}\tilde\gamma_{\alpha,x}p_\epsilon(\alpha|x,\lambda_A)=0,
\end{align}
Thus the inefficient response functions respect all the operational identities in $\mc O_{\rm all}(\tilde{\msf N})$. The same argument applies to Bob. Finally, it is straightforward to check that using the same distribution $p(\lambda_A\lambda_B)$, the inefficient response functions in Eq.~\eqref{eq:ideal-to-inef} reproduce the inefficient statistics. Hence $p_\epsilon(\alpha,\beta|x,y)$ is NCOM-realizable relative to $\mc O_{\rm aattt}(\mc B^\epsilon)$.

Conversely, suppose that $p_\epsilon(\alpha,\beta|x,y)$ from the inefficient Bell circuit $\mc B^\epsilon=(\tilde{\msf N},\tilde{\msf M},\rho^{AB})$ is NCOM-realizable relative to $\mc O_{\rm aattt}(\mc B^\epsilon)$. Then there exist
$p_\epsilon(\lambda_A\lambda_B)$ and local response functions
$p_\epsilon(\alpha|x,\lambda_A)$ and $p_\epsilon(\beta|y,\lambda_B)$ reproducing the inefficient statistics and respecting the operational identities of $\tilde{\msf N}$ and $\tilde{\msf M}$. Since all efficiencies $\epsilon_{a|x}$ and $\epsilon_{b|y}$ are strictly positive, we can define response functions for the ideal multi-measurement $\msf N$ and $\msf M$ by
\begin{subequations}
\begin{align}
p(a|x,\lambda_A)&=\frac{p_\epsilon(a|x,\lambda_A)}{\epsilon_{a|x}}, \label{eq:ap-at1}\\
p(b|y,\lambda_B)&=\frac{p_\epsilon(b|y,\lambda_B)}{\epsilon_{b|y}}.
\end{align}
\label{eq:ideal-response}
\end{subequations}
These functions are clearly nonnegative. They are normalized because the operational identity defining the \yz null-outcome \yujie effect $\tilde N^A_{\varnothing|x}=\sum_a\frac{1-\epsilon_{a|x}}{\epsilon_{a|x}}\tilde N^A_{a|x}$ in Eqs.~\eqref{eq:at-outcome1} and~\eqref{eq:null-outcome1} implies
\begin{align}
p_\epsilon(\varnothing|x,\lambda_A)=\sum_a\frac{1-\epsilon_{a|x}}{\epsilon_{a|x}}p_\epsilon(a|x,\lambda_A).
\label{eq:ap-null-at}
\end{align}
Combining this with the normalization condition
\begin{align}
\sum_a p_\epsilon(a|x,\lambda_A)+p_\epsilon(\varnothing|x,\lambda_A)=1
\label{eq:ap-norm}
\end{align}
gives
\begin{align}
1&=\sum_a p_\epsilon(a|x,\lambda_A)+\sum_a\frac{1-\epsilon_{a|x}}{\epsilon_{a|x}}p_\epsilon(a|x,\lambda_A)\notag \\
&=\sum_a\frac{p_\epsilon(a|x,\lambda_A)}{\epsilon_{a|x}}{=}\sum_a p(a|x,\lambda_A)
\end{align}
where the first equality follows from Eqs.~\eqref{eq:ap-norm} and ~\eqref{eq:ap-null-at} , and the last equality follows from Eq.~\eqref{eq:ap-at1}.

The reconstructed response functions also respect the operational identities in $\mc{O}_{\rm all}(\msf N)$ and $\mc{O}_{\rm all}(\msf M)$. This is shown as follows.  First, note that 
%Indeed, 
\begin{align}
    \sum_{a,x}\gamma_{a,x}N^A_{a|x}=0\Rightarrow \sum_{a,x}
    \frac{\gamma_{a,x}}{\epsilon_{a|x}}\tilde N^A_{a|x}=0.
\end{align}
Since the inefficient response functions respect the identity $ \sum_{a,x}\frac{\gamma_{a,x}}{\epsilon_{a|x}}p_{\epsilon}(a|x,\lambda_A)=0$, it follows that
\begin{align}
    \sum_{a,x}\gamma_{a,x}p(a|x,\lambda_A)=0.
\end{align}
The same argument holds for Bob. 

Finally, it is straightforward to check that using the same distribution
$p_{\epsilon}(\lambda_A\lambda_B)$, the reconstructed ideal response functions in Eq.~\eqref{eq:ideal-response} reproduce the ideal statistics. Hence $p(a,b|x,y)$ is NCOM-realizable relative to $\mc O_{\rm aattt}(\mc B)$.
\end{proof}
\blk

\section{Experimental details}
\label{sec: methodD}
\begin{figure*}[t]
    \centering
    \includegraphics[width=0.9\linewidth]{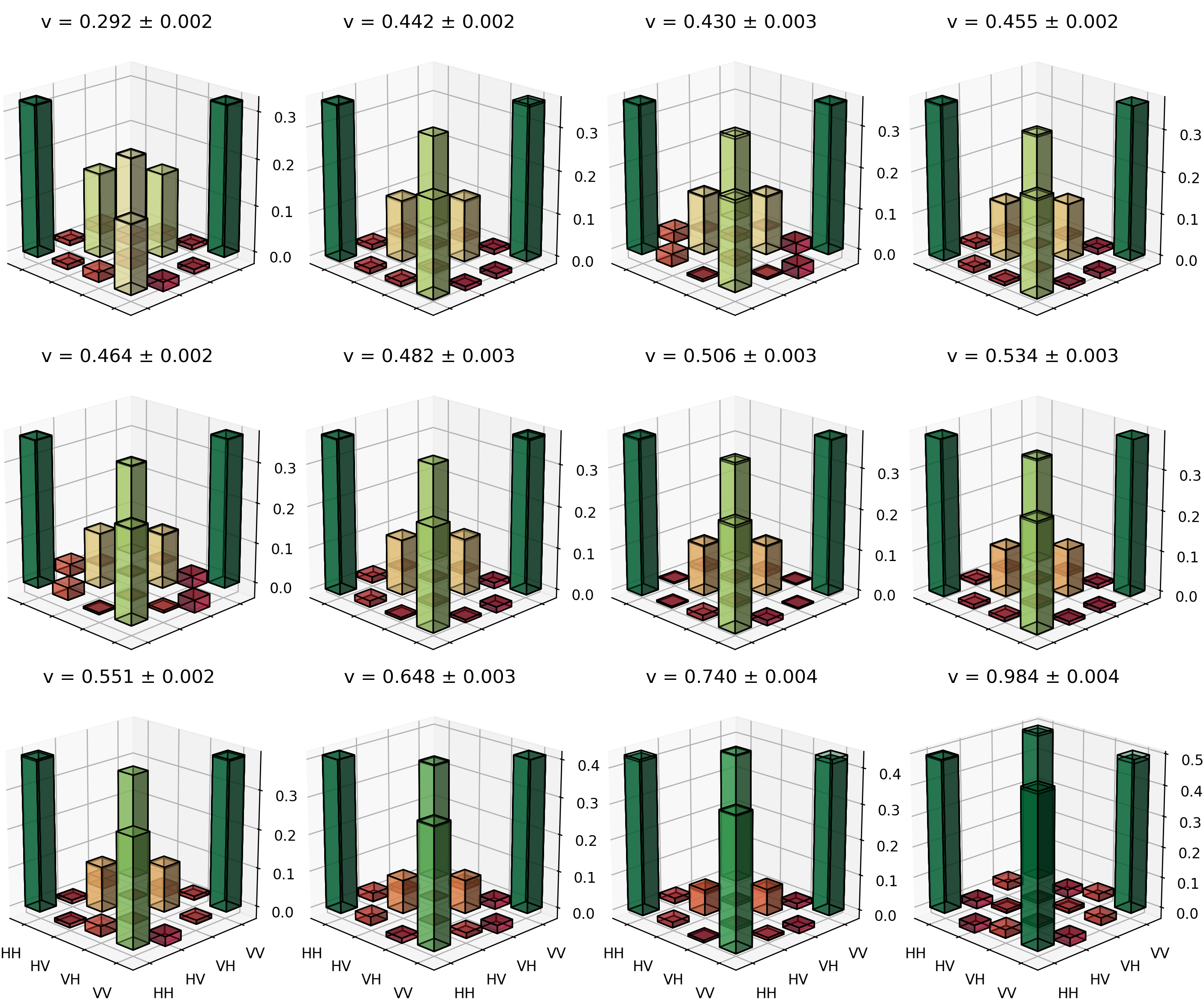}
    \caption{Quantum state tomography results: the real part of the reconstructed density matrix is shown, while the imaginary components are negligible, with values below $0.02$. The visibility $v$ denotes the weight of the closest isotropic state, with error bars estimated from 100 Monte Carlo trials accounting for both statistical and systematic uncertainties (see text below).}
\end{figure*}
\subsection{Generating two-qubit isotropic states experimentally }
\label{sec: methodD1}
We generate polarization-entangled photons using type-II spontaneous parametric down-conversion inside a Sagnac interferometer consisting of a $404~\text{nm}$ continuous-wave (cw) diode laser and a $10 \text{mm}$ PPKTP crystal. The source aims to generate the Bell state $\ket{\Phi^+} = \frac{1}{\sqrt{2}}\left(\ket{HH} + \ket{VV}\right)$, and is implemented with measured state fidelity greater than $97.5\%$ at a rate of {40000} coincidences/second for a laser power of 4 $\text{mW}$. The photons are then sent through a depolarization channel composed of two liquid crystal retarders (LCR), which depolarize the maximally entangled state to a high-fidelity isotropic state by producing a probabilistic mixture of Bell states as in Table~\ref{table:LCR}.
\begin{table}[ht]
\caption{Isotropic state generation}
\label{tab:wernerscheme}
\begin{ruledtabular}
\begin{tabular}{lcc}\label{table:LCR}
Gate & Probability & Output \\
\hline
$\hat{\mbb{1}} \otimes \hat{\mbb{1}}\hat{\mbb{1}}$ & $v^{\text{set}}+\frac{1-v^{\text{set}}}{4}$ & $\ket{\Phi^+} \propto \ket{HH} + \ket{VV}$ \\

$\hat{\mbb{1}} \otimes \hat{\mbb{1}} \hat{Z}$ & $\frac{1-v^{\text{set}}}{4}$ & $\ket{\Phi^-} \propto \ket{HH} - \ket{VV}$ \\

$\hat{\mbb{1}} \otimes \hat{X}\hat{\mbb{1}}$& $\frac{1-v^{\text{set}}}{4}$ & $\ket{\Psi^+} \propto \ket{HV} + \ket{VH}$ \\

$\hat{\mbb{1}} \otimes \hat{X} \hat{Z}$ & $\frac{1-v^{\text{set}}}{4}$ & $\ket{\Psi^-} \propto\ket{HV} - \ket{VH}$ \\
\end{tabular}
\end{ruledtabular}
\end{table}
\par  The produced isotropic state typically has a value of the visibility $v$ that deviates slightly from the value  $v^{\text{set}}$ that was targeted due to the initial state’s non-unit fidelity with $\ket{\Phi^+}$ \blk and the LCR’s limited switching speed (80 ms hold time per setting, with 5 ms rise time and 20 ms fall time). Instead of relying solely on the targeted value 
$v^{\text{set}}$, we estimated the realized value using the measured fidelity to the nearest isotropic state. In Table~\ref{tab:tomoresults}, we compare the targeted and estimated values, showing that the fidelity between the tomographically reconstructed state and the nearest isotropic state is consistently near $99\%$.

\begin{table*}[htbp]
\centering
\begin{tabular}{|c|c|c|c|r|}
\hline
$v^{\text{set}}$ & $\mc{F}(\rho_{\text{iso}
}^{v^{\text{set}}},\rho^{\text{tomo}})$  & $v$ & $\mc{F}(\rho_{\text{iso}
}^{v},\rho^{\text{tomo}})$ &  \multicolumn{1}{c|}{$\mc{I}$}\\
\hline
0.270& $0.9949\pm 0.0007$& $0.290\pm 0.002$&   $0.9950\pm 0.0006$ & $0.04513\pm 0.00015$\\
0.420& $0.9965\pm 0.0009$&  $0.441\pm 0.002$&  $0.9967\pm 0.0009$ & $0.00578\pm 0.00021$\\
0.430& $0.9892\pm 0.0010$&  $0.430\pm 0.003$&   $0.9892\pm 0.0010$ & $0.01068\pm 0.00019$\\
\rowcolor{red!20}
0.440& $0.9966\pm 0.0008$&  $0.454\pm 0.002$&   $0.9963\pm 0.0008$ &  $-0.00117\pm 0.00022$\\
\rowcolor{red!20}
0.450& $0.9934\pm 0.0011$& $0.465\pm 0.002$&  $0.9936\pm 0.0011$ & $-0.00409\pm 0.00018 $\\
\rowcolor{yellow!20}
0.480& $0.9938\pm 0.0012$& $0.483\pm 0.003$&   $0.9938\pm 0.0012$ & $-0.00585\pm 0.00039$\\
\rowcolor{yellow!20}
0.500& $0.9962\pm 0.0011$&  $0.505\pm 0.003$&   $0.9962\pm 0.0011$& $-0.01207\pm 0.00018$\\
\rowcolor{yellow!20}
0.520& $0.9935\pm 0.0013$&  $0.534\pm 0.003$&  $0.9936\pm 0.0013$ & $-0.01633\pm 0.00016$\\
\rowcolor{yellow!20}
0.540& $0.9957\pm 0.0010$&  $0.551\pm 0.002$&  $0.9958\pm 0.0010$ & $-0.02583\pm 0.00022$\\
\rowcolor{yellow!20}
0.650& $0.9911\pm 0.0017$&  $0.651\pm 0.003$&  $0.9911\pm 0.0017$& $-0.05280\pm 0.00011$\\
\rowcolor{yellow!20}
0.750& $0.9869\pm 0.0029$&  $0.742\pm 0.004$&  $0.9870\pm 0.0030$& $-0.07849\pm 0.00012$\\
\rowcolor{yellow!20}
1.000& $0.9689\pm 0.0036$&  $0.982\pm 0.004$&  $0.9802\pm 0.0043$& $-0.14256\pm 0.00013$\\
\hline
\end{tabular}
\caption{Experimental results for 12 tomographically reconstructed quantum states for which the parameter value being targeted is 
$v^{\text{set}}$ and the parameter of the closest isotropic state is $v$. \blk In red: Unsteerable states that violate the inequality $\mc{I}\ge0$ in the main text; In yellow: States that violate the inequality in the main text but for which the entanglement could in principle be certified by violation of a steering inequality. \blk The uncertainties in the violations arise from Poissonian statistics, whereas the uncertainties in $v$ and in fidelities are calculated from Monte Carlo simulations of the different state tomography results, as described in the text. 
}
\label{tab:tomoresults}
\end{table*}
\subsection{Measurement technique}
\label{sec: methodD2}

After preparing each isotropic state, we record coincidence events from both the transmitted and reflected output ports using single‐photon avalanche diodes (SPCM‐AQRH‐11‐FC, Excelitas) with a 1 ns coincidence window, {and a 5 s measurement time per measurement configuration}. We measure coincidence counts of the two photons in measurement configurations corresponding to the vertices of the shapes associated to the set of measurements being targeted. It is important to note that the polarizing beam splitters have different efficiencies for their transmitted and reflected ports, which has the potential to add non-isotropic noise to our state. To mitigate this noise, we follow the measurement scheme outlined in Table~\ref{tab:measureiterate}, where instead of assigning specific ports to specific effects, we iterate over the four permutations of ports and effects by adjusting the wave plates accordingly. Next, for each effect, we take the four trials (one per port configuration) and sum the coincidence counts for the desired effect. Finally, we divide each coincidence by the sum of the four to convert counts into relative frequencies for one measurement configuration. In total, $240$ distinct measurement configurations are implemented in the same way, resulting in a total integration time of approximately 2 hours. These raw relative frequencies are then converted to the raw frequency corresponding to the $60$ measurement settings (effectively, each setting is measured in four different experimental configurations).  

\begin{table}[h]
\caption{Measurement iteration scheme for setting $xy$}
\label{tab:measureiterate}
\begin{ruledtabular}
\begin{tabular}{c|c|c|c|c}
Iter\# & Alice Trans. & Alice Refl. & Bob Trans. & Bob Refl. \\
\hline
1 & $N_{0|x}$ &  $N_{1|x}$ & $M_{0|y}$ & $M_{1|y}$ \\
2 & $N_{1|x}$ &  $N_{0|x}$ & $M_{1|y}$ & $M_{0|y}$ \\
3 & $N_{0|x}$ &  $N_{1|x}$ & $M_{1|y}$ & $M_{0|y}$ \\
4 & $N_{1|x}$ &  $N_{0|x}$ & $M_{0|y}$ & $M_{1|y}$ \\

\end{tabular}
\end{ruledtabular}
\end{table}
\subsection{Experimental error analysis}
\label{sec: methodD3}
The experimental uncertainty in the noncontextual bound $\mc{I}$ is estimated using a Monte Carlo simulation with 100 independent trials, where each trial samples from a Poisson distribution based on the observed raw counts. To further quantify the uncertainties in the estimated state visibility $v$ and to certify the unsteerability of the states, we also account for systematic errors in the wave plate settings by drawing these from a normal distribution, similarly as reported in Ref.~\cite{Aguilar2024}. 

Specifically, we consider two primary sources of systematic error: (1) miscalibration or model imperfection of the measurement wave plates, modeled as a Gaussian distribution over phase shifts \blk with a standard deviation of $1^{\circ}$, and (2) phase shifts introduced by {model} imperfections 
also assumed to have an uncertainty of $1^{\circ}$. For each trial, the reconstructed density matrix is used to estimate the isotropic state visibility $v$, and the standard deviation across all trials is reported as the final uncertainty. Additional details and data are summarized in the Appendix~\ref{sec: methodD}.
\begin{table}[htbp]
\centering
\caption{Examples of experimental results on optimal secondary measurements:
For each quantum state involved in the test, we report the optimal values of the parameters $C_N$ and $C_M$ obtained during the construction of the secondary measurements.}
\label{tab:computed}
\begin{tabular}{|c|c|c|}
%\toprule
\hline
$\yujie v^{\text{set}} \blk$ & $C_M$ & $C_N$ \\
\hline
0.270 & 0.959 & 0.937 \\
0.420 & 0.958 & 0.980 \\
0.430 & 0.946 & 0.972 \\
0.440 & 0.950 & 0.982 \\
0.450 & 0.960 & 0.982 \\
0.480 & 0.955 & 0.957\\
0.500 & 0.956 & 0.969 \\
0.520 & 0.956 & 0.972 \\
0.540 & 0.959 & 0.971 \\
0.650 & 0.966 & 0.981 \\
0.750 & 0.972 & 0.985 \\
1.000 & 0.970 & 0.985 \\
\hline
\end{tabular}
\end{table}
\subsection{Detailed Experimental Results}
\label{sec: methodD4}
The validity of our experimental outcomes—and the subsequent violation of the inequality under study—is inherently linked to the precision of our isotropic state preparation. Recall that the fidelity $\mc{F}$ between states $\rho$ and $\sigma$ is defined as follows:
\begin{equation}
    \mathcal{F}(\rho, \sigma) = \tr\left(\sqrt{\sqrt{\rho}\sigma\sqrt{\rho}} \right)^2.
\end{equation}
\blk 
While the main text highlights a fidelity between our realized state and the best-matched isotropic state that is near $99\%$, this section provides a detailed evaluation of our measured fidelity with respect to both the best-matched isotropic state 
and the state we targeted. 
To rigorously assess the quality of our state preparation, we performed quantum state tomography using the same experimental setup described earlier. Measurements were conducted across all $12 \times 20$ basis settings, with the corresponding Bloch vectors arranged to realize the geometries of the icosahedron and dodecahedron, and the quantum state tomography was done by minimizing the $L^2$ norm 
\begin{equation}
     L^2= |f_{i} - \bra{\psi_{i}}\rho\ket{\psi_{i}}|^2,
\end{equation}
where $\ket{\psi_{i}}$ represents the $i^{th}$ measurement setting, and $f_i$ is the corresponding relative frequency. The result of this optimization yields a density operator, which is denoted ${\rho}^{\text{tomo}}$.
\par 
To address potential mischaracterizations of ${\rho}^{\text{tomo}}$ due to statistical and systematic errors, we conducted a comprehensive Monte Carlo simulation comprising 100 trials. In this simulation: (1) Raw count data were modeled using Poissonian statistics;  (2) The calibration of the measurement wave-plate angles was simulated using a normal distribution over phase-shifts \blk with a standard deviation of $1^{\circ}$; (3) Variations in wave-plate thickness—leading to phase-shift errors—were also represented by a normal distribution with $1^{\circ}$ of phase shift error.

After reconstructing the density operator via tomography, we computed its fidelity with the target isotropic state (defined by visibility $v^{\text{set}}$), and the closest isotropic state (defined by visibility $v$). 

The results, summarized in Table~\ref{tab:tomoresults}, confirm that for each of the 12 values of $v^{\text{set}}$ we considered, our realized state had a fidelity with our targeted state and a fidelity with the best-matched isotropic state of near \blk % exceeding 
$99\%$. 

\section{Certifying unsteerability for states}
\label{sec: methodE}
In our experiment, the prepared state is never exactly an isotropic state (for which unsteerability and locality are well-studied~\cite{Werner1989,Wiseman2007}). Therefore, in addition to noting the high fidelity of the states we prepare with an unsteerable state, we provide a direct demonstration of the unsteerability of some of the states we realized in the experiment. 

To achieve this, we employ recent findings in Ref.~\cite{zhang2024,Renner2024}, which show that the isotropic state
$\rho_{\text{iso}}^v=v\op{\Phi^+}{\Phi^+}+(1-v)\frac{\mbb{1}\otimes\mbb{1}}{4}$
is unsteerable for $v\le \frac{1}{2}$. 
We make use of each of the following facts: \\
(1) any state $\rho=(\mbb{1}\otimes U)\rho_{\text{iso}}^{v}
(\mbb{1}\otimes U^{\dagger})$ with $v\le \frac{1}{2}$
remains unsteerable for any unitary $U$; 
(2) all separable states are unsteerable; 
(3) any convex combination of unsteerable states is unsteerable.

\noindent We can witness the unsteerability of our realized state $\rho$ if we find feasible solutions of the following semidefinite program:
\begin{align*}
&\quad X \succeq 0,\quad X^{T_B} \succeq 0, \quad p_i \;\ge\; 0\;\;\forall i,\\
&\quad \sum_{i=1}^n p_i  \rho_i \;+\; X \;=\;\rho,
\end{align*}
where $\rho_i=(\mbb{1}\otimes U_i)\rho^v_{\text{iso}}(\mbb{1}\otimes U_i^{\dagger})$ with $v\le 1/2$ for some $U_i$ sampled from the Haar measure of $SU(2)$ and $X$ is any two-qubit quantum state with a positive partial transpose (and so is separable).  That is, the feasibility of this program is a sufficient condition for unsteerability.

Moreover, following an argument similar to that in Ref.~\cite{Aguilar2024}, the fact (1) can be replaced with a more general statement for any positive map; that is, 
any state 
$\rho=(\mbb{1}\otimes \Lambda)(\rho^v_{\text{iso}})\ge 0$  where $v\le 1/2$ remains unsteerable 
for any positive map $\Lambda$~\cite{Aguilar2024}. Using this condition leads to a slightly different SDP that is generally capable of certifying a broader class of unsteerable states.

The numerics show that two of the states we realized in our experiment, those with estimated isotropic visibilities $v =0.454$ and $v =0.465$, yield feasible solutions to the SDP (even after Monte Carlo simulations accounting for small measurement imperfections), thereby confirming their unsteerability. However, the state with estimated isotropic parameter $v =0.483$ fails this test and thus cannot be concluded to be 
unsteerable even though it is close to an unsteerable isotropic state.

\bibliographystyle{apsrev4-2} 
\bibliography{ref}

@article{Tavakoli2020,
  title = {Measurement incompatibility and steering are necessary and sufficient for operational contextuality},
  author = {Tavakoli, Armin and Uola, Roope},
  journal = {Phys. Rev. Res.},
  volume = {2},
  issue = {1},
  pages = {013011},
  numpages = {7},
  year = {2020},
  month = {Jan},
  publisher = {American Physical Society},
  doi = {10.1103/PhysRevResearch.2.013011},
  url = {https://link.aps.org/doi/10.1103/PhysRevResearch.2.013011}
}

@article{schmid2024shadowssubsystemsgeneralizedprobabilistic,
  doi = {10.22331/q-2025-10-13-1880},
  url = {https://doi.org/10.22331/q-2025-10-13-1880},
  title = {Shadows and subsystems of generalized probabilistic theories: when tomographic incompleteness is not a loophole for contextuality proofs},
  author = {Schmid, David and Selby, John H. and Rossi, Vinicius P. and Baldij{\~{a}}o, Roberto D. and Sainz, Ana Bel{\'{e}}n},
  journal = {{Quantum}},
  issn = {2521-327X},
  publisher = {{Verein zur F{\"{o}}rderung des Open Access Publizierens in den Quantenwissenschaften}},
  volume = {9},
  pages = {1880},
  month = oct,
  year = {2025}
}

@misc{pusey2019contextualityaccesstomographicallycomplete,
      title={Contextuality without access to a tomographically complete set}, 
      author={Matthew F. Pusey and Lídia del Rio and Bettina Meyer},
      year={2019},
      eprint={1904.08699},
      archivePrefix={arXiv},
      primaryClass={quant-ph},
      url={https://arxiv.org/abs/1904.08699}, 
}

@Article{Elben2023,
author={Elben, Andreas
and Flammia, Steven T.
and Huang, Hsin-Yuan
and Kueng, Richard
and Preskill, John
and Vermersch, Beno{\^i}t
and Zoller, Peter},
title={The randomized measurement toolbox},
journal={Nature Reviews Physics},
year={2023},
month={Jan},
day={01},
volume={5},
number={1},
pages={9-24},
issn={2522-5820},
doi={10.1038/s42254-022-00535-2},
url={https://doi.org/10.1038/s42254-022-00535-2}
}

@Article{Huang2020,
author={Huang, Hsin-Yuan
and Kueng, Richard
and Preskill, John},
title={Predicting many properties of a quantum system from very few measurements},
journal={Nature Physics},
year={2020},
month={Oct},
day={01},
volume={16},
number={10},
pages={1050-1057},
issn={1745-2481},
doi={10.1038/s41567-020-0932-7},
url={https://doi.org/10.1038/s41567-020-0932-7}
}

@article{Kofman2008,
  title = {Bell-inequality violation versus entanglement in the presence of local decoherence},
  author = {Kofman, A. G. and Korotkov, A. N.},
  journal = {Phys. Rev. A},
  volume = {77},
  issue = {5},
  pages = {052329},
  numpages = {4},
  year = {2008},
  month = {May},
  publisher = {American Physical Society},
  doi = {10.1103/PhysRevA.77.052329},
  url = {https://link.aps.org/doi/10.1103/PhysRevA.77.052329}
}

@article{Peres1996,
  title = {Separability Criterion for Density Matrices},
  author = {Peres, Asher},
  journal = {Phys. Rev. Lett.},
  volume = {77},
  issue = {8},
  pages = {1413--1415},
  numpages = {0},
  year = {1996},
  month = {Aug},
  publisher = {American Physical Society},
  doi = {10.1103/PhysRevLett.77.1413},
  url = {https://link.aps.org/doi/10.1103/PhysRevLett.77.1413}
}

@book{Wenninger1979,
  author    = {Magnus J. Wenninger},
  title     = {Spherical Models},
  publisher = {Cambridge University Press},
  year      = {1979},
}

@article{Horodecki1995,
title = {Violating Bell inequality by mixed spin-12 states: necessary and sufficient condition},
journal = {Physics Letters A},
volume = {200},
number = {5},
pages = {340-344},
year = {1995},
issn = {0375-9601},
doi = {https://doi.org/10.1016/0375-9601(95)00214-N},
url = {https://www.sciencedirect.com/science/article/pii/037596019500214N},
author = {R. Horodecki and P. Horodecki and M. Horodecki}
}

@article{chiribellaprob,
  title = {Probabilistic theories with purification},
  author = {Chiribella, Giulio and D'Ariano, Giacomo Mauro and Perinotti, Paolo},
  journal = {Phys. Rev. A},
  volume = {81},
  issue = {6},
  pages = {062348},
  numpages = {40},
  year = {2010},
  month = {Jun},
  publisher = {American Physical Society},
  doi = {10.1103/PhysRevA.81.062348},
  url = {https://link.aps.org/doi/10.1103/PhysRevA.81.062348}
}

@article{McMullen1970,
  author  = {McMullen, P.},
  title   = {The maximum numbers of faces of a convex polytope},
  journal = {Mathematika},
  volume  = {17},
  number  = {2},
  pages   = {179--184},
  year    = {1970},
  doi     = {10.1112/S0025579300002850}
}

@article{Wood_2015,
doi = {10.1088/1367-2630/17/3/033002},
url = {https://dx.doi.org/10.1088/1367-2630/17/3/033002},
year = {2015},
month = {mar},
publisher = {IOP Publishing},
volume = {17},
number = {3},
pages = {033002},
author = {Wood, Christopher J and Spekkens, Robert W},
title = {The lesson of causal discovery algorithms for quantum correlations: causal explanations of {Bell}-inequality violations require fine-tuning},
journal = {New Journal of Physics}
}

@article{Schmid2024add,
  title = {Addressing some common objections to generalized noncontextuality},
  author = {Schmid, David and Selby, John H. and Spekkens, Robert W.},
  journal = {Phys. Rev. A},
  volume = {109},
  issue = {2},
  pages = {022228},
  numpages = {18},
  year = {2024},
  month = {Feb},
  publisher = {American Physical Society},
  doi = {10.1103/PhysRevA.109.022228},
  url = {https://link.aps.org/doi/10.1103/PhysRevA.109.022228}
}

@article{Selby2023,
  title = {Contextuality without Incompatibility},
  author = {Selby, John H. and Schmid, David and Wolfe, Elie and Sainz, Ana Bel\'en and Kunjwal, Ravi and Spekkens, Robert W.},
  journal = {Phys. Rev. Lett.},
  volume = {130},
  issue = {23},
  pages = {230201},
  numpages = {7},
  year = {2023},
  month = {Jun},
  publisher = {American Physical Society},
  doi = {10.1103/PhysRevLett.130.230201},
  url = {https://link.aps.org/doi/10.1103/PhysRevLett.130.230201}
}

@article{Spekkens2005,
  title = {Contextuality for preparations, transformations, and unsharp measurements},
  author = {Spekkens, R. W.},
  journal = {Phys. Rev. A},
  volume = {71},
  issue = {5},
  pages = {052108},
  numpages = {17},
  year = {2005},
  month = {May},
  publisher = {American Physical Society},
  doi = {10.1103/PhysRevA.71.052108},
  url = {https://link.aps.org/doi/10.1103/PhysRevA.71.052108}
}

@article{David2018,
  title = {All the noncontextuality inequalities for arbitrary prepare-and-measure experiments with respect to any fixed set of operational equivalences},
  author = {Schmid, David and Spekkens, Robert W. and Wolfe, Elie},
  journal = {Phys. Rev. A},
  volume = {97},
  issue = {6},
  pages = {062103},
  numpages = {12},
  year = {2018},
  month = {Jun},
  publisher = {American Physical Society},
  doi = {10.1103/PhysRevA.97.062103},
  url = {https://link.aps.org/doi/10.1103/PhysRevA.97.062103}
}

@article{Horodecki2009,
  title = {Quantum entanglement},
  author = {Horodecki, Ryszard and Horodecki, Pawe\l{} and Horodecki, Micha\l{} and Horodecki, Karol},
  journal = {Rev. Mod. Phys.},
  volume = {81},
  issue = {2},
  pages = {865--942},
  numpages = {0},
  year = {2009},
  month = {Jun},
  publisher = {American Physical Society},
  doi = {10.1103/RevModPhys.81.865},
  url = {https://link.aps.org/doi/10.1103/RevModPhys.81.865}
}

@article{Brunner2014,
  title = {{Bell} nonlocality},
  author = {Brunner, Nicolas and Cavalcanti, Daniel and Pironio, Stefano and Scarani, Valerio and Wehner, Stephanie},
  journal = {Rev. Mod. Phys.},
  volume = {86},
  issue = {2},
  pages = {419--478},
  numpages = {60},
  year = {2014},
  month = {Apr},
  publisher = {American Physical Society},
  doi = {10.1103/RevModPhys.86.419},
  url = {https://link.aps.org/doi/10.1103/RevModPhys.86.419}
}

@article{Uola2020,
  title = {Quantum steering},
  author = {Uola, Roope and Costa, Ana C. S. and Nguyen, H. Chau and G\"uhne, Otfried},
  journal = {Rev. Mod. Phys.},
  volume = {92},
  issue = {1},
  pages = {015001},
  numpages = {40},
  year = {2020},
  month = {Mar},
  publisher = {American Physical Society},
  doi = {10.1103/RevModPhys.92.015001},
  url = {https://link.aps.org/doi/10.1103/RevModPhys.92.015001}
}

@article{Piani2011,
  title = {All Nonclassical Correlations Can Be Activated into Distillable Entanglement},
  author = {Piani, Marco and Gharibian, Sevag and Adesso, Gerardo and Calsamiglia, John and Horodecki, Pawe\l{} and Winter, Andreas},
  journal = {Phys. Rev. Lett.},
  volume = {106},
  issue = {22},
  pages = {220403},
  numpages = {4},
  year = {2011},
  month = {Jun},
  publisher = {American Physical Society},
  doi = {10.1103/PhysRevLett.106.220403},
  url = {https://link.aps.org/doi/10.1103/PhysRevLett.106.220403}
}

@article{Werner1989,
  title = {Quantum states with {Einstein}-{Podolsky}-{Rosen} correlations admitting a hidden-variable model},
  author = {Werner, Reinhard F.},
  journal = {Phys. Rev. A},
  volume = {40},
  issue = {8},
  pages = {4277--4281},
  numpages = {0},
  year = {1989},
  month = {Oct},
  publisher = {American Physical Society},
  doi = {10.1103/PhysRevA.40.4277},
  url = {https://link.aps.org/doi/10.1103/PhysRevA.40.4277}
}

@article{Wiseman2007,
  title = {Steering, Entanglement, Nonlocality, and the {Einstein}-{Podolsky}-{Rosen} Paradox},
  author = {Wiseman, H. M. and Jones, S. J. and Doherty, A. C.},
  journal = {Phys. Rev. Lett.},
  volume = {98},
  issue = {14},
  pages = {140402},
  numpages = {4},
  year = {2007},
  month = {Apr},
  publisher = {American Physical Society},
  doi = {10.1103/PhysRevLett.98.140402},
  url = {https://link.aps.org/doi/10.1103/PhysRevLett.98.140402}
}

@article{Bell1964,
  title = {On the {Einstein} {Podolsky} {Rosen} paradox},
  author = {Bell, J. S.},
  journal = {Physics Physique Fizika},
  volume = {1},
  issue = {3},
  pages = {195--200},
  numpages = {6},
  year = {1964},
  month = {Nov},
  publisher = {American Physical Society},
  doi = {10.1103/PhysicsPhysiqueFizika.1.195},
  url = {https://link.aps.org/doi/10.1103/PhysicsPhysiqueFizika.1.195}
}

@article{Guhne2009,
title = {Entanglement detection},
journal = {Physics Reports},
volume = {474},
number = {1},
pages = {1-75},
year = {2009},
issn = {0370-1573},
doi = {https://doi.org/10.1016/j.physrep.2009.02.004},
url = {https://www.sciencedirect.com/science/article/pii/S0370157309000623},
author = {Otfried Gühne and Géza Tóth}
}

@article{Skrzypczyk2014,
  title = {Quantifying {Einstein}-{Podolsky}-{Rosen} Steering},
  author = {Skrzypczyk, Paul and Navascu\'es, Miguel and Cavalcanti, Daniel},
  journal = {Phys. Rev. Lett.},
  volume = {112},
  issue = {18},
  pages = {180404},
  numpages = {5},
  year = {2014},
  month = {May},
  publisher = {American Physical Society},
  doi = {10.1103/PhysRevLett.112.180404},
  url = {https://link.aps.org/doi/10.1103/PhysRevLett.112.180404}
}

@article{zhang2023,
  doi = {10.22331/q-2025-10-31-1902},
  url = {https://doi.org/10.22331/q-2025-10-31-1902},
  title = {Cost of {S}imulating {E}ntanglement in {S}teering {S}cenarios},
  author = {Zhang, Yujie and Zhang, Jiaxuan and Chitambar, Eric},
  journal = {{Quantum}},
  issn = {2521-327X},
  publisher = {{Verein zur F{\"{o}}rderung des Open Access Publizierens in den Quantenwissenschaften}},
  volume = {9},
  pages = {1902},
  month = oct,
  year = {2025}
}

@article{alon1985,
title = {A Simple Proof of the Upper Bound Theorem},
journal = {European Journal of Combinatorics},
volume = {6},
number = {3},
pages = {211-214},
year = {1985},
issn = {0195-6698},
doi = {https://doi.org/10.1016/S0195-6698(85)80029-9},
url = {https://www.sciencedirect.com/science/article/pii/S0195669885800299},
author = {N. Alon and G. Kalai}
}

@article{Merkel2013,
  author = {Merkel, Seth T. and Gambetta, Jay M. and Smolin, John A. and Poletto, Stefano and C{\'o}rcoles, Antonio D. and Johnson, Blake R. and Ryan, Colm A. and Steffen, Matthias},
  title = {Self-consistent quantum process tomography},
  journal = {Physical Review A},
  volume = {87},
  number = {6},
  pages = {062119},
  year = {2013},
  doi = {10.1103/PhysRevA.87.062119}
}

@article{Greenbaum2015,
  author = {Greenbaum, Daniel},
  title = {Introduction to Quantum Gate Set Tomography},
  journal = {arXiv preprint arXiv:1509.02921},
  year = {2015}
}

@article{Busch2003,
  title = {Quantum States and Generalized Observables: A Simple Proof of Gleason's Theorem},
  author = {Busch, P.},
  journal = {Phys. Rev. Lett.},
  volume = {91},
  issue = {12},
  pages = {120403},
  numpages = {4},
  year = {2003},
  month = {Sep},
  publisher = {American Physical Society},
  doi = {10.1103/PhysRevLett.91.120403},
  url = {https://link.aps.org/doi/10.1103/PhysRevLett.91.120403}
}

@article{Schmid2021,
  title = {Characterization of Noncontextuality in the Framework of Generalized Probabilistic Theories},
  author = {Schmid, David and Selby, John H. and Wolfe, Elie and Kunjwal, Ravi and Spekkens, Robert W.},
  journal = {PRX Quantum},
  volume = {2},
  issue = {1},
  pages = {010331},
  numpages = {12},
  year = {2021},
  month = {Feb},
  publisher = {American Physical Society},
  doi = {10.1103/PRXQuantum.2.010331},
  url = {https://link.aps.org/doi/10.1103/PRXQuantum.2.010331}
}

@article{Schmid2024,
  doi = {10.22331/q-2024-03-14-1283},
  url = {https://doi.org/10.22331/q-2024-03-14-1283},
  title = {A structure theorem for generalized-noncontextual ontological models},
  author = {Schmid, David and Selby, John H. and Pusey, Matthew F. and Spekkens, Robert W.},
  journal = {{Quantum}},
  issn = {2521-327X},
  publisher = {{Verein zur F{\"{o}}rderung des Open Access Publizierens in den Quantenwissenschaften}},
  volume = {8},
  pages = {1283},
  month = mar,
  year = {2024}
}

@article{accessiblefrags,
  title = {Accessible fragments of generalized probabilistic theories, cone equivalence, and applications to witnessing nonclassicality},
  author = {Selby, John H. and Schmid, David and Wolfe, Elie and Sainz, Ana Bel\'en and Kunjwal, Ravi and Spekkens, Robert W.},
  journal = {Phys. Rev. A},
  volume = {107},
  issue = {6},
  pages = {062203},
  numpages = {21},
  year = {2023},
  month = {Jun},
  publisher = {American Physical Society},
  doi = {10.1103/PhysRevA.107.062203},
  url = {https://link.aps.org/doi/10.1103/PhysRevA.107.062203}
}

@article{Christensen2013,
  title = {Detection-Loophole-Free Test of Quantum Nonlocality, and Applications},
  author = {Christensen, B. G. and McCusker, K. T. and Altepeter, J. B. and Calkins, B. and Gerrits, T. and Lita, A. E. and Miller, A. and Shalm, L. K. and Zhang, Y. and Nam, S. W. and Brunner, N. and Lim, C. C. W. and Gisin, N. and Kwiat, P. G.},
  journal = {Phys. Rev. Lett.},
  volume = {111},
  issue = {13},
  pages = {130406},
  numpages = {5},
  year = {2013},
  month = {Sep},
  publisher = {American Physical Society},
  doi = {10.1103/PhysRevLett.111.130406},
  url = {https://link.aps.org/doi/10.1103/PhysRevLett.111.130406}
}

@article{Giustina2015,
  title = {Significant-Loophole-Free Test of Bell's Theorem with Entangled Photons},
  author = {Giustina, Marissa and Versteegh, Marijn A. M. and Wengerowsky, S\"oren and Handsteiner, Johannes and Hochrainer, Armin and Phelan, Kevin and Steinlechner, Fabian and Kofler, Johannes and Larsson, Jan-\AA{}ke and Abell\'an, Carlos and Amaya, Waldimar and Pruneri, Valerio and Mitchell, Morgan W. and Beyer, J\"orn and Gerrits, Thomas and Lita, Adriana E. and Shalm, Lynden K. and Nam, Sae Woo and Scheidl, Thomas and Ursin, Rupert and Wittmann, Bernhard and Zeilinger, Anton},
  journal = {Phys. Rev. Lett.},
  volume = {115},
  issue = {25},
  pages = {250401},
  numpages = {7},
  year = {2015},
  month = {Dec},
  publisher = {American Physical Society},
  doi = {10.1103/PhysRevLett.115.250401},
  url = {https://link.aps.org/doi/10.1103/PhysRevLett.115.250401}
}

@article{Shalm2015,
  title = {Strong Loophole-Free Test of Local Realism},
  author = {Shalm, Lynden K. and Meyer-Scott, Evan and Christensen, Bradley G. and Bierhorst, Peter and Wayne, Michael A. and Stevens, Martin J. and Gerrits, Thomas and Glancy, Scott and Hamel, Deny R. and Allman, Michael S. and Coakley, Kevin J. and Dyer, Shellee D. and Hodge, Carson and Lita, Adriana E. and Verma, Varun B. and Lambrocco, Camilla and Tortorici, Edward and Migdall, Alan L. and Zhang, Yanbao and Kumor, Daniel R. and Farr, William H. and Marsili, Francesco and Shaw, Matthew D. and Stern, Jeffrey A. and Abell\'an, Carlos and Amaya, Waldimar and Pruneri, Valerio and Jennewein, Thomas and Mitchell, Morgan W. and Kwiat, Paul G. and Bienfang, Joshua C. and Mirin, Richard P. and Knill, Emanuel and Nam, Sae Woo},
  journal = {Phys. Rev. Lett.},
  volume = {115},
  issue = {25},
  pages = {250402},
  numpages = {10},
  year = {2015},
  month = {Dec},
  publisher = {American Physical Society},
  doi = {10.1103/PhysRevLett.115.250402},
  url = {https://link.aps.org/doi/10.1103/PhysRevLett.115.250402}
}

@article{Kim2006,
  title = {Phase-stable source of polarization-entangled photons using a polarization Sagnac interferometer},
  author = {Kim, Taehyun and Fiorentino, Marco and Wong, Franco N. C.},
  journal = {Phys. Rev. A},
  volume = {73},
  issue = {1},
  pages = {012316},
  numpages = {5},
  year = {2006},
  month = {Jan},
  publisher = {American Physical Society},
  doi = {10.1103/PhysRevA.73.012316},
  url = {https://link.aps.org/doi/10.1103/PhysRevA.73.012316}
}

@article{Acin2006,
  title = {Grothendieck's constant and local models for noisy entangled quantum states},
  author = {Ac\'{\i}n, Antonio and Gisin, Nicolas and Toner, Benjamin},
  journal = {Phys. Rev. A},
  volume = {73},
  issue = {6},
  pages = {062105},
  numpages = {5},
  year = {2006},
  month = {Jun},
  publisher = {American Physical Society},
  doi = {10.1103/PhysRevA.73.062105},
  url = {https://link.aps.org/doi/10.1103/PhysRevA.73.062105}
}

@article{Proctor2017,
  title = {What Randomized Benchmarking Actually Measures},
  author = {Proctor, Timothy and Rudinger, Kenneth and Young, Kevin and Sarovar, Mohan and Blume-Kohout, Robin},
  journal = {Phys. Rev. Lett.},
  volume = {119},
  issue = {13},
  pages = {130502},
  numpages = {6},
  year = {2017},
  month = {Sep},
  publisher = {American Physical Society},
  doi = {10.1103/PhysRevLett.119.130502},
  url = {https://link.aps.org/doi/10.1103/PhysRevLett.119.130502}
}

@article{schmid2024a,
      title={Noncontextuality inequalities for prepare-transform-measure scenarios}, 
      author={David Schmid and Roberto D. Baldijão and John H. Selby and Ana Belén Sainz and Robert W. Spekkens},
      year={2024},
      journal={arXiv preprint},
      eprint={2407.09624},
      archivePrefix={arXiv},
      primaryClass={quant-ph},
      url={https://arxiv.org/abs/2407.09624} 
}

@article{Renner2024,
  title = {Compatibility of Generalized Noisy Qubit Measurements},
  author = {Renner, Martin J.},
  journal = {Phys. Rev. Lett.},
  volume = {132},
  issue = {25},
  pages = {250202},
  numpages = {8},
  year = {2024},
  month = {Jun},
  publisher = {American Physical Society},
  doi = {10.1103/PhysRevLett.132.250202},
  url = {https://link.aps.org/doi/10.1103/PhysRevLett.132.250202}
}

@Inbook{Kochen1975,
author="Kochen, Simon
and Specker, E. P.",
editor="Hooker, C. A.",
title="The Problem of Hidden Variables in Quantum Mechanics",
bookTitle="The Logico-Algebraic Approach to Quantum Mechanics: Volume I: Historical Evolution",
year="1975",
publisher="Springer Netherlands",
address="Dordrecht",
pages="293--328",
isbn="978-94-010-1795-4",
doi="10.1007/978-94-010-1795-4_17",
url="https://doi.org/10.1007/978-94-010-1795-4_17"
}

@article{Budroni2022,
  title = {Kochen-Specker contextuality},
  author = {Budroni, Costantino and Cabello, Ad\'an and G\"uhne, Otfried and Kleinmann, Matthias and Larsson, Jan-\AA{}ke},
  journal = {Rev. Mod. Phys.},
  volume = {94},
  issue = {4},
  pages = {045007},
  numpages = {62},
  year = {2022},
  month = {Dec},
  publisher = {American Physical Society},
  doi = {10.1103/RevModPhys.94.045007},
  url = {https://link.aps.org/doi/10.1103/RevModPhys.94.045007}
}

@article{Spekkens2008,
  title = {Negativity and Contextuality are Equivalent Notions of Nonclassicality},
  author = {Spekkens, Robert W.},
  journal = {Phys. Rev. Lett.},
  volume = {101},
  issue = {2},
  pages = {020401},
  numpages = {4},
  year = {2008},
  month = {Jul},
  publisher = {American Physical Society},
  doi = {10.1103/PhysRevLett.101.020401},
  url = {https://link.aps.org/doi/10.1103/PhysRevLett.101.020401}
}

@article{Shahandeh2021,
  title = {Contextuality of General Probabilistic Theories},
  author = {Shahandeh, Farid},
  journal = {PRX Quantum},
  volume = {2},
  issue = {1},
  pages = {010330},
  numpages = {15},
  year = {2021},
  month = {Feb},
  publisher = {American Physical Society},
  doi = {10.1103/PRXQuantum.2.010330},
  url = {https://link.aps.org/doi/10.1103/PRXQuantum.2.010330}
}

@article{Wrigh2023,
  title = {Invertible Map between {Bell} Nonlocal and Contextuality Scenarios},
  author = {Wright, Victoria J. and Farkas, M\'at\'e},
  journal = {Phys. Rev. Lett.},
  volume = {131},
  issue = {22},
  pages = {220202},
  numpages = {6},
  year = {2023},
  month = {Nov},
  publisher = {American Physical Society},
  doi = {10.1103/PhysRevLett.131.220202},
  url = {https://link.aps.org/doi/10.1103/PhysRevLett.131.220202}
}

@article{Plavala2024,
  title = {Contextuality as a Precondition for Quantum Entanglement},
  author = {Pl\'avala, Martin and G\"uhne, Otfried},
  journal = {Phys. Rev. Lett.},
  volume = {132},
  issue = {10},
  pages = {100201},
  numpages = {7},
  year = {2024},
  month = {Mar},
  publisher = {American Physical Society},
  doi = {10.1103/PhysRevLett.132.100201},
  url = {https://link.aps.org/doi/10.1103/PhysRevLett.132.100201}
}

@article{zhang2024,
  title = {Exact Steering Bound for Two-Qubit Werner States},
  author = {Zhang, Yujie and Chitambar, Eric},
  journal = {Phys. Rev. Lett.},
  volume = {132},
  issue = {25},
  pages = {250201},
  numpages = {7},
  year = {2024},
  month = {Jun},
  publisher = {American Physical Society},
  doi = {10.1103/PhysRevLett.132.250201},
  url = {https://link.aps.org/doi/10.1103/PhysRevLett.132.250201}
}

@article{Mazurek2016,
author={Mazurek, Michael D.
and Pusey, Matthew F.
and Kunjwal, Ravi
and Resch, Kevin J.
and Spekkens, Robert W.},
title={An experimental test of noncontextuality without unphysical idealizations},
journal={Nature Communications},
year={2016},
month={Jun},
day={13},
volume={7},
number={1},
pages={ncomms11780},
issn={2041-1723},
doi={10.1038/ncomms11780},
url={https://doi.org/10.1038/ncomms11780}
}

@article{barrett2006,
  title = {Information processing in generalized probabilistic theories},
  author = {Barrett, Jonathan},
  journal = {Phys. Rev. A},
  volume = {75},
  issue = {3},
  pages = {032304},
  numpages = {21},
  year = {2007},
  month = {Mar},
  publisher = {American Physical Society},
  doi = {10.1103/PhysRevA.75.032304},
  url = {https://link.aps.org/doi/10.1103/PhysRevA.75.032304}
}

@misc{hardy2001quantumtheoryreasonableaxioms,
      title={Quantum Theory From Five Reasonable Axioms}, 
      author={Lucien Hardy},
      year={2001},
      eprint={quant-ph/0101012},
      archivePrefix={arXiv},
      primaryClass={quant-ph},
      url={https://arxiv.org/abs/quant-ph/0101012}, 
}

@article{zhang2024parallel,
  title = {Reassessing the Boundary between Classical and Nonclassical for Individual Quantum Processes},
  author = {Zhang, Yujie and Schmid, David and Y\ifmmode \bar{\imath}\else \={\i}\fi{}ng, Y\`{\i}l\`e and Spekkens, Robert W.},
  journal = {Phys. Rev. X},
  volume = {16},
  issue = {2},
  pages = {021050},
  numpages = {54},
  year = {2026},
  month = {Jun},
  publisher = {American Physical Society},
  doi = {10.1103/vqfz-wzjg},
  url = {https://link.aps.org/doi/10.1103/vqfz-wzjg}
}

@article{zhang2025tech,
  title={Quantifiers and witnesses for the nonclassicality of measurements and of states},
  author={Zhang, Yujie and Y{\=\i}ng, Y{\`\i}l{\`e} and Schmid, David},
  journal={arXiv preprint arXiv:2504.02944},
  year={2025}
}

@article{Gisin1992,
title = {Maximal violation of {Bell}'s inequality for arbitrarily large spin},
journal = {Physics Letters A},
volume = {162},
number = {1},
pages = {15-17},
year = {1992},
issn = {0375-9601},
doi = {https://doi.org/10.1016/0375-9601(92)90949-M},
url = {https://www.sciencedirect.com/science/article/pii/037596019290949M},
author = {N. Gisin and A. Peres},
}

@article{Gisin1991,
  author  = {Gisin, N.},
  title   = {Bell's inequality holds for all non-product states},
  journal = {Physics Letters A},
  volume  = {154},
  number  = {5},
  pages   = {201--202},
  year    = {1991},
  issn    = {0375-9601},
  doi     = {10.1016/0375-9601(91)90805-I},
  url     = {https://www.sciencedirect.com/science/article/pii/037596019190805I}
}

@article{Barrett2002,
  title = {Nonsequential positive-operator-valued measurements on entangled mixed states do not always violate a {Bell} inequality},
  author = {Barrett, Jonathan},
  journal = {Phys. Rev. A},
  volume = {65},
  issue = {4},
  pages = {042302},
  numpages = {4},
  year = {2002},
  month = {Mar},
  publisher = {American Physical Society},
  doi = {10.1103/PhysRevA.65.042302},
  url = {https://link.aps.org/doi/10.1103/PhysRevA.65.042302}
}

@article{Branciard2013,
  title = {Measurement-Device-Independent Entanglement Witnesses for All Entangled Quantum States},
  author = {Branciard, Cyril and Rosset, Denis and Liang, Yeong-Cherng and Gisin, Nicolas},
  journal = {Phys. Rev. Lett.},
  volume = {110},
  issue = {6},
  pages = {060405},
  numpages = {5},
  year = {2013},
  month = {Feb},
  publisher = {American Physical Society},
  doi = {10.1103/PhysRevLett.110.060405},
  url = {https://link.aps.org/doi/10.1103/PhysRevLett.110.060405}
}

@article{Buscemi2012,
  title = {All Entangled Quantum States Are Nonlocal},
  author = {Buscemi, Francesco},
  journal = {Phys. Rev. Lett.},
  volume = {108},
  issue = {20},
  pages = {200401},
  numpages = {5},
  year = {2012},
  month = {May},
  publisher = {American Physical Society},
  doi = {10.1103/PhysRevLett.108.200401},
  url = {https://link.aps.org/doi/10.1103/PhysRevLett.108.200401}
}

@article{Designolle2023,
  title = {Improved local models and new {Bell} inequalities via Frank-Wolfe algorithms},
  author = {Designolle, S\'ebastien and Iommazzo, Gabriele and Besan\ifmmode \mbox{\c{c}}\else \c{c}\fi{}on, Mathieu and Knebel, Sebastian and Gel\ss{}, Patrick and Pokutta, Sebastian},
  journal = {Phys. Rev. Res.},
  volume = {5},
  issue = {4},
  pages = {043059},
  numpages = {6},
  year = {2023},
  month = {Oct},
  publisher = {American Physical Society},
  doi = {10.1103/PhysRevResearch.5.043059},
  url = {https://link.aps.org/doi/10.1103/PhysRevResearch.5.043059}
}

@misc{nobel2022,
  author       = "{The Royal Swedish Academy of Sciences}",
  title        = "{Scientific Background on the Nobel Prize in Physics 2022: Experiments with Entangled Photons, Establishing the Violation of {Bell} Inequalities and Pioneering Quantum Information Science}",
  year         = 2022,
  howpublished = "\url{https://www.nobelprize.org/uploads/2022/10/advanced-physicsprize2022.pdf}",
  note         = "[Accessed: November 5, 2024]"
}

@article{Clauser1969,
  title = {Proposed Experiment to Test Local Hidden-Variable Theories},
  author = {Clauser, John F. and Horne, Michael A. and Shimony, Abner and Holt, Richard A.},
  journal = {Phys. Rev. Lett.},
  volume = {23},
  issue = {15},
  pages = {880--884},
  numpages = {0},
  year = {1969},
  month = {Oct},
  publisher = {American Physical Society},
  doi = {10.1103/PhysRevLett.23.880},
  url = {https://link.aps.org/doi/10.1103/PhysRevLett.23.880}
}

@article{Lin2018,
  title = {Device-independent point estimation from finite data and its application to device-independent property estimation},
  author = {Lin, Pei-Sheng and Rosset, Denis and Zhang, Yanbao and Bancal, Jean-Daniel and Liang, Yeong-Cherng},
  journal = {Phys. Rev. A},
  volume = {97},
  issue = {3},
  pages = {032309},
  numpages = {15},
  year = {2018},
  month = {Mar},
  publisher = {American Physical Society},
  doi = {10.1103/PhysRevA.97.032309},
  url = {https://link.aps.org/doi/10.1103/PhysRevA.97.032309}
}

@article{Krystek2007,
doi = {10.1088/0957-0233/18/11/025},
url = {https://dx.doi.org/10.1088/0957-0233/18/11/025},
year = {2007},
month = {sep},
publisher = {},
volume = {18},
number = {11},
pages = {3438},
author = {Krystek, Michael and Anton, Mathias},
title = {A weighted total least-squares algorithm for fitting a straight line},
journal = {Measurement Science and Technology}
}

@article{Grabowecky2022,
  title = {Experimentally bounding deviations from quantum theory for a photonic three-level system using theory-agnostic tomography},
  author = {Grabowecky, Michael J. and Pollack, Christopher A. J. and Cameron, Andrew R. and Spekkens, Robert W. and Resch, Kevin J.},
  journal = {Phys. Rev. A},
  volume = {105},
  issue = {3},
  pages = {032204},
  numpages = {16},
  year = {2022},
  month = {Mar},
  publisher = {American Physical Society},
  doi = {10.1103/PhysRevA.105.032204},
  url = {https://link.aps.org/doi/10.1103/PhysRevA.105.032204}
}

@article{bowles2018device,
	title        = {Device-independent entanglement certification of all entangled states},
	author       = {Bowles, Joseph and {\v{S}}upi{\'c}, Ivan and Cavalcanti, Daniel and Ac{\'\i}n, Antonio},
	year         = 2018,
	journal      = {Phys. Rev. Lett.},
	publisher    = {APS},
	volume       = 121,
	number       = 18,
	pages        = 180503,
	url          = {https://journals.aps.org/prl/abstract/10.1103/PhysRevLett.121.180503}
}

@article{Mazurek2021,
  title = {Experimentally Bounding Deviations From Quantum Theory in the Landscape of Generalized Probabilistic Theories},
  author = {Mazurek, Michael D. and Pusey, Matthew F. and Resch, Kevin J. and Spekkens, Robert W.},
  journal = {PRX Quantum},
  volume = {2},
  issue = {2},
  pages = {020302},
  numpages = {33},
  year = {2021},
  month = {Apr},
  publisher = {American Physical Society},
  doi = {10.1103/PRXQuantum.2.020302},
  url = {https://link.aps.org/doi/10.1103/PRXQuantum.2.020302}
}

@software{zhang2025github,
  author       = {Zhang,Yujie},
  title        = {{Bipartite‑nonclassicality: Data and code}},
  year         = {2025},
  version      = {1932f03},
  url          = {https://github.com/yujie4phy/Bipartite-nonclassicality},
  note         = {{GitHub} repository. commit 1932f03, accessed 12, Apr,2025}
}

@article{Saunders2010,
  title={Experimental {EPR}-steering using {Bell}-local states},
  author={Saunders, Dylan John and Jones, Steve J and Wiseman, Howard M and Pryde, Geoff J},
  journal={Nature Physics},
  volume={6},
  number={11},
  pages={845--849},
  year={2010},
  publisher={Nature Publishing Group UK London},
url={https://doi.org/10.1038/nphys1766}
}

@article{Nielsen2021,
  doi = {10.22331/q-2021-10-05-557},
  url = {https://doi.org/10.22331/q-2021-10-05-557},
  title = {Gate {S}et {T}omography},
  author = {Nielsen, Erik and Gamble, John King and Rudinger, Kenneth and Scholten, Travis and Young, Kevin and Blume-Kohout, Robin},
  journal = {{Quantum}},
  issn = {2521-327X},
  publisher = {{Verein zur F{\"{o}}rderung des Open Access Publizierens in den Quantenwissenschaften}},
  volume = {5},
  pages = {557},
  month = oct,
  year = {2021}
}

@article{Schmid2023understanding,
  doi = {10.22331/q-2023-12-04-1194},
  url = {https://doi.org/10.22331/q-2023-12-04-1194},
  title = {Understanding the interplay of entanglement and nonlocality: motivating and developing a new branch of entanglement theory},
  author = {Schmid, David and Fraser, Thomas C. and Kunjwal, Ravi and Sainz, Ana Belen and Wolfe, Elie and Spekkens, Robert W.},
  journal = {{Quantum}},
  issn = {2521-327X},
  publisher = {{Verein zur F{\"{o}}rderung des Open Access Publizierens in den Quantenwissenschaften}},
  volume = {7},
  pages = {1194},
  month = dec,
  year = {2023}
}

@article{Rosset2020,
  title = {Type-Independent Characterization of Spacelike Separated Resources},
  author = {Rosset, Denis and Schmid, David and Buscemi, Francesco},
  journal = {Phys. Rev. Lett.},
  volume = {125},
  issue = {21},
  pages = {210402},
  numpages = {7},
  year = {2020},
  month = {Nov},
  publisher = {American Physical Society},
  doi = {10.1103/PhysRevLett.125.210402},
  url = {https://link.aps.org/doi/10.1103/PhysRevLett.125.210402}
}

@article{Bennet2012,
  title = {Arbitrarily Loss-Tolerant Einstein-Podolsky-Rosen Steering Allowing a Demonstration over 1 km of Optical Fiber with No Detection Loophole},
  author = {Bennet, A. J. and Evans, D. A. and Saunders, D. J. and Branciard, C. and Cavalcanti, E. G. and Wiseman, H. M. and Pryde, G. J.},
  journal = {Phys. Rev. X},
  volume = {2},
  issue = {3},
  pages = {031003},
  numpages = {12},
  year = {2012},
  month = {Jul},
  publisher = {American Physical Society},
  doi = {10.1103/PhysRevX.2.031003},
  url = {https://link.aps.org/doi/10.1103/PhysRevX.2.031003}
}

@article{Larsson2014,
doi = {10.1088/1751-8113/47/42/424003},
url = {https://doi.org/10.1088/1751-8113/47/42/424003},
year = {2014},
month = {oct},
publisher = {IOP Publishing},
volume = {47},
number = {42},
pages = {424003},
author = {Larsson, Jan-{\AA}ke},
title = {Loopholes in Bell inequality tests of local realism},
journal = {Journal of Physics A: Mathematical and Theoretical}
}

@article{Cavalcanti2009,
  title = {Experimental criteria for steering and the Einstein-Podolsky-Rosen paradox},
  author = {Cavalcanti, E. G. and Jones, S. J. and Wiseman, H. M. and Reid, M. D.},
  journal = {Phys. Rev. A},
  volume = {80},
  issue = {3},
  pages = {032112},
  numpages = {16},
  year = {2009},
  month = {Sep},
  publisher = {American Physical Society},
  doi = {10.1103/PhysRevA.80.032112},
  url = {https://link.aps.org/doi/10.1103/PhysRevA.80.032112}
}

@article{Gisin1999,
title = {A local hidden variable model of quantum correlation exploiting the detection loophole},
journal = {Physics Letters A},
volume = {260},
number = {5},
pages = {323-327},
year = {1999},
issn = {0375-9601},
doi = {https://doi.org/10.1016/S0375-9601(99)00519-8},
url = {https://www.sciencedirect.com/science/article/pii/S0375960199005198},
author = {N. Gisin and B. Gisin}
}

@article{Marco2015,
  title = {Inequivalence of entanglement, steering, and {Bell} nonlocality for general measurements},
  author = {Quintino, Marco T\'ulio and V\'ertesi, Tam\'as and Cavalcanti, Daniel and Augusiak, Remigiusz and Demianowicz, Maciej and Ac\'{\i}n, Antonio and Brunner, Nicolas},
  journal = {Phys. Rev. A},
  volume = {92},
  issue = {3},
  pages = {032107},
  numpages = {6},
  year = {2015},
  month = {Sep},
  publisher = {American Physical Society},
  doi = {10.1103/PhysRevA.92.032107},
  url = {https://link.aps.org/doi/10.1103/PhysRevA.92.032107}
}

@article{Massar2003,
  title = {Violation of local realism versus detection efficiency},
  author = {Massar, Serge and Pironio, Stefano},
  journal = {Phys. Rev. A},
  volume = {68},
  issue = {6},
  pages = {062109},
  numpages = {7},
  year = {2003},
  month = {Dec},
  publisher = {American Physical Society},
  doi = {10.1103/PhysRevA.68.062109},
  url = {https://link.aps.org/doi/10.1103/PhysRevA.68.062109}
}

@article{Aguilar2024,
author={Villegas-Aguilar, Luis
and Polino, Emanuele
and Ghafari, Farzad
and Quintino, Marco T{\'u}lio
and Laverick, Kiarn T.
and Berkman, Ian R.
and Rogge, Sven
and Shalm, Lynden K.
and Tischler, Nora
and Cavalcanti, Eric G.
and Slussarenko, Sergei
and Pryde, Geoff J.},
title={Nonlocality activation in a photonic quantum network},
journal={Nature Communications},
year={2024},
month={Apr},
day={10},
volume={15},
number={1},
pages={3112},
issn={2041-1723},
doi={10.1038/s41467-024-47354-w},
url={https://doi.org/10.1038/s41467-024-47354-w}
}

\end{document}